\documentclass[12pt,onecolumn,english]{IEEEtran}
\IEEEoverridecommandlockouts    
\usepackage{epsf}
\usepackage{epsfig}
\usepackage{times}
\usepackage{tablists}
\usepackage{amsmath}
\usepackage{amsfonts}
\usepackage{latexsym}
\usepackage{url}
\usepackage{amsthm}
\usepackage{amssymb}
\usepackage{stmaryrd}
\usepackage{tikz}
\usepackage{dblfloatfix} 
\usepackage{pgfplots}
\usepackage{subcaption}
\usepackage{enumerate}
\usepackage{blkarray}
\usepackage{hyperref}
\hypersetup{
    colorlinks,
    linktocpage,
    linkcolor=blue,
}

\usetikzlibrary{calc,shadings,patterns}
\usepackage{romannum}

\newcommand{\markov}{\mathrel\multimap\joinrel\mathrel-\mspace{-9mu}\joinrel\mathrel-}
\newtheorem{theorem}{Theorem}
\newtheorem{definition}{Definition}
\newtheorem{lemma}{Lemma}
\newtheorem{corollary}{Corollary}
\newtheorem{proposition}{Proposition}
\newtheorem{example}{Example} 
\newtheorem{remark}{Remark}

\title{
The $K$-User DM Broadcast Channel with Two Groupcast Messages: 
Achievable Rate Regions and the Combination Network as 
a Case Study}

\author{Mohamed Salman and Mahesh K. Varanasi \\
\thanks{This work was presented in part at the 2018 IEEE International Symposium on Information Theory, Vail, CO \cite{salman2018achievable}. It was funded in part by the 2018 and 2019 Qualcomm Faculty Awards.
The authors are with the Electrical, Computer and Energy Engineering Department, University of Colorado, Boulder, CO, USA (emails: \{mohamed.salman, varanasi\}@colorado.edu).}}

\begin{document}

\maketitle
\thispagestyle{empty}
\pagestyle{headings}
\pagenumbering{arabic}

\begin{abstract}
A novel class of achievable rate regions is obtained for the general $K$-receiver discrete memoryless broadcast channel over which two groupcast messages are to be transmitted, with each message required by an arbitrary group of receivers.  
The associated achievability schemes are parameterized by an expansion of the message set which then determines how 
random coding techniques are employed. These techniques include generalized versions of {\em up-set} message-splitting, the generation of possibly multiple auxiliary codebooks for certain compositions of split messages using superposition coding with subset inclusion order, partial interference decoding at all receivers in general, joint unique decoding at receivers that desire both messages, and non-unique or indirect decoding at receivers that desire only one of the two messages.

The generality of the proposed class of schemes implies new achievable rate regions for problems previously not considered as well as those that were studied before, with specific members of that class having rate regions that coincide with previously found capacity regions for special classes of broadcast channels with two private or two {\em nested} groupcast messages, wherein the group of receivers desiring one message is contained in that desiring the other. Moreover, new capacity results are established for certain partially ordered classes of general broadcast channels for certain non-nested groupcast messages such as two messages each desired at $K-1$ receivers and two messages with one desired at $K-1$ receivers and the other message desired at some arbitrary subset of receivers including the receiver at which the first message is not desired. To further show the strength of the proposed class of achievable rate regions we consider the so-called {\em combination network}, a 
special class of linear deterministic broadcast channels, as a test case. When specialized to the combination network, some members of the class of inner bounds 
are shown, via converse results, to result in the capacity region for three different scenarios, namely, (a) the two messages are intended for two distinct sets of $K{-}1$ receivers each and (b) two nested messages in which one message is intended for one or (c) two (common) receivers and both messages are intended for all other (private) receivers. In the latter two nested messages cases, we hence recover, in a top-down manner, previous results by Bidokhti, Prabhakaran and Diggavi obtained therein using network coding schemes based on rate-splitting and linear superposition coding tailored to the combination network, while in the first case we obtain a new capacity result for a non-nested message set, which was hitherto not studied. Furthermore, we show, via a suitable choice of the distribution of the auxiliary random variables and the encoding function in our framework, the achievability of rate pairs in two interesting examples of combination networks, with three and four common receivers each. These examples were used in the previous literature to show the sub-optimality of the rate-splitting and linear superposition coding mentioned earlier, and hence to motivate the additional consideration of a pre-encoding technique and a block-Markov linear superposition coding for the combination network, with the latter then lifted to the general broadcast channel. Our results suggest that the proposed framework here
for the general broadcast channel when specialized to the combination network is strong enough to incorporate the enhancements afforded by those two latter techniques, thereby suggesting among other things, that perhaps block-Markov superposition coding is not necessary in the general broadcast channel. Moreover, there is a trade-off between the complexity of the coding scheme within the class of schemes we propose when applied to the combination network and that of the determination of the distribution of the auxiliary random variables and the encoding function that achieve the capacity region. This may have interesting implications for the general broadcast channel as well.  

\end{abstract}

\newpage

\section{Introduction}
\label{Sec_Into}

The problem of sending two groupcast messages over the $K$-receiver broadcast channel (BC) is studied. Each such message is intended for a distinct group of receivers, with the two groups of receivers assumed to be arbitrary in general.
In spite of its apparent simplicity, this problem remains unsolved in general in the Shannon-theoretic sense. However, capacity results have been obtained in the literature for various special cases of the BC and for specific (especially degraded) message sets in the two and three-receiver cases, and more recently, in the general $K$-receiver case.

The most studied problem however of sending two messages over the BC is the two-receiver discrete memoryless (DM) case with {\em private} messages. The capacity region is notoriously difficult in general even in this case and remains unsolved in general to date. However, for the increasingly larger classes of degraded \cite{bergmans1973random}, less noisy \cite[Definition 2]{korner1975source} and more capable \cite[Definition 3]{korner1975source} channels, the capacity region was found in the series of seminal papers \cite{cover1972broadcast,bergmans1973random,gallager1974capacity,korner1977images,el1979capacity} in the 1970s. In particular, the superposition coding scheme proposed in \cite{bergmans1973random} was shown, using a clever identification of auxiliary random variable, to achieve the capacity region in \cite{gallager1974capacity} for the degraded BC. The same scheme was also shown to achieve the capacity region for the larger class of less noisy and more capable BCs in \cite{korner1975source} and \cite{el1979capacity}, where the images-of-a-set technique \cite{korner1977images} and the Csiszar sum lemma \cite[Lemma 7]{csiszar1978broadcast} were used to prove the converses, respectively.




The capacity region for the two-receiver DM BC with two nested (i.e., degraded) messages was found by Korner and Marton in 1977 \cite{kiirner1977general}. Interestingly, with superposition coding as the achievability scheme and a converse based on the images-of-a-set technique \cite{korner1977images}, the authors therein established the capacity region without {\em any} restriction on the channel. However, 
the generalization of this result for three or more receivers has remained elusive ever since.





In the $K$-receiver BC with two {\em nested} messages the receivers can be classified into $L$ common receivers that require only one (common) message and $P$ private receivers that require both messages (with $P+L=K$). The result of Korner and Marton in \cite{kiirner1977general} might suggest that the nested structure of the messages might render a straightforward extension of their superposition coding scheme to be capacity-optimal even in this $K$-receiver setting. However, the authors of \cite{nair2017sub} and \cite{nair2009capacity} showed that superposition coding alone is not optimal for the three-receiver DM BC with one and two common receivers, respectively.
In the latter case, they proposed a more general scheme that involves a simple form of rate-splitting along with superposition coding \cite{nair2009capacity}. However, even this scheme was only shown to achieve capacity for the restricted class of DM BCs wherein the private receiver is less noisy than one of the two common receivers. 

One of the challenges of obtaining capacity results for rate-splitting based schemes beyond the three-receiver case is the difficulty of obtaining a closed-form polyhedral description for the inner bound in terms of the message rates due to the large number of split rates possible.
We make progress on this problem in
\cite[Theorem 2]{salman2017capacity,salman2019capacity} where an achievable rate region that generalizes in one direction the capacity result for the three-user, two-common receiver problem in \cite{nair2009capacity} to arbitrary $K$ and arbitrary $L$ is obtained. In particular, the private message is split into $L$ sub-messages, and each common receiver decodes the common message uniquely, and certain sub-messages of the private message assigned to it, indirectly. 
The inner bound is presented in terms of the two nested message rates only, by eliminating all split rates for any $K$ and any $L$, in general. Also, this inner bound is shown to be capacity-optimal for various classes of channels characterized by certain pair-wise relationships between and among the common and private receivers \cite[Theorem 3]{salman2019capacity}. For example, the scheme is optimal for the class of four-receiver DM BCs with $L{=}3$ and $P{=}1$ in which the private receiver is less noisy than two of the three common receivers.

Combination networks, first proposed in \cite{ngai2004network} to demonstrate that network coding can attain unbounded gain over routing alone, can be seen to be a special class of linear deterministic broadcast channels. The capacity regions of the combination network for the two- and three-receiver cases were established in \cite{grokop2008fundamental} under the guise of fundamental constraints in multicast capacity regions and where the transmitter must transmit all possible $2^K{-}1$ independent groupcast messages. The achievability scheme depends mainly on the rate transfer argument. For example, in a two-receiver combination network, we have three possible independent messages; two private messages and one common message. If the transmitter is able to simultaneously send a rate of 1 bit per channel use for each of the three messages, then by sending the same information in each of the two private messages, it must be able to send a common message at rate of 2 bits per channel use for both receivers. Another possible rate transfer operation is when the transmitter merely uses the common bit to send private information to one of the receivers. Then, the channel can deliver 2 bits of private message per channel use to that receiver and 1 bit to the other. 
  
In other words, the achievability of any $2^K{-}1$-dimensional rate-tuple of messages implies the achievability of a certain $2^K{-}1$-dimension rate region regardless of the channel. The approach of \cite{grokop2008fundamental} is to exhaustively determine all possibilities for rate transfer to characterize the inner bound for $K=2$, whereas for $K=3$, rate transfer and network coding are employed to establish the inner bound. On the other hand, the outer bound depends on cut-set bounds with some extensions. These proofs of the converse are specific to $K$ being two or three, and hard to extend to $K>3$. In fact, the capacity region of the general $K$-receiver combination network is an open problem to date for $K>3$.

Because it is unclear how to generalize the approach of \cite{grokop2008fundamental} to more that three receivers since the complexity of rate transfer increases exponentially with the number of users, Tian in \cite{tian2011latent}, under the guise of latent capacity regions, effectively considers a restricted class of {\em symmetric} $K$-receiver combination networks with the capacities of certain sets of finite capacity links being the same, and with symmetric message rates, wherein the messages required by the same number of receivers have the same rate. 
By simplifying the channel model and the message structure in this manner, Tian was able to establish the symmetric capacity region (where the rates of all messages of the same order are equal) of the symmetric $K$-receiver combination network by extending the rate-transfer approach of \cite{grokop2008fundamental} to this scenario.

Later, in \cite{salimi2015generalized}, Salimi {\em et al.} proposed a general framework for the outer bound of broadcast networks in which they obtain a large family of outer bounds based on the sub-modularity of entropy they call generalized cut-set bounds. These bounds are used to reproduce the outer bounds of \cite{grokop2008fundamental} for the two- and three-receiver combination networks and, along with an explicit polyhedral description, the symmetric capacity region of the $K$-receiver symmetric combination network of \cite{tian2011latent}.


Romero and Varanasi in \cite{romero2016superposition} obtain the capacity region for the combination networks via a top-down approach unlike in \cite{grokop2008fundamental} for $K\leq 3$ and in \cite{tian2011latent} for $K$-user symmetric combination networks. They first obtained a inner bound for the much more general $K$-user DM BC with general message sets \cite[Theorem 1]{romero2016superposition} and then specialized to the combination network to recover the results in \cite{grokop2008fundamental} and \cite{tian2011latent}. This bolsters the case for considering superposition coding and rate splitting (i.e., without binning) in the DM BC and specializing it to the combination network for achieving even the general (not just symmetric) capacity region of the general (not only symmetric) combination network for $K>3$. This is the approach we take on the combination network in this paper for the two groupcast message set.

Meanwhile, the authors in \cite{bidokhti2016capacity} studied two nested message set broadcasting over the $K$-receiver combination network where three different coding schemes tailored to it were proposed. The first is a linear superposition coding scheme with rate splitting for the private message where the transmitted signal is obtained by the multiplication of a carefully designed matrix over a finite field with the information symbols vector over that field. The structure of this matrix follows the zero-structured matrices of \cite[Definition 2]{bidokhti2016capacity} while the rank of this matrix dominates the feasibility of decoding analysis. This zero-structured matrix has a zero in specific positions such that, when it is multiplied with the information symbols, the received signal at each common receiver does not depend on too many private sub-messages. The second scheme is a linear superposition coding scheme with a pre-encoder, i.e., the information symbols vector is multiplied first by a pre-encoder matrix before it is multiplied with a zero-structured matrix. The purpose of the pre-encoder matrix is to introduce dependency among the sub-messages of the private message. The last coding scheme is a block Markov coding scheme. The main idea is, instead of introducing the dependency among the sub-messages of the private message by multiplying with a pre-encoder matrix over one-time (one-block) code, a block Markov coding scheme can be used to introduce those dependencies sequentially across blocks. Upon receiving all the $n$ blocks, each receiver finds its intended message by performing backward decoding \cite{willems1985discrete}. 

The second and third coding schemes of \cite{bidokhti2016capacity} are motivated via two examples, \cite[Example 2]{bidokhti2016noisy} and \cite[Example 4]{bidokhti2016noisy}. In \cite[Example 2]{bidokhti2016noisy}, it was shown that the second coding scheme (with pre-encoding) can achieve a rate pair that is not achievable by the first, while in \cite[Example 4]{bidokhti2016noisy} the block Markov coding scheme was shown to achieve a rate pair that was not achievable by the second coding scheme. The first rate splitting and linear superposition coding scheme, was shown to be capacity achieving for a general $K$-receiver combination network with at most two common receivers \cite[Proposition 1 and Theorem 3]{bidokhti2016noisy}. The second the third coding schemes were shown to be optimal for the general $K$-receiver combination networks with three common receivers \cite[Theorems 1, 2, 4 and 5]{bidokhti2016noisy}. Capacity for more than three common receivers remains an open problem.

\subsection{Main Contributions}

In this work, we begin by proposing an inner bound for the general $K$-user DM BC with two general groupcast messages in Section \ref{Sec_General_two}. That inner bound is based on a class of coding schemes parameterized by a flexible form of message set expansion, and involves general forms of rate-splitting, superposition coding, unique and non-unique decoding. Specific choices of the message set expansion parameter in the two, three and $K$ receiver DM BCs lead to recovering the previously proposed achievability schemes based on rate-splitting and superposition coding for two, three- and $K$-receiver DM BCs  \cite{bergmans1973random,kiirner1977general,el1979capacity,nair2009capacity,salman2019capacity}. 
Moreover, our inner bound also gives new and potentially strictly larger achievable rate regions than those proposed based on rate-splitting and superposition coding for two nested groupcast messages in the three- and $K$-receiver DM BCs of \cite{nair2009capacity} and \cite{salman2019capacity}, respectively. Furthermore, it gives classes of achievable rate regions for $K$-user DM BCs for any two non-nested groupcast message sets as well, none of which appear to have been proposed in the literature before for three or more receivers. Hence our inner bound provides a general and unifying framework for studying the problem of sending two groupcast messages over a $K$-receiver DM BC. It is in fact an adaptation of the rate-splitting and superposition coding framework recently proposed by Romero and Varanasi in \cite{romero2017rate} which is applicable to general message sets with any number (up to $2^K-1$) of groupcast messages. By specializing to the two-message set case, our aim here is to gain an in-depth understanding (which \cite{romero2017rate} does not do) \textemdash via explicit descriptions, converses, and choices of extremal coding distributions \textemdash of the strength of an adaptation of the general inner bound of \cite[Corollary 1]{romero2017rate} in the two-message set case. 

The adaptation of the framework of \cite[Corollary 1]{romero2017rate} to the two groupcast message case considered here involves replacing unique joint decoding in the coding scheme of \cite[Corollary 1]{romero2017rate} with non-unique joint decoding at all common receivers that require just one of the two groupcast messages. While the adoption of non-unique decoding may or may not strictly improve the achievable rate region in general (which is difficult to prove one way or the other) we choose it because (a) it expands the rate region per coding distribution (b) it does not shrink the rate region taken as a union over all admissible coding distributions (c) there are fewer inequalities leading to relative ease of doing projections via the Fourier-Motzkin Elimination (FME) technique \cite{el2011network} and (d) proving converses becomes easier as a result in some special cases as well. 
Indeed, for the case of two groupcast messages where (a) each message is desired at $K-1$ receivers, and (b) one message is desired at $K-1$ receivers and the other at some arbitrary subset of receivers that includes the one where the first message is not desired, we obtain the capacity region for classes of partially ordered DM BCs wherein certain pairs of receivers are ordered by the well-known less noisy and more capable ordering relations. These results complement the ones we obtained recently for two nested groupcast messages in \cite{salman2019capacity}. 

Beyond the fact that we propose a general and unifying framework 
which provides general descriptions of rate-splitting, superposition coding, unique and non-unique decoding, and a succinct characterization of the associated rate regions, 
the key novel aspect underlying our achievable schemes 
is the message set expansion. It affects how everything is done: how the messages are split, what messages are reconstructed and how they are superposed, which reconstructed messages are uniquely decoded and which are non-uniquely decoded by the private and common receivers, respectively. There are in general $2^{2^{K}-3}$ expanded message sets and each gives a distinct achievable rate region for a given two-message set and $K$. Three specific choices of message set expansions for the three-receiver case previously considered in \cite{nair2009capacity} with two nested messages with 2 common receivers recover the results therein while the rest of the other choices yield new rate region descriptions for that setting. More importantly, there are message set expansions that result in coding schemes that assign {\em multiple} codewords to certain groups of sub-messages (after message splitting) according to different distributions, a feature not found in any previously proposed coding schemes, 
to the best of our knowledge. To be sure, the largest message set expansion yields a rate region that subsumes all others but it also yields the most complex coding scheme of all. It is of interest to consider simpler coding schemes corresponding to smaller message set expansions when such schemes suffice to achieve capacity. 

To more explicitly show the strength of the proposed class of achievable rate regions we consider the combination network as a test case in Section \ref{Sec_Capacity_combianation}. In particular, when specialized to the combination network, the largest of the class of inner bounds, 
corresponding to the most complex of the class of coding schemes,
is shown, via converse results, to ``easily" result in the capacity region via the choice of a single distribution of the auxiliary random variables and encoding function for three different scenarios, namely, (a) the two messages are intended for two distinct sets of $K{-}1$ receivers each and (b) two nested messages in which one message is intended for one or (c) two (common) receivers and both messages are intended for all other (private) receivers. In the latter two nested messages cases, we hence recover, in a top-down manner, the previous result of \cite[Theorem 3]{bidokhti2016capacity} obtained therein using network coding schemes based on rate-splitting and linear superposition coding tailored to the combination network and two nested messages, while in the first case we obtain a new capacity result for a non-nested message set, which was hitherto unknown. Moreover, we also show that via the same and single choice of auxiliary random variables and choice of encoding function in our inner bound applied to the combination network we recover, in a top-down manner, the rate region of the rate-splitting and linear superposition coding given by \cite[Proposition 1]{bidokhti2016capacity} for any two nested messages set for the $K$-user combination network.

Furthermore, we show in Section \ref{Sec_Relation_Bidokhti}, again in a top-down manner, via suitable choices of the distribution of the auxiliary random variables and the encoding function in our framework, the achievability of ``difficult" rate pairs in two interesting examples of six- and seven-receiver combination networks studied in detail \cite[Examples 2-5]{bidokhti2016capacity}, with three and four common receivers, respectively. The first example was proposed therein to show the sub-optimality of the rate-splitting and linear superposition coding scheme mentioned earlier, and hence to motivate the additional consideration of a pre-encoding technique in that case \cite[Example 2]{bidokhti2016capacity}. The second example was proposed to show the sub-optimality in turn of the pre-encoding technique as well, and to further motivate a third block-Markov linear superposition coding scheme for the combination network. Our results hence suggest that the proposed framework here and the associated class of achievable rate regions for the general broadcast channel may be strong enough that, when it is specialized to the combination network, it would incorporate the enhancements afforded by pre-encoding with rate-splitting and linear superposition coding and block-Markov linear superposition coding schemes of \cite{bidokhti2016capacity}.

The authors in \cite{bidokhti2016capacity} lift their block-Markov linear superposition coding scheme for the combination network to a block-Markov superposition coding for the general broadcast channel because the former can achieve the ``difficult" rate pair of the 7-receiver combination network with four common receivers of \cite[Examples 4-5]{bidokhti2016capacity} whereas their pre-encoding technique with linear superposition coding does not. The fact, however, that that rate pair can be achieved via our proposed inner bound suggests that perhaps block-Markov superposition coding is not necessary in the general broadcast channel. 

Finally, in Section \ref{Sec:Smaller_F}, we show that the consideration of simpler coding schemes within the class we propose corresponding to suitably chosen smaller message set expansions suffice in each case to achieve capacity for the combination network for the three message set cases mentioned previously for which we have capacity results. However, finding the extremal distributions is more involved. This is what we do in detail in Section \ref{Sec:Smaller_F} where the allowed {\em dependencies} among the auxiliary random variables are exploited. 

Thus, the results in Sections \ref{Sec_Capacity_combianation} and \ref{Sec:Smaller_F} on the combination network seems to suggest that there is perhaps a trade-off between the complexity of coding scheme adopted within the class of coding schemes we propose and the ease of finding coding distributions to extract much (if not all) of the performance of that scheme in DM BCs in general.
In other words, even if we know that a union of regions inner bound is the capacity region it may still be worthwhile to consider more complex coding schemes in the hope of discovering coding distributions that achieve the capacity region in more explicit form that otherwise may prove difficult or impossible.

The rest of this paper is organized as follows. In Section \ref{Sec_system}, we describe the system model. In Section \ref{Sec_General_two}, we present the new achievable rate region for the DM BC with two general groupcast messages and the capacity region for certain partially ordered DM BCs for certain message sets including one with two messages each required by $K-1$ receivers. The inner bound is specified for the nested messages case in Section \ref{Sec_Two_nested}. Then, in Section \ref{Sec_Capacity_combianation}, we obtain the capacity regions for combination networks for three different message sets. In Section \ref{Sec_Relation_Bidokhti} we illustrate the relation between our capacity regions for combination networks and those proposed before in \cite{bidokhti2016capacity}. In Section \ref{Sec:Smaller_F}, a trade-off between complexity of coding scheme (via message set expansion) and choice of random coding distributions is studied. Finally, the paper is concluded in Section \ref{Sec_Conc}.

\section{System Model and Preliminaries }
\label{Sec_system}

\subsection{System Model}

\begin{figure*}[t]
\centering
\begin{tikzpicture}
\node (E) at (0,0) {Encoder};
\node (C) at (3,0.3) {$p(y_1y_2y_3|x)$};
\node at (3,1.8) {DM BC};
\node at (1.35,.2) {$x^n$};
\node (msg) at (-1.8,0) {$M_{\mathsf{E}}$};

\node (Es) at (-.9,0) {};
\node (Ee) at (.9,0) {};
\node (Cs) at (1.78,0) {};

\node at (4.6,1.65) {$y_1^n$};
\node (rec_msg1) at (7.45,1.4) {};
\node (rec_msg1P) at (7.8,1.4) {$\tilde{M}_{\mathsf{W}_1^{\mathsf{E}}}$};
\node (D1) at (6,1.4) {Decoder $1$};
\node (Ce) at (4.2,1.4) {};
\node (Ds) at (5.1,1.4) {};
\node (De) at (6.9,1.4) {};
\draw [->] (Ce) -- (Ds) ;
\draw [->] (De) -- (rec_msg1) ;
\draw (5,1) rectangle (7,1.8);

\node at (4.6,.7) {$y_2^n$};
\node (rec_msg2) at (7.8,0.4) {$\tilde{M}_{\mathsf{W}_2^{\mathsf{E}}}$};
\node (DP) at (6,0.4) {Decoder $2$};
\node (CeP) at (4.2,0.4) {};
\node (DsP) at (5.1,0.4) {};
\node (DeP) at (6.9,0.4) {};
\draw [->] (CeP) -- (DsP) ;
\draw [->] (DeP) -- (rec_msg2) ;
\draw (5,0.0) rectangle (7,0.8);

\node at (4.63,-.35) {$y_{3}^n$};
\node (rec_msg3) at (7.8,-0.6) {$\tilde{M}_{\mathsf{W}_3^{\mathsf{E}}}$};
\node (DP) at (6,-0.6) {Decoder $3$};
\node (CeP1) at (4.2,-0.6) {};
\node (DsP1) at (5.1,-0.6) {};
\node (DeP1) at (6.9,-0.6) {};
\draw [->] (CeP1) -- (DsP1) ;
\draw [->] (DeP1) -- (rec_msg3) ;
\draw (5,-1) rectangle (7,-.2);

\draw [->] (msg) -- (Es) ;
\draw [->] (Ee) -- (Cs) ;
\draw (-1,-0.4) rectangle (1,0.4);
\draw (1.75,-1) rectangle (4.25,1.6);
\end{tikzpicture}
\caption{The three-receiver DM BC where $M_{\mathsf{E}}=\{M_S: S \in \mathsf{E}\}$ are the messages sent through the channel and $\tilde{M}_{\mathsf{W}_i^{\mathsf{E}}}=\{ \tilde{m}_{S} : S \in \mathsf{W}_i^{\mathsf{E}} \}$ are the messages decoded by receiver $Y_i$.}
\label{Fig_SystemModekK_3}
\end{figure*}
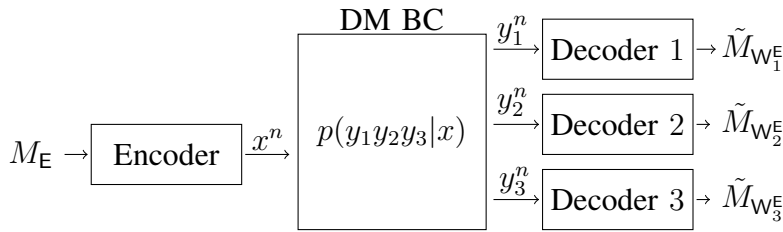

The DM BC consists of a transmitter $X\in \mathcal{X}$, $K$ receivers $Y_i\in \mathcal{Y}_k$, for $1 \leq k \leq K$, and the channel transition probability $W(y_1, \cdots , y_K|x)$ with the conditional probability of the sequence of $n$ outputs at the $K$ receivers, denoted as $Y_1^n, \cdots , Y_K^n $ with $Y_k^n \triangleq (Y_{k,1}, \cdots , Y_{k,n})$, respectively, for the $n$ inputs $X^n \triangleq (X_{1}, \cdots , X_n)$ factors as
\begin{equation}
p(y_1^n, \cdots , y_K^n|x^n)= \prod_{j=1}^nW(y_{1,j}, \cdots , y_{K,j}|x_j)
\end{equation}
where $X_j, Y_{1,j},\cdots Y_{K,j}$ are the channel input and outputs in the $j^{th}$ channel use. Denote the set of consecutive integers from $i$ to $j$ as $[i:j]$.
The message $M_{S} \in [1:2^{nR_S}] $ of rate $R_S$ is indexed by the subset $S\subseteq [1:K]$ of receivers it is intended for. Define $\mathsf{E}$ as the set of all message indices (with a message index being a subset of $[1:K]$) and let $\mathsf{P}$ be the power set of $[1:K]$ excluding the empty set. Hence, in general, $\mathsf{E} \subseteq \mathsf{P}$. 

For any $\mathsf{F}\subseteq \mathsf{P}$ and $i\in[1:K]$, define $\mathsf{W}_i^{\mathsf{F}}$ as the set of message indices in $\mathsf{F}$ of messages intended for receiver $i$ so that
\begin{equation}
\label{Eq:W_i_F}
\mathsf{W}_i^{\mathsf{F}} \triangleq \{S\in \mathsf{F}:i\in S\}
\end{equation}

Denote the set of all messages $\{M_S : S \in \mathsf{E} \} $ to be sent over a $K$-user DM BC as $ M_{\mathsf{E}}$. A $ (\{2^{nR_{S}}\}_{S \in \mathsf{E}}, n) $ code consists of (i) an encoder that assigns to each message tuple $ m_{\mathsf{E}} \in \prod_{S \in \mathsf{E}}[1:2^{nR_S}] $ a codeword $x^n( m_\mathsf{E})$ (ii) a decoder at each receiver, with the $k^{th}$ decoder mapping the received sequence $Y_{k}^n$ for each $k \in [1:K]$ into the set of decoded messages $ \{ \tilde{m}_{S} : S \in \mathsf{W}_k^{\mathsf{E}} \} \in \prod_{S \in \mathsf{W}_k^{\mathsf{E}}}[1:2^{nR_S}]$, denoted as $ \tilde{m}_{\mathsf{W}_k^{\mathsf{E}}}$. We will have occasion to refer to the received sequence at receiver $k$ from times $t_1$ to $t_2$, which we denote as $Y_{k,t_1}^{t_2}$. For simplicity, we choose to write  $Y_{i,1}^t$ (i.e., when $t_1=1$) as just $Y_i^t$.
The three-receiver DM BC is illustrated in Fig. \ref{Fig_SystemModekK_3}. The probability of error $P_e^{(n)}$ is the probability that not all receivers decode their intended messages correctly. The rate tuple $( R_{S} : S\in \mathsf{E} )$ is said to be achievable if there exists a sequence of $ (\{2^{nR_{S}}\}_{S \in \mathsf{E}},n) $ codes with $P_e^{(n)} \rightarrow 0$ as $n  \rightarrow  \infty $. The closure of the union of achievable rates is the capacity region.

When describing examples, we find it convenient to make certain notational simplifications when no confusion arises. For example, consider the three-receiver DM-BC with the message index set $\mathsf{E}{=}\{\{1\}, \{1,2,3\}\}$, so that there are two messages $M_{\{1\}}$ and $M_{\{1,2,3\}}$, the first one intended for the first receiver and the second for all three receivers. For simplicity, we will denote these messages as $M_1$ and $M_{123}$. Similarly, we will write their rates $R_{\{1\}}$ and $R_{\{1,2,3\}}$ simply as $R_1$ and $R_{123}$. Also, for convenience, we denote $\mathsf{E}{=}\{1, 123\}$ in this case. 
In other words, for simplicity, and when there is no confusion, we abbreviate the set $\{i_1,i_2,..,i_N\} \subseteq [1:K]$ as $i_1i_2 \cdots i_N$ (adopting the convention that $i_1 < i_2 < \cdots < i_N$) when all $i_j \in \{0,1, \cdots 9\}$. Note that with this notational simplification, when $K=3$, we have $\mathsf{P}{=}\{1,2,3,12,13,23,123\}$. If we have numbers and variables in the set, say as in $\{1,2,3,K-2,K-1,K\}$, we write the elements of the set consecutively as before but we separate the variables with dots to avoid confusion. Hence, the set $\{1,2,3,K-2,K-1,K\}$ is abbreviated as $123.K-2.K-1.K$, not $123K-2K-1K$ (which is confusing). If there is only one variable at the end, such as in $\{1,2,3,\cdots,K-2\}$ we still abbreviate it as by $123\cdots K-2$ since there is no confusion here.

In some cases, especially when the set $\{i_1,i_2,..,i_N\}$ has many elements, we find it more convenient to denote it by its complement. For example, the common message intended for all $K$ receivers is denoted by $M_{123\cdots K-1.K}$. It is simpler to denote it as $M_{\overline{\phi}}$ where $\phi$ is the empty set and $\overline{S}=\{1,2,3,\cdots,K\}\backslash S$ is the complement of $S$. Similarly, we can represent the message index set $\mathsf{E}=\{123\cdots K-2.K,123\cdots K-2.K-1\}$ of two messages, each required by $K-1$ receivers, simply as $\mathsf{E}=\{\overline{K-1},\overline{K}\}$. 


The combination network \cite{salimi2015generalized,romero2016superposition}, which is a special case of the general DM-BC, is described next. 
It consists, as described in \cite{salimi2015generalized}, of three layers of nodes, as shown in Fig. \ref{Fig_K3_combination networks} for the three-receiver case. The top layer consists of a single source node $X$, and the bottom layer consists of $K$ receivers $Y_i$, $i\in[1:K]$. The middle layer consists of $2^K{-}1$ intermediate nodes, denoted $V_{S}$ for all $S\in \mathsf{P}$. The source is connected to each of the intermediate nodes $V_S$ through a noiseless link of capacity $C_S$ (per channel use). Receiver $Y_i$ is connected to the intermediate nodes $V_{S}$ for all $ S \in \mathsf{W}_i^{\mathsf{P}} $ via noiseless links of unlimited capacity. An equivalent representation for the combination networks is given in \cite{romero2016superposition} wherein it is considered to be a network of noiseless DM BCs with the channel input $X$ connected in different ways to the channel outputs ($Y_1,Y_2, \cdots ,Y_K$) each through a noiseless BC. In particular, the channel input $X$ contains $2^{K}-1$ components $V_S$, for all $S \in  \mathsf{P}$. For each $S$, the component $V_S \in \mathcal{V}_S$, where $|\mathcal{V}_S|=2^{C_S}$, is noiselessly received at each receiver $Y_i$ for all $i \in S$ and {\em not} received at the receivers $Y_j$ with $j\not\in S $, i.e., $Y_i= \{V_S : S \in \mathsf{W}_i^{\mathsf{P}} \} \triangleq V_{\mathsf{W}_i^{\mathsf{P}}} $.

\begin{figure}
\centering
\begin{tikzpicture}



\node  (s) at (0,4)  [circle,draw] {$X$};
\node (v1) at (-3.5,1)   [circle,draw] {$V_{1}$};
\node (v2) at (-2.3,1)    [circle,draw] {$V_{2}$};
\node (v3) at (-1.1,1)    [circle,draw] {$V_{3}$};
\node (v12) at (-0,1)    [circle,draw] {$V_{12}$};
\node (v13) at (1.2,1)    [circle,draw] {$V_{13}$};
\node (v23) at (2.5,1)    [circle,draw] {$V_{23}$};
\node (v123) at (3.8,1)    [circle,draw] {$V_{123}$};

\node (y1) at (-1.5,-1)    [circle,draw] {$Y_1$};
\node (y2) at (0,-1)    [circle,draw] {$Y_2$};
\node (y3) at (1.5,-1)    [circle,draw] {$Y_3$};

\node (c1) at (-1.9,2.7) {$C_{1}$};
\node (c2) at (-1.4,2.5) {$C_{2}$};
\node (c3) at (-.85,2.3) {$C_{3}$};

\node (c12) at (-.2,2.2) {$C_{12}$};
\node (c13) at (1,2.24) {$C_{13}$};
\node (c23) at (1.7,2.3) {$C_{23}$};

\node (c123) at (2.5,2.4) {$C_{123}$};


\draw [->] (s) to  (v1);
\draw [->] (s) to  (v2);
\draw [->] (s) to  (v3);
\draw [->] (s) to  (v12);
\draw [->] (s) to  (v13);
\draw [->] (s) to  (v23);
\draw [->] (s) to  (v123);
\draw[thick,dash dot] [->] (v1) to  (y1);
\draw[thick,dash dot] [->] (v2) to  (y2);
\draw[thick,dash dot] [->] (v3) to  (y3);
\draw[thick,dash dot] [->] (v12) to  (y1);
\draw[thick,dash dot] [->] (v12) to  (y2);
\draw[thick,dash dot] [->] (v13) to  (y1);
\draw[thick,dash dot] [->] (v13) to  (y3);
\draw[thick,dash dot] [->] (v23) to  (y2);
\draw[thick,dash dot] [->] (v23) to  (y3);
\draw[thick,dash dot] [->] (v123) to  (y1);
\draw[thick,dash dot] [->] (v123) to  (y2);
\draw[thick,dash dot] [->] (v123) to  (y3);

\end{tikzpicture}
\caption{A combination network with $7$ intermediate nodes and three receivers. The dark lines represent finite capacity links while the dashed lines represent infinite capacity links. The capacity of the dark line connecting the node $X$ to the node $V_S$ is $C_S$ per channel use for each $S\in \mathsf{P}$. For brevity, the source/destination nodes are denoted by their transmitted/received symbols and the intermediate nodes by their output symbols.
\label{Fig_K3_combination networks}
}
\end{figure}
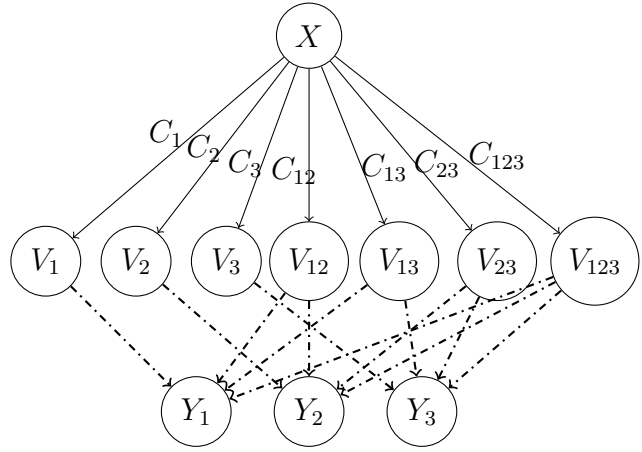
\subsection{Just Enough Order Theory}

We introduce ideas from order theory following the notation in \cite{davey2002introduction}. Any set equipped with an order is an ordered set. Let $P$ be such an ordered set and $Q$ be a subset of $P$. We say that $Q$ is 

\begin{enumerate}
\item an {\em up-set} if $x\in Q$, $y\in P$, and $y\geq x$ implies $y\in Q$.
\item a {\em down-set} if $x\in Q$, $y\in P$, and $y\leq x$ implies $y\in Q$.
\end{enumerate} Note that these two types of subsets are duals of each other, i.e., if $Q$ is a {\em down-set} then $P\backslash Q$ is an {\em up-set}. Moreover, for any subset $Q\subseteq P$, we define the smallest {\em down-set} containing $Q$ as $\downarrow_{P} Q = \{y \in P: y\leq x, x \in Q\}$ and the smallest {\em up-set} containing $Q$ as $\uparrow_{P} Q = \{y \in P: x\leq y, x \in Q\}$. Further, for any $Q_1,Q_2\subseteq P$, denote the part of the smallest {\em down-set} containing $Q_1$ that is also in $Q_2$, i.e., $ ( \downarrow_{P} Q_1 ) \cap Q_2$, as $\downarrow_{Q_2} Q_1 $. Similarly, $ ( \uparrow_{P} Q_1 ) \cap Q_2$, the smallest {\em up-set} of $Q_1$ that is in $Q_2$, is denoted as $ \uparrow_{Q_2} Q_1 $. Henceforth, for brevity, $\downarrow_{Q_2} Q_1 $ (or $ \uparrow_{Q_2} Q_1 $) is referred to as the down-set (or up-set, respectively) of $Q_1$ in $Q_2$.

Also, let $\mathcal{F}_{\downarrow}(P)$ denote the family of all down-sets of $P$ and $\mathcal{F}_{\uparrow}(P)$ denote the family of all up-sets of $P$. Finally, let $\mathcal{F}_{\uparrow_Q}(P)$ and $\mathcal{F}_{\downarrow_Q}(P)$ denote the family of all {\em up-sets} and all {\em down-sets} of $P$ that contain $Q$, respectively. 

In this paper, we will take the ground set to be a set of sets, such as the set of non-empty subsets of $[1:K]$, the receiver index set. We will denote a set of sets in sans-serif font to distinguish it from sets. The order on the ground set considered in this paper is exclusively that of set inclusion, i.e., $S_1\leq S_2$ if and only if $S_1\subseteq S_2$. 
Recall that, for simplicity, we write the index set $\{i_1,i_2,..,i_N\}$ as $i_1i_2 \cdots i_N$ (adopting the convention that $i_1 < i_2 < \cdots < i_N$). To illustrate such notation, consider the example of $K=3$. The ground set in this case could be the set of all non-empty subsets of $[1:3]$, denoted as $\mathsf{P} = \{1,2,3,12,13,23,123\}$. The down-set of, say, $\{12\}$ in $\mathsf{P}$ is $\downarrow_{\mathsf{P}} \{12\}=\{1,2,12\}$, and the up-set of $\{12\}$ in $\mathsf{P}$ is $\uparrow_{\mathsf{P}} \{12\}= \{12,123\}$. For the same $\mathsf{P}$, we have $\uparrow_{\mathsf{W}_3^{\mathsf{P}}} \{12\}=\{\uparrow_{\mathsf{P}} \{12\}\}\cap \mathsf{W}_3^{\mathsf{P}}=\{123\}$, whereas $\downarrow_{\mathsf{W}_3^{\mathsf{P}}} \{12\}=\{\downarrow_{\mathsf{P}} \{12\}\}\cap \mathsf{W}_3^{\mathsf{P}}= \phi$. 

To illustrate families of up-sets and down-sets, consider the ground set $\mathsf{P}=\{1,2,12\}$. Then, $\mathcal{F}_{\downarrow}(\mathsf{P}) = \{ \{1\} , \{2\} , \{1,2\} , \{1,2,12\} \}$, while $\mathcal{F}_{\uparrow}(\mathsf{P}) = \{\{1,12\},\{2,12\},\{12\},\{1,2,12\}\}$. For the same $\mathsf{P}=\{1,2,12\}$, we have $\mathcal{F}_{\uparrow_{ \{1\} } }(\mathsf{P})=\{   \{1,12\},\{1,2,12\}   \}$ and $\mathcal{F}_{\downarrow_{ \{1\} } }(\mathsf{P})=\{   \{1\},\{1,2\},\{1,2,12\}   \}$.


\begin{lemma}
\label{Lemma_Sets}
The following relationships are true:
\begin{enumerate}
\item For any set $S=\{i_1,i_2,\cdots, i_N\} \subseteq \{1,2,\cdots ,K\}$, we have  
\begin{align}
\cup_{k\in S}\mathsf{W}_k^{\mathsf{P}}&= \uparrow_{\mathsf{P}}\{i_1,i_2,\cdots, i_N\}\label{Eq_Lemma_1}\\
\cap_{k\in S}\mathsf{W}_k^{\mathsf{P}}&= \uparrow_{\mathsf{P}}\{i_1i_2\cdots i_N\} \label{Eq_Lemma_2}
\end{align}
\item For any set of sets $\mathsf{W}\subseteq \mathsf{P}$,
\begin{align}
\cup_{S\in \mathsf{W}}\downarrow_{\mathsf{W}_i^{\mathsf{P}}}\{S\}&= \downarrow_{\mathsf{W}_i^{\mathsf{P}} } \mathsf{W} \label{Eq_Lemma_3}\\
\cap_{S\in \mathsf{W}}\downarrow_{\mathsf{W}_i^{\mathsf{P}}}\{S\}&= \downarrow_{\mathsf{W}_i^{\mathsf{P}} }\{\cap_{S\in \mathsf{W}}S\} \label{Eq_Lemma_4}
\end{align}
\item For any set $S=i_1i_2\cdots i_N \subset \{1,2,\cdots, K\}$ and $i\in [1:K]$ 
\begin{align}
\downarrow_{\mathsf{W}_i^{\mathsf{P}}}\{\overline{i_1},\overline{i_2},\cdots, \overline{i_N}\} \cup \uparrow_{\mathsf{W}_i^\mathsf{P}}\{S\}&=\mathsf{W}_i^{\mathsf{P}} \label{Eq_Lemma_*}\\
\downarrow_{\mathsf{W}_i^{\mathsf{P}}}\{\overline{i_1},\overline{i_2},\cdots, \overline{i_N}\} \cap \uparrow_{\mathsf{W}_i^\mathsf{P}}\{S\}&=\phi \label{Eq_Lemma_**}\\
\downarrow_{\mathsf{W}_i^{\mathsf{P}}}\{\overline{S}\} \cup \uparrow_{\mathsf{W}_i^\mathsf{P}}\{i_1,i_2,\cdots ,i_N\}&=\mathsf{W}_i^{\mathsf{P}} \label{Eq_Lemma_***}\\
\downarrow_{\mathsf{W}_i^{\mathsf{P}}}\{\overline{S}\} \cap \uparrow_{\mathsf{W}_i^\mathsf{P}}\{i_1,i_2,\cdots ,i_N\}&=\phi
\label{Eq_Lemma_****}
\end{align}

\end{enumerate}
\end{lemma}

\begin{proof}
The proofs of all the above equalities are straightforward given the order theoretic definitions except that of \eqref{Eq_Lemma_*} and \eqref{Eq_Lemma_***}, which are given in Appendix \ref{App_Prove_Lemma_Sets}.
\end{proof}

\section{A Class of Achievable Rate Regions and Some Converses: Two  Messages}
\label{Sec_General_two}


This paper is devoted to the problem of sending two groupcast messages over the $K$-receiver DM-BC. Let the two general messages be $M_{S_1}$ and $ M_{S_2}$, so that the message index set is $\mathsf{E}=\{S_1,S_2\}$. Without loss of generality, we let $S_1=\{1, 2 , \cdots , P, P+1, \cdots , P+L_1\}$ and $S_2= \{1, 2, \cdots, P, P+L_1+1,\cdots ,P+L_1+L_2\}$. The set of indices of receivers that decode both messages is denoted as $S_p=\{1,2,\cdots P\}$, that decode only $M_{S_1}$ is denoted by $S_{l_1}=\{P+1,P+2,\cdots , P+L_1\}$, and that decode only $M_{S_2}$ is denoted by $S_{l_2}=\{P+L_1+1, P+L_1+2, \cdots , P+L_1+L_2\}$. The $P$ receivers with indices in $S_p$ can be thought of as private receivers, the $L_1$ receivers with indices in $S_{l_1}$ can be thought of as the first group of common receivers that decode only $M_{S_1}$, and the $L_2$ receivers with indices in $S_{l_2}$ can be thought of as the second group of common receivers that decode only $M_{S_2}$. Obviously, $P+L_1+L_2{=}K$. 

Of special interest for proving converses later in this paper are the three special cases of (a) two order-$(K-1)$ messages (i.e., messages that are intended at two distinct sets of $K-1$ of the $K$ users) so that $L_1=L_2=1$ and $P=K-2$ (b) one order-$(K-1)$ message and another order-$(P+1)$ message so that $L_1=K-P-1$ and $L_2=1$ for any $ 0 \leq P \leq K-2 $ and (c) two {\em nested} messages so that either $L_1=0$ or $L_2=0$.
  
Next, we obtain a class of new inner bounds for the $K$-user DM BC with two general messages.
We use order theory to describe our result. In particular, let $\mathsf{P}$, the set of all non-empty subsets of receiver indices $[1:K]$, be the ground set. As stated previously, we will think of $\mathsf{P}$ as an ordered set with the order relation defined by set inclusion, i.e., $ S \leq S'$ if and only if $S \subseteq S'$. Evidently, the message index set $ \mathsf{E} = \{ S_1, S_2\} \subset \mathsf{P}$. 

\begin{remark}[Message Set Expansion]
The class of inner bounds of the next theorem is parameterized by a message index superset $\mathsf{F}$ that is some superset of message indices that contains the message indices in $\mathsf{E}$ and can be as large as $\mathsf{P}$, so that $\mathsf{P}\supseteq \mathsf{F}\supseteq\mathsf{E}$. For each $K$ and each message set $\mathsf{E}$, every choice of the design parameter $\mathsf{F}$ in the next theorem gives a distinct inner bound. Note that there are $2^{2^K-3}$ choices of $\mathsf{F}$ for any two groupcast message set $\mathsf{E} = \{S_1, S_2\}$. We refer to expanding $\mathsf{E}$ to $\mathsf{F}$ simply as {\em message set expansion}. While the choice of $\mathsf{F}$ strongly determines the coding scheme employed and hence the rate region it achieves, both the description of the scheme and the characterization of its rate region can be done for any admissible $\mathsf{F}$ in a general way.
\end{remark}


\begin{theorem}
\label{TH_genral_inner_General_msgs}
Let $\mathsf{F}$ be some message index superset so that $\mathsf{P}\supseteq \mathsf{F}\supseteq\mathsf{E}$. The rate pair $(R_{S_1}, R_{S_2})$ is achievable if there exist non-negative up-set split rates ($R_{S\rightarrow S'} :S\in \mathsf{E}, S' \in \mathsf{F}, S\subseteq S^{'}$) such that for each $i \in \{1,2\}$
\begin{equation}
\label{Eq_RS_upset_splittin}
R_{S_i}=\sum_{S^{'}\in \uparrow_{\mathsf{F}} S_i }R_{S_i \rightarrow S^{'}} 
\end{equation}
and reconstruction rates 
\begin{equation}
\label{Eq_reconstruction_rates_upset_splitting}
\hat{R}_{S^{'}}= \sum_{S\in \downarrow_{\mathsf{E}} S^{'}}R_{S\rightarrow S^{'}} \qquad \forall S^{'}\in \mathsf{F}
\end{equation}
that satisfy the inequalities
\begin{align}
&\sum_{S^{'}\in \mathsf{B}}  \hat{R}_{S^{'}} \leq I(U_{\mathsf{B}}; Y_j|U_{\mathsf{W}_j^{\mathsf{F}}\backslash \mathsf{B} }, Q), \; \forall \mathsf{B}\in \mathcal{F}_{\downarrow}(\mathsf{W}_j^{\mathsf{F}}), \; \forall  j\in S_p \label{Eq_condition_reconstruction_rates1}
\end{align}
and, for each $ i \in \{1,2\}$, the inequalities
\begin{align}
\sum_{S^{'}\in \mathsf{B}}  \hat{R}_{S^{'}} \leq I(U_{\mathsf{B}}; Y_j|U_{\mathsf{W}_j^{\mathsf{F}}\backslash \mathsf{B}   }, Q), \; \forall \mathsf{B}\in \mathcal{F}_{\downarrow_{   \{S_i\}  }  }(\mathsf{W}_j^{\mathsf{F}}), \; \forall j\in S_{l_i} 
\label{Eq_condition_reconstruction_rates2}
\end{align}
for some time sharing and auxiliary random variables $Q$ and $U_{\mathsf{F}} \triangleq \{U_S : S \in \mathsf{F} \}$ with a joint distribution that factors as $p(q, u_{\mathsf{F}})= p(q) \prod_{S\in \mathsf{F}}p(u_S|u_{(\uparrow_{\mathsf{F}} S)\backslash \{S \}},q)$ and $X$ taken to be a deterministic function of $(Q, U_{\mathsf{F}})$. 
\end{theorem}

\begin{proof}
A detailed proof is given in Appendix \ref{App_Prove_General_msgs}. We only provide an outline here.
Each message $M_S$ is divided into a collection of sub-messages $ M_{S\rightarrow S'}$ where $S' \in \uparrow_{ \mathsf{F}} S$, for each $S\in \mathsf{E}$ as per \eqref{Eq_RS_upset_splittin}. This explains why this form of rate-splitting is called up-set rate splitting.
The sub-message $M_{S\rightarrow S^{'}}$ will be treated as if it was intended for the receivers in the larger set $S^{'}$ instead of in $S$. 
By reassembling all sub-messages intended for the set of receivers with indices in $S'$, we obtain the reconstructed message $\hat{M}_{S^{'}}=(M_{S\rightarrow S^{'}}:S\in  \downarrow_{\mathsf{E}} S^{'})$ for all $S^{'}\in \mathsf{F}$ with rate $\hat{R}_{S^{'}}$ given by \eqref{Eq_reconstruction_rates_upset_splitting}. Note that the message set expansion parameter $\mathsf{F}$ determines how the message splitting is done and hence how the newly reconstituted messages $ M_{S\rightarrow S'} $ for all $ S^{'}\in \mathsf{F}$ are reconstructed, which in turn determines the superposition coding scheme described next.
The set of reconstructed messages with indices in $\mathsf{F}$ are encoded using superposition coding with dependent auxiliary codeword generation according to the subset inclusion order (cf. \cite{romero2017unifying}) as described in Appendix \ref{App_Prove_General_msgs}. Private receiver $Y_j$ jointly decodes the desired messages $M_{S_1}$ and $M_{S_2}$ 
via the {\em unique} joint decoding of the set of reconstructed messages $( \hat{M}_S : S \in \mathsf{W}_j^{\mathsf{F}} )$ that contain those two messages for every $j\in S_p$. As shown in Appendix \ref{App_Prove_General_msgs}, the reconstructed messages can be reliably transmitted over the DM BC if the partial sums of the reconstructed message rates satisfy the inequalities given by \eqref{Eq_condition_reconstruction_rates1}.
On the other hand, the common receiver $Y_j$ (with $j\in S_{l_i} $, $i \in \{1,2\}$) only needs to decode the message $M_{S_{i}}$. 
Hence, non-unique decoding can be employed by these receivers. Note that for each $j\in S_{l_i} $, the reconstructed messages $( \hat{M}_S : S \in \mathsf{W}_j^{\mathsf{F}} )$ contain the desired message $M_{S_{i}}$ as well as partial interference via up-set message splitting and reconstruction. Thus, among these reconstructed messages, only the ones with indices in $\uparrow_{\mathsf{F}} \mathsf{W}_j^{\mathsf{E}}$  are uniquely decoded because per \eqref{Eq_RS_upset_splittin} and \eqref{Eq_reconstruction_rates_upset_splitting}, all such reconstructed messages are needed to reconstitute this receiver's desired message $M_{S_{i}}$, whereas the rest of the reconstructed messages do not, and these messages are hence decoded non-uniquely. This happens successfully with high probability if the partial sums of the reconstructed message rates satisfy the inequalities given by \eqref{Eq_condition_reconstruction_rates2}.
\end{proof}


\begin{remark}[Non-unique decoding]
In the conference paper \cite[Theorem 2]{romero2017rate} by Romero and Varanasi, an inner bound for a {\em general} message set $\mathsf{E}$ was proposed that used the same encoding scheme but with 
each receiver $Y_j$ {\em uniquely} decoding all reconstructed messages with indices in $\mathsf{W}_j^{\mathsf{F}}$ for all $j$. Hence, a common receiver may end up decoding some reconstructed messages that only contain sub-messages of the message that is not of interest to it which in turn produces more inequalities on the reconstruction rates, hence possibly strictly limiting the achievable rate region compared to the one given in Theorem \ref{TH_genral_inner_General_msgs}. 
The idea of non-unique decoding was first presented in the context of the two-user interference channel in \cite{chong2006comparison} and later used as indirect decoding in \cite{nair2009capacity} for the DM BC with the nested message set $\mathsf{E}=\{1,123\}$ to make notable progress on establishing a matching converse for a class of DM BCs. In both those cases, each receiver that does non-unique decoding decodes its desired message(s) uniquely and the common sub-message of the undesired message non-uniquely. That idea can be seen to be generalized here as follows: each common receiver decodes (a) a subset of reconstructed messages intended for it uniquely, each necessarily including one or more sub-message of the desired message and possibly also some sub-message(s) of the undesired message and (b) a subset of reconstructed messages that do not include any sub-message of the desired message non-uniquely. Moreover, it is interesting, thanks to the order-theoretic formulation, that a general expression \eqref{Eq_condition_reconstruction_rates2} for the achievable rate region is possible even after incorporating such decoding at the common receivers. 
\end{remark}

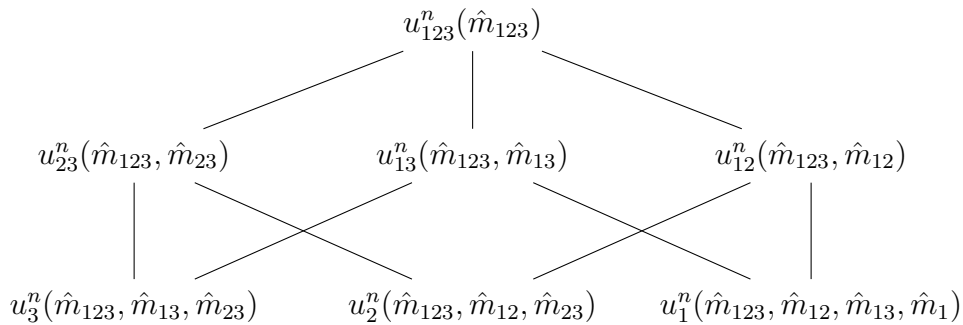
\begin{figure}[t]
\centering
\begin{tikzpicture}
\node (u123) at (0,1.75)  {$u_{123}^n(\hat{m}_{123})$};
\node (u12) at (4.5,0)  {$u_{12}^n(\hat{m}_{123},\hat{m}_{12})$};
\node (u13) at (0,0)  {$u_{13}^n(\hat{m}_{123},\hat{m}_{13})$};
\node (u23) at (-4.5,0)  {$u_{23}^n(\hat{m}_{123},\hat{m}_{23})$};
\node (u1) at (4.5,-2)  { $ u_1^n(\hat{m}_{123},\hat{m}_{12},\hat{m}_{13},\hat{m}_{1})$};
\node (u2) at (0,-2)  {$  u_2^n(\hat{m}_{123},\hat{m}_{12},\hat{m}_{23})$};
\node (u3) at (-4.5,-2)  {$ u_3^n(\hat{m}_{123},\hat{m}_{13},\hat{m}_{23})$};

\draw [-] (u123) -- (u12);
\draw [-] (u123) -- (u13);
\draw [-] (u123) -- (u23);
\draw [-] (u1) -- (u12);
\draw [-] (u1) -- (u13);
\draw [-] (u2) -- (u12);
\draw [-] (u2) -- (u23);
\draw [-] (u3) -- (u13);
\draw [-] (u3) -- (u23);


\end{tikzpicture}
\caption{A Hasse diagram for the coding scheme for Example \ref{Example_P_0_C1_1_C2_2} where the message index superset $\mathsf{F}{=}\mathsf{P}$. Each line represents superposition coding with codebooks generated top to bottom, i.e., we first generate the codebook $U_{123}$ (whose codewords are ``cloud centers") that represents $\hat{M}_{123}$. Then, using superposition coding, we conditionally independently generate the codebooks $U_{12},U_{13}$ and $U_{23}$ (conditioned on the $U_{123}$ cloud centers) that represent $\hat{M}_{12},\hat{M}_{13}$ and $\hat{M}_{23}$, respectively. The corresponding codewords form three primary satellite codebooks for each cloud center. Finally, for each $U_{12}$ and $U_{13}$ satellite codeword pair, we generate the secondary satellite $U_1$ codebook that represents the message $\hat{M}_{1}$. Moreover, for each $U_{12}$ and $U_{23}$ satellite codeword pair, we generate a single $U_2$ codeword dependent on that pair (since $\hat{R}_2 {=} 0$). Similarly, for each $U_{13}$ and $U_{23}$ satellite codeword pair, we generate a single $U_3$ codeword dependent on that pair (since $\hat{R}_3 {=} 0$). See Remarks \ref{rem-largeF} and \ref{remarkzerorate} for further discussion of this point.
\label{Fig_P_0_C1_1_C2_2}} 
\end{figure}

Next, we provide a simple example to illustrate  Theorem \ref{TH_genral_inner_General_msgs}.
\begin{example}
\label{Example_P_0_C1_1_C2_2}
Consider the case $K=3$ and $\mathsf{E}=\{1,23\}$ so that $P=0$, $L_1=1$, and $L_2=2$ and $S_p=\phi$,  $S_{l_1}=\{1\}$, $S_{l_2}=\{2,3\}$. Choose the largest possible message index superset $\mathsf{F}{=}\mathsf{P}$. 
Up-set message splitting described in the proof of Theorem \ref{TH_genral_inner_General_msgs} yields $M_1 {=} ( M_{1\rightarrow 1}, M_{1\rightarrow 12}, M_{1\rightarrow 13}, M_{1\rightarrow 123} )$ and $M_{23} {=} (M_{23 \rightarrow 23}, M_{23 \rightarrow 123})$ with split rates defined according to \eqref{Eq_RS_upset_splittin}. The reconstructed messages and their rates as per \eqref{Eq_reconstruction_rates_upset_splitting} are given as 
\begin{align*}
&\hat{M}_1=M_{1\rightarrow 1}  &&\hat{R}_1=R_{1\rightarrow 1}\\
&\hat{M}_2=\phi  &&\hat{R}_2=0\\
&\hat{M}_3=\phi  &&\hat{R}_3=0\\
&\hat{M}_{12}=M_{1\rightarrow 12}  &&\hat{R}_{12}=R_{1\rightarrow 12}\\
&\hat{M}_{13}=M_{1\rightarrow 13}  &&\hat{R}_{13}=R_{1\rightarrow 13}\\
&\hat{M}_{23}=M_{23\rightarrow 23}  &&\hat{R}_{23}=R_{23\rightarrow 23}\\
&\hat{M}_{123}=(M_{1\rightarrow 123},M_{23\rightarrow 123})  &&\hat{R}_{123}=R_{1\rightarrow 123}+R_{23\rightarrow 123}
\end{align*}
The resulting rate-splitting/superposition coding scheme described in the proof of Theorem \ref{TH_genral_inner_General_msgs} is illustrated in Fig. \ref{Fig_P_0_C1_1_C2_2} with the specifics explained in its caption. The key point to note is the generation of the $U_2$ and $U_3$ codewords towards the end of that caption which has to do with $\hat{R}_2 {=} 0$ and $\hat{R}_3 {=} 0$. For every pair of $U_{12}$ and $U_{23}$ codewords there is a single $U_2$ codeword and similarly for every pair of $U_{13}$ and $U_{23}$ codewords there is a single $U_3$ codeword. If one uses the well known ``cloud and satellite" symbolism, it is as if each $U_2$ (and $U_3$) codeword were a single satellite codeword for every pair of $U_{12}$ and $U_{23}$ (and $U_{13}$ and $U_{23}$, resp.) cloud codewords. Whereas, the usual picture one has is quite the contrary in that for every cloud codeword there are many satellite codewords. Note that this latter and usual symbolism works for all other codewords (i.e., other than $U_2$ and $U_3$ codewords). It is no accident that it is precisely the case that only $\hat{R}_2 {=} 0$ and $\hat{R}_3 {=} 0$, which in turn is a consequence of having chosen the message index superset $\mathsf{F}$ to be equal to $\mathsf{P}$. This phenomenon of zero-rate codebooks superposed over other codebooks has crucially important consequences that we will explain further in this example and in more depth in much of the paper later.

From the conditions for reliable communication of the messages at their desired destinations given in \eqref{Eq_condition_reconstruction_rates2} (note that  \eqref{Eq_condition_reconstruction_rates1} is vacuous since $S_p {=} \phi $ in this example) of Theorem \ref{TH_genral_inner_General_msgs}, we get that the reconstructed message rates must satisfy the inequalities
\begin{align}
\hat{R}_{123}+\hat{R}_{13}+\hat{R}_{12}+\hat{R}_{1}\leq &I(U_{123},U_{13},U_{12},U_{1};Y_1 |Q)\nonumber \\
\hat{R}_{13}+\hat{R}_{12}+\hat{R}_{1}\leq &I(U_{13},U_{12},U_{1};Y_1|U_{123},Q)\nonumber \\
\hat{R}_{12}+\hat{R}_{1}\leq &I(U_{12},U_{1};Y_1|U_{123},U_{13},Q)\nonumber \\
\hat{R}_{13}+\hat{R}_{1}\leq &I(U_{13},U_{1};Y_1|U_{123},U_{12},Q)\nonumber \\
\hat{R}_{1}\leq &I(U_{1};Y_1|U_{123},U_{13},U_{12},Q)\nonumber \\ \nonumber \\
\hat{R}_{123}+\hat{R}_{23}+\hat{R}_{12}\leq &I(U_{123},U_{23},U_{12},U_{2};Y_2|Q)\nonumber \\
\hat{R}_{23}+\hat{R}_{12}\leq &I(U_{23},U_{12},U_{2};Y_2|U_{123},Q)\nonumber \\
\hat{R}_{23}\leq &I(U_{23},U_{2};Y_2|U_{123},U_{12},Q)\nonumber \\ \nonumber \\
\hat{R}_{123}+\hat{R}_{23}+\hat{R}_{13}\leq &I(U_{123},U_{23},U_{13},U_{3};Y_3|Q)\nonumber \\
\hat{R}_{23}+\hat{R}_{13}\leq &I(U_{23},U_{13},U_{3};Y_3|U_{123},Q)\nonumber \\
\hat{R}_{23}\leq &I(U_{23},U_{3};Y_3|U_{123},U_{13},Q)\nonumber
\end{align}
for some $p(q) p(u_{123},q)p(u_{12}|u_{123},q)p(u_{13}|u_{123},q)$ $p(u_{23}|u_{123},q)$ $p(u_1|u_{12},u_{13},u_{123},q)$ $p(u_2|u_{12},u_{23},u_{123},q)$ $p(u_3|u_{13},u_{23},u_{123},q)$ are achievable. 
\end{example}

\begin{remark}[What's unusual in the rate region?]
The inequalities above may appear to be ``standard" for rate-splitting, superposition coding and non-unique 
decoding if one realizes that in order to get a satellite codeword right one has to get its cloud codeword right as well and hence one doesn't have sum rate inequalities for all possible sum rates (cf. \cite{romero2017unifying}). However, even though $\hat{R}_2 {=} 0$ and $\hat{R}_3 {=} 0$ the unusual, notice the unusual presence of the $U_2$ and $U_3$ auxiliary random variables.
\end{remark}

\begin{remark}[Perspective]
While the above reliability conditions may seem almost familiar (except for the presence of $U_2$ and $U_3$), it is important to note that such inequalities have been succinctly expressed in Theorem \ref{TH_genral_inner_General_msgs} in complete generality, i.e., for any $K$, any message set $\mathsf{E}=\{S_1, S_2\}$, and any choice of the message set expansion parameter $\mathsf{F}$. The order-theoretic formulation turns out to just the right one to express a result of its generality with such analytical brevity. Quite importantly, this brevity is key to further work with the result of Theorem \ref{TH_genral_inner_General_msgs} as we do in the rest of this paper.
\end{remark}

As an aside, we note the simple point that the inner bound for the same example using the result in \cite[Theorem 2]{romero2017rate} would have the two additional inequalities
\begin{align*}
\hat{R}_{12} \leq I(U_{12},U_2;Y_2|U_{123},U_{23},Q)\\
\hat{R}_{13} \leq I(U_{13},U_3;Y_3|U_{123},U_{23},Q)
\end{align*}
because in that scheme 
receiver $Y_2$ uniquely decodes the undesired sub-message $M_{1\rightarrow 12} $ and receiver $Y_3$ uniquely decodes the undesired sub-message $ M_{1\rightarrow 13}$, whereas in the scheme of Theorem \ref{TH_genral_inner_General_msgs} those sub-messages are decoded non-uniquely at Receivers 2 and 3, respectively.

\begin{remark}[Is non-unique decoding necessary \cite{bidokhti2014non}? The answer is ``yes"]
In previous uses of non-unique decoding in \cite{chong2008han,nair2009capacity,salman2019capacity,bidokhti2014non} (see also the several references of \cite{bidokhti2014non}) it has been shown that the rate region for non-unique decoding taken as a union of admissible distributions is not strictly larger than that with unique decoding. Whether the use of non-unique decoding strictly enlarges the inner bound of Theorem \ref{TH_genral_inner_General_msgs} is not the focus of this work, although a strict enlargement of the rate region would be desirable. Rather, we have adopted non-unique decoding because (a) it expands the rate region per coding distribution (i.e., admissible distribution of the auxiliary random variables and encoding function) (b) it does not shrink the rate region taken as a union over all admissible coding distributions and (c) there are fewer inequalities leading to relative ease of doing projections via FME and (d) proving converses in some special cases becomes easier as a result as we show later since there are fewer inequalities in the rate region after FME. The points (b-d) have been articulated before in their respective contexts (cf. \cite{chong2008han,nair2009capacity,salman2019capacity,bidokhti2014non}), and as in those contexts, and even to a greater degree here, we realize the benefits (c-d) in performing FME more efficiently than would be otherwise possible with unique decoding (if FME would be possible at all with reasonable effort in the latter case) in obtaining Corollaries \ref{Colollary_general_indirect_C1_1_C2_1},
\ref{Colollary_nested_indirect_1commonRec}, and 
\ref{Colollary_nested_indirect_2commonRec}, and
proving converse results for the DM BC with certain message sets in Theorem \ref{Th-Capacity_Region_Two_K-1_Order_msgs} 
and Proposition 
\ref{Proposition_1Private_1OrderK-1_messages}, all of which are to follow. But it is also the first point (a) that we find very interesting and useful when we specialize Theorem \ref{TH_genral_inner_General_msgs} to the combination network, and choosing certain single coding distributions to obtain ``large" rate regions and even extremal ones in some cases leading to the capacity region as in Theorem \ref{Th_Capacity_combination_networks_C1_1_C2_1} and Propositions \ref{Th_Capacity_combination_networks_C1},  \ref{Th_Capacity_combination_networks_C2} and \ref{Proposition_Regenerate_Bidokhti_el_al} and Examples \ref{Example_Bidokhti_Example_K=6} and \ref{Example_Bidokhti_Example_K=7} to follow. 
Indeed, all of Section \ref{Sec:Smaller_F} is dedicated to the interesting task of finding simple coding schemes among the class of coding schemes proposed and analyzed in Theorem \ref{TH_genral_inner_General_msgs} that yet achieve the capacity region in several settings of message sets through clever choices of the coding distribution in each case. The success we achieve in Section \ref{Sec:Smaller_F} is due in no small part to non-unique decoding as well.
Finally, we note that
even if some of the aforementioned achievability and capacity results were possible with unique decoding we believe that the associated rate and capacity region descriptions would be less compact, and the proofs more laborious, in each case than the ones we obtain here.
\end{remark}

\begin{remark}[Rate region versus complexity of the coding schemes]
\label{rem-largeF}
In Theorem \ref{TH_genral_inner_General_msgs}, for every possible message set expansion from $\mathsf{E}$ to $\mathsf{F}$ such that $\mathsf{P} {\supseteq }\mathsf{F} {\supseteq} \mathsf{E}$, and there are $2^{2^K-3}$ of them, we get a different achievable region which involves a different set of auxiliary random variables. Denote it as ${\cal R}(K,\mathsf{E},\mathsf{F})$ for brevity.
Expanding $\mathsf{F}$ leads to finer message splitting, and hence to the use of more auxiliary random variables/codebooks, and it therefore cannot reduce the achievable region. Hence, the full power of the coding scheme of Theorem \ref{TH_genral_inner_General_msgs} is realized by setting $\mathsf{F}{=}\mathsf{P}$. This fact is exploited extensively in Sections \ref{Sec_General_two}-\ref{Sec_Relation_Bidokhti}.
More generally, if 
$\mathsf{F}_1 \supset \mathsf{F}_2 $ then ${\cal R}(K,\mathsf{E},\mathsf{F}_1) \supseteq {\cal R}(K,\mathsf{E},\mathsf{F}_2) $. 
Nevertheless, we prefer to leave $\mathsf{F}$ as a parameter to be chosen rather than replace it with $\mathsf{P} $ in Theorem \ref{TH_genral_inner_General_msgs} since a smaller $\mathsf{F}$ leads to a simpler coding scheme (with fewer random variables and hence fewer codebooks), and sometimes a specific such choice suffices to achieve capacity as we illustrate later in Section \ref{Sec:Smaller_F}.  
\end{remark}

\begin{remark}[Multiple codewords for some groups of sub-messages]
\label{remarkzerorate}
Interestingly, when we choose $\mathsf{F} \supset \uparrow_{\mathsf{P}}\mathsf{E}$ we get some zero reconstruction rates in \eqref{Eq_reconstruction_rates_upset_splitting} and this important point is explained in this remark. For example, when we chose $\mathsf{F}{=}\mathsf{P} \supset \uparrow_{\mathsf{P}}\mathsf{E}$ in Example \ref{Example_P_0_C1_1_C2_2}, we got two {\em zero} reconstruction rates, namely, $\hat{R}_2$ and $\hat{R}_3$, per \eqref{Eq_reconstruction_rates_upset_splitting}. 
This is reflected in Fig. \ref{Fig_P_0_C1_1_C2_2} which depicts the superposition coding scheme described in Appendix \ref{App_Prove_General_msgs} for
Example \ref{Example_P_0_C1_1_C2_2}. In particular, the codewords $u_{2}^n(\hat{m}_{123},\hat{m}_{12},\hat{m}_{23})$ and $u_3^n(\hat{m}_{123},\hat{m}_{13},\hat{m}_{23})$ do not encode more messages than those already encoded in $u_{12}^n$, $u_{13}^n$, and $u_{23}^n$. As mentioned earlier, for every pair of codewords $u_{12}^n(\hat{m}_{123},\hat{m}_{12})$ and $u_{23}^n(\hat{m}_{123},\hat{m}_{23})$, we generate a {\em single} codeword $u_{2}^n(\hat{m}_{123},\hat{m}_{12},\hat{m}_{23})$ according to $\prod_{i=1}^n p(u_{2i}|u_{12i},u_{23i},u_{123i})$. Similarly, for every pair of codewords $u_{13}^n(\hat{m}_{123},\hat{m}_{13})$ and $u_{23}^n(\hat{m}_{123},\hat{m}_{23})$, we generate a {\em single} codeword $u_{3}^n(\hat{m}_{123},\hat{m}_{13},\hat{m}_{23})$ according to $\prod_{i=1}^n p(u_{3i}|u_{13i},u_{23i},u_{123i})$. However, since $\hat{R}_1\neq 0$, we generate $2^{n\hat{R}_1}$ codewords $u_{1}^n(\hat{m}_{123},\hat{m}_{13},\hat{m}_{12},\hat{m}_{1})$ for every pair of codewords $u_{12}^n(\hat{m}_{123},\hat{m}_{12})$ and $u_{13}^n(\hat{m}_{123},\hat{m}_{13})$.   Hence, in general, in the coding scheme of Theorem \ref{TH_genral_inner_General_msgs}, superposition coding is not only used to encode a message over other messages (satellites over cloud centers), but also to encode some messages multiple times using different distributions. This novel feature of generating a single satellite per one or more cloud centers will be present in general as long as we choose $\mathsf{F}$ such that $\mathsf{F}  \supset \uparrow_{\mathsf{P}} \mathsf{E}$.

\end{remark}

\subsection{Explicit polyhedral representation for the inner bound with $\mathsf{E}=\{\overline{K},\overline{K-1}\}$ }

For this case, we have $P{=}K-2$, $L_1{=}1$ and $L_2{=}1$. Hence, $S_p {=} \{1, \cdots , P\}, S_{l_1} {=}\{K-1\}, S_{l_2}{=} \{K\}$. We use Theorem \ref{TH_genral_inner_General_msgs} to get a polyhedral description of the inner bound by eliminating the split rates. Here, the message $M_{\overline{K}}$ is split into two parts via \eqref{Eq_RS_upset_splittin}, i.e., $M_{\overline{K}} {=} (M_{\overline{K} \rightarrow \overline{K}},M_{\overline{K}\rightarrow \overline{\phi}})$ while the other message $M_{\overline{K-1}}$ is split into $M_{\overline{K-1} \rightarrow \overline{K-1}}$ and $M_{\overline{K-1}\rightarrow \overline{\phi}}$. The polyhedral representation is presented in the next corollary.

\begin{corollary}
\label{Colollary_general_indirect_C1_1_C2_1}
An inner bound of $K$-user DM BC for the message index set $\mathsf{E}=\{\overline{K},\overline{K-1}\}$ is the set of rate pairs ($R_{\overline{K}},R_{\overline{K-1}}$) satisfying
\begin{align}
R_{\overline{K-1}} &\leq I(U_{\mathsf{W}_K^{\mathsf{P}}};Y_K|Q) \label{Eq_Corollary_indirect_C1_1_C2_1_1}\\
R_{\overline{K}} &\leq I(U_{\mathsf{W}_{K-1}^{\mathsf{P}}};Y_{K-1}|Q) \label{Eq_Corollary_indirect_C1_1_C2_1_11}\\
R_{\overline{K-1}}+R_{\overline{K}} &\leq  I(U_{\mathsf{W}_j^{\mathsf{P}}};Y_j|Q) \enspace \forall j\in  S_p  \label{Eq_Corollary_indirect_C1_1_C2_1_2}\\
R_{\overline{K-1}}+R_{\overline{K}}  &\leq   I(U_{  \downarrow_{\mathsf{W}_j^{\mathsf{P}}} \{\overline{K}\} };Y_j|U_{\mathsf{W}_j^{\mathsf{P}} \backslash  \downarrow_{\mathsf{W}_j^{\mathsf{P}}} \{\overline{K}\}  },Q)  + I(U_{\mathsf{W}_{K}^{\mathsf{P}}};Y_{K}|Q) \quad  \forall j\in S_p \cup \{ K-1\} \label{Eq_Corollary_indirect_C1_1_C2_1_3} \\
R_{\overline{K-1}}+R_{\overline{K}} &\leq I(U_{  \downarrow_{\mathsf{W}_j^{\mathsf{P}}} \{\overline{K-1}\} };Y_j|U_{\mathsf{W}_j^{\mathsf{P}} \backslash  \downarrow_{\mathsf{W}_j^{\mathsf{P}}} \{\overline{K-1}\}  },Q) + I(U_{\mathsf{W}_{K-1}^{\mathsf{P}}};Y_{K-1}|Q) \quad  \forall j\in S_p \cup \{ K\} \label{Eq_Corollary_indirect_C1_1_C2_1_4} \\
2R_{\overline{K-1}}+2R_{\overline{K}}&\leq I(U_{\downarrow_{\mathsf{W}_j^{\mathsf{P}}} \{\overline{K-1},\overline{K}\} };Y_j|U_{\mathsf{W}_j^{\mathsf{P}}\backslash \downarrow_{\mathsf{W}_j^{\mathsf{P}}} \{\overline{K-1},\overline{K}\} },Q)+I(U_{\mathsf{W}_{K}^{\mathsf{P}}};Y_{K}|Q)+ I(U_{\mathsf{W}_{K-1}^{\mathsf{P}}};Y_{K-1}|Q)\enspace \forall j\in  S_p \label{Eq_Corollary_indirect_C1_1_C2_1_5} 
\end{align} 
for some $p(q, u_{\mathsf{P}})= p(q) \prod_{S\in \mathsf{P}}p(u_S|u_{\uparrow_{\mathsf{P}} S\backslash \{ S \} },q)$ and $X$ as a deterministic function of $(Q, U_{\mathsf{P}})$.
\end{corollary}
\begin{proof}
The proof begins with the result of Theorem \ref{TH_genral_inner_General_msgs} by setting $\mathsf{F}{=}\mathsf{P}$. In other words, the rate region of Corollary  \ref{Colollary_general_indirect_C1_1_C2_1} is ${\cal R}(K,\mathsf{E}=\{\overline{K},\overline{K-1}\},\mathsf{F}{=}\mathsf{P})  $ expressed in explicit form after having projected split rates away.
Since we have $L_1{=}1$ and $L_2{=}1$, we have from \eqref{Eq_RS_upset_splittin} that each message is split only into two parts so that
\begin{align}
&R_{\overline{K-1}}=R_{\overline{K-1}\rightarrow \overline{K-1}}+R_{\overline{K-1}\rightarrow \overline{\phi}}\\
&R_{\overline{K}}=R_{\overline{K}\rightarrow \overline{K}}+R_{\overline{K}\rightarrow \overline{\phi}}
\end{align}
Moreover, from \eqref{Eq_reconstruction_rates_upset_splitting},
we have the three non-zero reconstruction rates  $\hat{R}_{\overline{\phi}},\hat{R}_{\overline{K}},\hat{R}_{\overline{K-1}}$, such that
\begin{align}
\hat{R}_{\overline{\phi}}&=R_{\overline{K}\rightarrow \overline{\phi}}+R_{\overline{K-1}\rightarrow \overline{\phi}}\label{Eq_General_new_construction1}\\
\hat{R}_{\overline{K}}&=R_{\overline{K}\rightarrow \overline{K}}\label{Eq_General_new_construction2}\\
\hat{R}_{\overline{K-1}}&=R_{\overline{K-1}\rightarrow \overline{K-1}}\label{Eq_General_new_construction3}
\end{align}

Hence, we can write \eqref{Eq_condition_reconstruction_rates1} as follow
\begin{align}
&R_{\overline{K-1}}+R_{\overline{K}}\leq I(U_{\mathsf{W}_j^{\mathsf{P}} };Y_j|Q)\label{Eq_general_case_C1_1_C2_1_1}\\
&R_{\overline{K-1}\rightarrow \overline{K-1}}\leq  I(U_{\downarrow_{\mathsf{W}_j^{\mathsf{P}}} \{\overline{K-1}\} };Y_j|U_{\mathsf{W}_j^{\mathsf{P}} \backslash \downarrow_{\mathsf{W}_j^{\mathsf{P}}} \{\overline{K-1}\}  },Q)\label{Eq_general_case_C1_1_C2_1_2}\\
&R_{\overline{K}\rightarrow \overline{K}}\leq I(U_{\downarrow_{\mathsf{W}_j^{\mathsf{P}}} \{\overline{K}\} };Y_j|U_{\mathsf{W}_j^{\mathsf{P}} \backslash \downarrow_{\mathsf{W}_j^{\mathsf{P}}} \{\overline{K}\}  },Q)\label{Eq_general_case_C1_1_C2_1_3}\\
&R_{\overline{K-1}\rightarrow \overline{K-1}}+R_{\overline{K} \rightarrow \overline{K}}\leq \nonumber \\ & \quad I(U_{\downarrow_{\mathsf{W}_j^{\mathsf{P}}} \{\overline{K-1},\overline{K}\} };Y_j|U_{\mathsf{W}_j^{\mathsf{P}} \backslash \downarrow_{\mathsf{W}_j^{\mathsf{P}}} \{\overline{K-1},\overline{K}\}  },Q)\label{Eq_general_case_C1_1_C2_1_4}
\end{align} for all $j\in S_p$. This follows from the fact we have just three non-zero reconstruction rates, and hence, $\mathsf{B}\in\{   \mathsf{W}_j^{\mathsf{P}}, \downarrow_{\mathsf{W}_j^{\mathsf{P}}} \{\overline{K-1}\}, \downarrow_{\mathsf{W}_j^{\mathsf{P}}} \{\overline{K}\},\downarrow_{\mathsf{W}_j^{\mathsf{P}}} \{\overline{K-1},\overline{K}\}      \}$, because the rest of the down-sets in $\mathcal{F}_{\downarrow}(\mathsf{W}_j^{\mathsf{P}})$ give redundant inequalities. 



On the other hand, we can rewrite \eqref{Eq_condition_reconstruction_rates2} as follows
\begin{align}
&R_{\overline{K}}+R_{\overline{K-1}\rightarrow \overline{\phi}}\leq I(U_{\mathsf{W}_{K-1}^{\mathsf{P}}};Y_{K-1}|Q)\label{Eq_general_case_C1_1_C2_1_5}\\
&R_{\overline{K}\rightarrow \overline{K}}\leq I(U_{\downarrow_{\mathsf{W}_{K-1}^{\mathsf{P}}} \{\overline{K}\}};Y_{K-1}|U_{\mathsf{W}_{K-1}^{\mathsf{P}}  \backslash \downarrow_{\mathsf{W}_{K-1}^{\mathsf{P}}} \{\overline{K}\} } , Q)\label{Eq_general_case_C1_1_C2_1_6}\\
&R_{\overline{K-1}}+R_{\overline{K}\rightarrow \overline{\phi}}\leq I(U_{\mathsf{W}_K^{\mathsf{P}}};Y_K|Q)\label{Eq_general_case_C1_1_C2_1_7}\\
&R_{\overline{K-1}\rightarrow \overline{K-1}}\leq I(U_{\downarrow_{\mathsf{W}_K^{\mathsf{P}}} \{\overline{K-1}\}};Y_K|U_{\mathsf{W}_K^{\mathsf{P}}  \backslash \downarrow_{\mathsf{W}_K^{\mathsf{P}}} \{\overline{K-1}\} },Q )\label{Eq_general_case_C1_1_C2_1_8}
\end{align} 
Moreover, in \eqref{Eq_condition_reconstruction_rates2} the only sets in $\mathcal{F}_{\downarrow_{ \{\overline{K}\} }  }(\mathsf{W}_j^{\mathsf{P}})$ and $\mathcal{F}_{\downarrow_{ \{\overline{K-1}\} }  }(\mathsf{W}_j^{\mathsf{P}})$ that do not give redundant inequalities are \{$\mathsf{W_j}^{\mathsf{P}},\downarrow_{\mathsf{W}_j^{\mathsf{P}}} \mathsf{W}_j^{\mathsf{E}}$ for $j\in\{K-1,K\}\}$ again because we have just three non-zero reconstruction rates where $\mathsf{W}_j^{\mathsf{E}}=\{\overline{K}\}$ for $j=K-1$ and $\mathsf{W}_j^{\mathsf{E}}=\{\overline{K-1}\}$ for $j=K$. By eliminating the sub-rates from \eqref{Eq_general_case_C1_1_C2_1_1}-\eqref{Eq_general_case_C1_1_C2_1_4} and \eqref{Eq_general_case_C1_1_C2_1_5}-\eqref{Eq_general_case_C1_1_C2_1_8} using FME \cite{schrijver1986theory}, we get the polyhedral description \eqref{Eq_Corollary_indirect_C1_1_C2_1_1}-\eqref{Eq_Corollary_indirect_C1_1_C2_1_5}. 
\end{proof}

\subsection{Converses for $\mathsf{E}=\{\overline{K},\overline{K-1}\}$ and more generally for $\mathsf{E}=\{\overline{K}, S\}$ for any $S$ with $K \in S$}

In what follows, we show that the achievable rate region of Corollary \ref{Colollary_general_indirect_C1_1_C2_1} is optimal for certain classes of partially ordered DM BCs with the partial order denoting the relative strengths of channels to the receivers in terms of the well known ordering relations of one receiver being less noisy or more capable than another. We state here those definitions (including two for the more capable order) for completeness.
\begin{definition}  \label{Def_Less_Noisy} \cite[Definition 2]{korner1975source}
Receiver $Y_s$ is less noisy than $Y_c$ if $I(U;Y_s){\geq} I(U;Y_c)$ for all $p(u,x)$. Henceforth, we denote this condition as $Y_s  \succeq Y_c$.
\end{definition}
\begin{definition} \label{Def_More_Capable_Original}
Receiver $Y_s$ is more capable than $Y_c$ if for every $0< \epsilon < 1$ and $\delta>0$ there exists an $n_0$, a function of $\epsilon$ and $\delta$, such that for $n>n_0$ every $\epsilon$-code $\mathcal{B}\subseteq \mathcal{X}^n$ for the channel $Y_c$ contains an $\epsilon$-code $\mathcal{B}^{'}$ for the channel $Y_s$ such that $\frac{1}{n}\log(|\mathcal{B}^{'}|)>\frac{1}{n}\log(|\mathcal{B}|)-\delta$. 
\end{definition}
This is essentially equivalent to saying that $Y_s$ could decode any codebook that $Y_c$ could decode. Korner and Marton showed that the above definition is equivalent to following one.
\begin{definition} \cite[Definition 3]{korner1975source} \label{Def_More_Capable}
Receiver $Y_s$ is more capable than $Y_c$ if $I(X;Y_s) \geq I(X;Y_c)$ for all $p(x)$. Henceforth, we denote this condition as $ Y_s  \sqsupseteq Y_c $\footnote{This notation can, but should not be, mistaken for ``superset" (which, conveniently, is meaningless in the present context).}. 
\end{definition}

\begin{theorem}
\label{Th-Capacity_Region_Two_K-1_Order_msgs}
The capacity region of $K$-user DM BC for the message index set $\mathsf{E}=\{\overline{K},\overline{K-1}\}$, i.e., $S_p=\{1,2,\dots ,K-2 \}$ is the set of rate pairs ($R_{\overline{K}},R_{\overline{K-1}}$) satisfying
\begin{align}
    R_{\overline{K-1}}&\leq I(U;Y_K)\label{Eq:Capacity_Region_Two_K-1_Order_msgs1_1}\\
    R_{\overline{K-1}}+R_{\overline{K}}&\leq I(X;Y_j|U)+I(U;Y_K) \; \forall j\in S_p \cup \{K-1\}\label{Eq:Capacity_Region_Two_K-1_Order_msgs1_3}\\
    R_{\overline{K-1}}+R_{\overline{K}}&\leq I(X;Y_{K-1})\label{Eq:Capacity_Region_Two_K-1_Order_msgs1_4}
\end{align} for some $p(u,x)$ for each of the following two classes of partially ordered DM BCs
\begin{itemize}
    \item[A.] DM BCs for which $Y_j  \sqsupseteq  Y_{K-1} \sqsupseteq Y_K $ for all $j\in S_p $
    \item[B.] DM BCs for which $Y_i \succeq Y_K $ for all $i\in S_p \cup \{K-1\}$
\end{itemize}
\end{theorem} 
\begin{proof}
The achievability proof of \eqref{Eq:Capacity_Region_Two_K-1_Order_msgs1_1}-\eqref{Eq:Capacity_Region_Two_K-1_Order_msgs1_4}  follow from Corollary \ref{Colollary_general_indirect_C1_1_C2_1} by setting $U_{\overline{\phi}}=U, U_{\overline{K}}=X $ while setting the rest of the auxiliary random variables to be constants and $|Q|=1$. In particular, the inequalities in \eqref{Eq:Capacity_Region_Two_K-1_Order_msgs1_1}-\eqref{Eq:Capacity_Region_Two_K-1_Order_msgs1_4} are obtained from \eqref{Eq_Corollary_indirect_C1_1_C2_1_1}, \eqref{Eq_Corollary_indirect_C1_1_C2_1_3}, and \eqref{Eq_Corollary_indirect_C1_1_C2_1_4}, respectively. Note that from \eqref{Eq_Corollary_indirect_C1_1_C2_1_2}, we get the inequality $R_{\overline{K-1}}+R_{\overline{K}}\leq I(X;Y_j) \; \forall j\in S_p$ which is redundant for both classes A and B of DM BCs defined in the statement of Theorem \ref{Th-Capacity_Region_Two_K-1_Order_msgs}. This concludes the achievbility proof.

For Class A of DM BCs, that is, when $Y_j  \sqsupseteq  Y_{K-1} \sqsupseteq Y_K $ for all $j\in S_p $, the region in \eqref{Eq:Capacity_Region_Two_K-1_Order_msgs1_1}-\eqref{Eq:Capacity_Region_Two_K-1_Order_msgs1_4} can be written as 
\begin{align}
    R_{\overline{K-1}}&\leq I(U;Y_K)\label{Eq:Capacity_Region_Two_K-1_Order_msgs1_1_Class1}\\
    R_{\overline{K-1}}+R_{\overline{K}}&\leq I(X;Y_{K-1}|U)+I(U;Y_K) \label{Eq:Capacity_Region_Two_K-1_Order_msgs1_2_Class1}\\
    R_{\overline{K-1}}+R_{\overline{K}}&\leq I(X;Y_{K-1})\label{Eq:Capacity_Region_Two_K-1_Order_msgs1_3_Class1}
\end{align} 
for some $p(u,x)$. There are two ways to prove the converse for the above inequalities. One is to use the Csiszar sum lemma \cite[Lemma 7]{csiszar1978broadcast} with the optimal choice of $U_i= M_{\overline{K-1}},Y_{K-1,1}^{i-1},Y_{K,i+1}^n$. This proof is given in detail in Appendix \ref{Appendix_Proof_Th_Two_K-1_Order_msgs}. The other is to notice that the region \eqref{Eq:Capacity_Region_Two_K-1_Order_msgs1_1_Class1}-\eqref{Eq:Capacity_Region_Two_K-1_Order_msgs1_3_Class1}, is the capacity region of the two receiver ($Y_{K-1},Y_K$) more capable DM BC with two private messages (cf. \cite{el1979capacity}). In other words, it depends only on the channels of the two receivers $Y_{K-1}, Y_K$, although we are considering the $K$-receiver DM BC. Since Receiver $Y_{K-1}$ decodes from $X$ it is forced to decode both messages, and hence, from Definition \ref{Def_More_Capable_Original}, adding any number of receivers that are more capable than $Y_{K-1}$, does not change the capacity region from the two receiver case. This concludes an alternative proof of the converse for the Class A of DM BCs.
  
For Class B of DM BCs in which $Y_i \succeq Y_K $ for all $i\in S_p \cup \{K-1\}$, the region in \eqref{Eq:Capacity_Region_Two_K-1_Order_msgs1_1}-\eqref{Eq:Capacity_Region_Two_K-1_Order_msgs1_4} can be written as 
\begin{align}
    R_{\overline{K-1}}&\leq I(U;Y_K)\label{Eq:Capacity_Region_Two_K-1_Order_msgs1_1_Class2}\\
    R_{\overline{K-1}}+R_{\overline{K}}&\leq I(X;Y_j|U)+I(U;Y_K) \; \forall j\in S_p \cup \{K-1\}\label{Eq:Capacity_Region_Two_K-1_Order_msgs1_2_Class2}
\end{align}
for some $p(u,x)$ where  \eqref{Eq:Capacity_Region_Two_K-1_Order_msgs1_4} are redundant for this class of channels. The converse proof depends on the information inequality in \cite[Lemma 1]{nair2011capacity} and the details are given in Appendix \ref{Appendix_Proof_Th_Two_K-1_Order_msgs}. 
\end{proof}

Evidently, one can exchange the roles of $Y_K$ and $Y_{K-1}$ in Theorem \ref{Th-Capacity_Region_Two_K-1_Order_msgs} to get a result that mirrors it by exchanging $K-1$ and $K$ in the rate region and in the definition of classes of DM BCs.


\begin{remark}
We showed in \cite[Lemma 3]{salman2017capacity} that the region in \eqref{Eq:Capacity_Region_Two_K-1_Order_msgs1_1_Class2}-\eqref{Eq:Capacity_Region_Two_K-1_Order_msgs1_2_Class2} is also the capacity region of the Class B of DM BCs for two {\em nested} messages with a single common receiver, i.e., $\mathsf{E}=\{\overline{\phi}, \overline{K}\}$, where a similar achievability scheme and converse proof were used. Hence, even though we relaxed the decoding requirements for receiver $Y_{K-1}$ by considering the message set $\mathsf{E}=\{\overline{K-1}, \overline{K}\}$ compared to $\mathsf{E}=\{\overline{\phi}, \overline{K}\}$, we can prove the same outer bound for both message sets. Hence, our converse proof in Appendix \ref{Appendix_Proof_Th_Two_K-1_Order_msgs} is stronger than the one in \cite[Lemma 3]{salman2017capacity}. Applied to the case of $K=3$, this implies that the converse proof in Appendix \ref{Appendix_Proof_Th_Two_K-1_Order_msgs} (for $\mathsf{E}=\{12,13\}$) is stronger than the one for $\mathsf{E}=\{12,123\}$ in \cite[Proposition 11]{nair2009capacity} where the Csiszar sum lemma is used, not the information inequality as in Appendix \ref{Appendix_Proof_Th_Two_K-1_Order_msgs}.  From the achievability perspective, there is no loss of optimality and it does not reduce the achievable region to force receiver $Y_{K-1}$ to decode $M_{\overline{K-1}}$ as long as $Y_{K-1}\succeq Y_K$. 

\end{remark}


In the following proposition, we extend the argument in the previous remark to its most general conclusion by considering the message set $\mathsf{E}=\{\overline{K},S\}$ (where $K\in S$), and hence, relaxing the decoding requirements for the receivers $Y_i$ $i\in\overline{S}$ instead of only $Y_{K-1}$, and still getting the same capacity region for the same class of channels.  

\begin{proposition}
\label{Proposition_1Private_1OrderK-1_messages}
The capacity region of $K$-user DM BC with message index set $\mathsf{E}=\{\overline{K},S\}$, where $S\in \mathsf{P},K\in S$ for the class of channels $Y_i \succeq Y_K $ for all $i\in \{1,2,\cdots ,K-1\}$ is the set of rate pairs ($R_{\overline{K}},R_{S}$) satisfying
\begin{align}
    R_{S}&\leq I(U;Y_K)\label{Eq:Capacity_Region_1Private_1OrderK-1_messages_1}\\
    R_{S}+R_{\overline{K}}&\leq I(X;Y_j|U)+I(U;Y_K) \; \forall j\in \{1,2,\cdots,K-1\}\label{Eq:Capacity_Region_1Private_1OrderK-1_messages_2}
\end{align}
for some $p(u,x)$. 
\end{proposition}
\begin{proof}
The proof is left to the reader since it is similar to that of Theorem \ref{Th-Capacity_Region_Two_K-1_Order_msgs}. 
\end{proof}

Consider Proposition \ref{Proposition_1Private_1OrderK-1_messages} in its strongest case, i.e., for $S=K$. Here, the proposition in effect states that forcing the receivers $Y_1,Y_2,\cdots,Y_{K-1}$ to decode $M_K$ is without loss of optimality as long as $Y_i \succeq Y_K $ for all $i\in \{1,2,\cdots ,K-1\}$. While from the converse perspective, relaxing the decoding requirement of the receivers $Y_1,Y_2,\cdots,Y_{K-1}$ by considering  $\mathsf{E}=\{\overline{K},K\}$ instead of $\mathsf{E}=\{\overline{K},\overline{\phi}\}$ does not enlarge the outer bound from what it is for $\mathsf{E}=\{\overline{K},\overline{\phi}\}$ when $Y_i \succeq Y_K $ for all $i\in \{1,2,\cdots ,K-1\}$. In fact, the rationale behind this is the following. Unlike the more capable condition, the less noisy condition is a strong order relation. Hence, by considering the class $Y_i \succeq Y_K $ for all $i\in \{1,2,\cdots ,K-1\}$, it is assured that the receivers $Y_1,Y_2,\cdots,Y_{K-1}$ that decode from $X$ can decode whatever message receiver $Y_K$ decodes from $U$ which explains why forcing them to decode the cloud center $U$ is without loss of optimality. 
On the other hand, our converse proof given in \cite[Lemma 3]{salman2017capacity} or in Appendix \ref{Appendix_Proof_Th_Two_K-1_Order_msgs} does not use the fact that the message encoded in $U$ that is required by $Y_K$ might be also required by any subset of the receivers $Y_1,Y_2,\cdots,Y_{K-1}$. Hence, for $\mathsf{E}=\{\overline{\phi},\overline{K}\}$ or $\mathsf{E}=\{K,\overline{K}\}$, we get the same outer bound for the class of channels $Y_i \succeq Y_K $ for all $i\in \{1,2,\cdots ,K-1\}$. 


\section{Achievable Rate Region: Two Nested Messages}
\label{Sec_Two_nested}

In this section, we focus on the special case of two {\em nested} messages. Let $S_1=\{1, 2, \cdots, P\}$ and $S_2= \{1, 2, \cdots, P+L\}$ with  $P+L{=}K$. Hence, the set of indices of private receivers that demand both messages is $S_p =\{1,2,\cdots,P\}$ and the set of indices of common receivers that demand only the common or multicast message $M_{S_2}$ is $S_l{=}\{P+1, \cdots , P+L\}$. It follows from the up-set rate-splitting technique described in Theorem \ref{TH_genral_inner_General_msgs} that the private message is split into at most $2^L$ parts (depending on the choice of $\mathsf{F}$) while the common message is not split because it is desired by all receivers. 


Next, we present the inner bound in a more explicit way for the two nested message case than in Theorem \ref{TH_genral_inner_General_msgs}. In particular, the inner bound is presented in terms of the original message rates $R_{S}$ (where $S\in\mathsf{E}$) and the split message rates $R_{S\rightarrow S^{'}}$ ($S^{'} \in \mathsf{F}$) instead of the reconstruction rates as in Theorem \ref{TH_genral_inner_General_msgs}. We are able to do this because of the simplicity afforded by the structure of the nested message set.

\begin{theorem}
\label{Th_Inner_Nested_Bound_indirect_F}
Let $\mathsf{F}$ be any message index superset so that $\mathsf{P}\supseteq \mathsf{F}\supseteq\mathsf{E}$. A rate pair $(R_{S_p}, R_{\overline{\phi}})$ is achievable over the $K$-user DM BC if 
\begin{align}
 R_{\overline{\phi}}+R_{S_p} & \leq  I(U_{\mathsf{W}_j^{\mathsf{F}}};Y_j|Q) \enspace \forall j\in S_p \label{Eq_TH_indirect_generalF_private1}\\
 \sum_{S\in \mathsf{B}}  R_{S_p\rightarrow S} 
 &\leq I(U_{\mathsf{B}}; Y_j|U_{\mathsf{W}_j^{\mathsf{F}}\backslash \mathsf{B}   },Q)  \;\quad \mathsf{B}\in \mathcal{F}_{\downarrow}(\mathsf{W}_j^{\mathsf{F} }\backslash \overline{\phi}), \; \forall j\in S_p \label{Eq_TH_indirect_generalF_private2}\\
 R_{\overline{\phi}}+\sum_{S\in \uparrow_{\mathsf{F}} \{ 1 2 \cdots P.i\}} R_{S_p\rightarrow S}
 &\leq I(U_{\mathsf{W}_i^{\mathsf{F}}};Y_i|Q)  \enspace \forall i\in S_l \label{Eq_TH_indirect_generalF_common}
\end{align}
for some $p(q,u_{\mathsf{F}})=p(q) \prod_{S\in \mathsf{F}}p(u_S|u_{\uparrow_{\mathsf{F}} S\backslash    \{S\}     },q)$ and some deterministic function $X$ of the time sharing and auxiliary random variables $(Q,U_{\mathsf{F}})$.
\end{theorem}
\begin{proof}
The proof of this theorem follows from Theorem \ref{TH_genral_inner_General_msgs} when we specialize it to two nested messages. First,  note that since the private message $M_{S_p}$ is split according to \eqref{Eq_RS_upset_splittin}, we can simplify the reconstructed message rates in \eqref{Eq_reconstruction_rates_upset_splitting} as follows
\begin{align}
&\hat{R}_{S}= R_{S_p\rightarrow S} \enspace \forall S\in \uparrow_{\mathsf{F}} S_p\backslash \overline{\phi} \label{Eq_reconstruction_rates_nested1}\\
&\hat{R}_{\overline{\phi}}=R_{\overline{\phi}}+R_{S_p \rightarrow \overline{\phi}} \\
& \hat{R}_{S} = 0 \quad S\in \mathsf{F}\backslash \uparrow_{\mathsf{F}} S_p
\label{Eq_reconstruction_rates_nested2}
\end{align} 
Next, consider the set of inequalities in \eqref{Eq_condition_reconstruction_rates1}. Note that $\mathcal{F}_{\downarrow}(\mathsf{W}_j^{\mathsf{F} })$ for any $j$ can be written as $\mathcal{F}_{\downarrow}(\mathsf{W}_j^{\mathsf{F} }\backslash \overline{\phi}) \cup \mathsf{W}_j^{\mathsf{F} } $ because the only down-set that contains $\overline{\phi}$ must be $\mathsf{W}_j^{\mathsf{F} } $. 
Hence, the inequality in \eqref{Eq_condition_reconstruction_rates1} simplifies to \eqref{Eq_TH_indirect_generalF_private1} and \eqref{Eq_TH_indirect_generalF_private2}. On the other hand, for all common receivers $Y_i$ ($i\in S_l$), we have $\mathcal{F}_{\downarrow_{ \overline{\phi} }  }(\mathsf{W}_i^{\mathsf{F}})= \mathsf{W}_i^{\mathsf{F}}$. However, since the only non-zero reconstruction rates are indexed by $\uparrow_{\mathsf{F}}S_p$, we can replace $\mathcal{F}_{\downarrow_{   \overline{\phi} }  }(\mathsf{W}_i^{\mathsf{F}})$ in \eqref{Eq_condition_reconstruction_rates2} by $\mathsf{W}_i^{\mathsf{F}} \cap \uparrow_{\mathsf{F}}S_p = \uparrow_{\mathsf{F}}\{12 \cdots P.i\}$ which yields \eqref{Eq_TH_indirect_generalF_common}. This completes the proof of Theorem \ref{Th_Inner_Nested_Bound_indirect_F}.   
\end{proof}

\subsection{Recovering Prior Results from Theorem \ref{Th_Inner_Nested_Bound_indirect_F}}

Theorem \ref{Th_Inner_Nested_Bound_indirect_F} can be seen as a generalization of previously proposed achievable schemes in the context of particular examples. These include (a) the two-receiver DM-BC with degraded messages in which superposition coding alone (without rate-splitting) is sufficient to achieve the capacity region as was shown in the important work of Korner and Marton in \cite{kiirner1977general} (b) the three-receiver DM-BC with two nested messages with one and two common receivers which was investigated in-depth in \cite{nair2009capacity} and (c) the DM-BC with two nested messages for an arbitrary number of receivers and with an arbitrary number of common receivers studied by the authors in \cite{salman2017capacity}. In particular, the achievable rate regions based on rate-splitting and superposition coding obtained in the aforementioned papers can be obtained by specific choices of $\mathsf{F}$. Those rate regions were shown in those works to be the capacity regions for certain classes of channels. In this section, we briefly describe those regions from the lens of Theorem \ref{Th_Inner_Nested_Bound_indirect_F} and also specify the conditions on the DM-BC for which they yield the capacity region. This allows us to not only place known results in the general framework of Theorem \ref{Th_Inner_Nested_Bound_indirect_F} but also to show how further improvements may be possible in DM-BCs for which the capacity region is not yet known.



\begin{example}
\label{Example_K_2_P_1}
Consider the case of $K=2$ and two degraded messages so that $\mathsf{E}=\{1,12\}$. Let us choose $\mathsf{F}=\mathsf{E}$ and $X=U_1$ and $Q$ to be an uninformative constant in Theorem \ref{Th_Inner_Nested_Bound_indirect_F}. We get that that a rate pair $(R_{1}, R_{12})$ is achievable if there exist non-negative split rates with $R_{1}= R_{1 \rightarrow 1} + R_{1\rightarrow 12} $ such that the following inequalities hold
\begin{align}
R_{12}+R_1\leq &I(X;Y_1)\nonumber  
\\
R_{1 \rightarrow 1}\leq &I(X;Y_1|U_{12}) \nonumber \\
R_{12}+R_{1\rightarrow 12}\leq &I(U_{12};Y_2)\nonumber 
\end{align}
for some $p(u_{12})p(x|u_{12})$. It was shown that the above region is the capacity region for $K=2$ with two nested messages in \cite{kiirner1977general} without any channel restrictions. Note that even rate splitting is not necessary in this case.
\end{example}

\begin{example}
\label{Example_K_3_P_2}
Consider the three-receiver case with the message index set $\mathsf{E}=\{12,123\}$. We choose $\mathsf{F}=\mathsf{E}$ and $X=U_{12}$ and $Q={\rm const}$ in Theorem \ref{Th_Inner_Nested_Bound_indirect_F} which then implies that a rate pair $(R_{12}, R_{123})$ is achievable if there exist non-negative split rates with $R_{12}= R_{12 \rightarrow 12} + R_{12\rightarrow 123} $ such that the following inequalities hold
\begin{align*}
R_{123}+R_{12} & \leq \min\{I(X;Y_1),I(X;Y_2)\} \\
R_{12 \rightarrow 12} & \leq \min\{I(X;Y_1|U_{123}), I(X;Y_2|U_{123})\} \\
R_{123}+R_{12\rightarrow 123} & \leq I(U_{123};Y_3)
\end{align*} 
for some $p(u_{123})p(x|u_{123})$. Projecting away the split rates yields that rate pairs for which
\begin{align*}
R_{123}+R_{12} & \leq \min\{I(X;Y_1),I(X;Y_2)\} \\
R_{123}+ R_{12} & \leq \min\{I(X;Y_1|U_{123}), I(X;Y_2|U_{123})\}+ I(U_{123};Y_3) \\
R_{123} &\leq I(U_{123};Y_3)
\end{align*}
for some $p(u_{123})p(x|u_{123})$ are achievable. This region is identical to that given in \cite[Corollary 1]{nair2009capacity} (as it should be, since our general scheme reduces to that of \cite[Corollary 1]{nair2009capacity} with the choices $\mathsf{F}=\mathsf{E}$ and $X=U_{12}$ and $Q={\rm const} $). Notably, the above region was shown to be the capacity region when $Y_1\succeq Y_3$ and $Y_2\succeq Y_3$ in \cite[Proposition 11]{nair2009capacity} \footnote{There is a typo in the statement of Proposition 11 of \cite{nair2009capacity}. The conditions stated therein $Y_1 {\succeq}  Y_2 $ and $Y_3 {\succeq}  Y_2 $ should be $Y_2{\succeq} Y_3 $ and $Y_1{\succeq}  Y_3$. The correct conditions are stated in the proof of the proposition.}. 
\end{example}

\begin{example}
\label{Example:E_1,123_F_1,13,123}
Here, we 
consider $K=3$ and $\mathsf{E}=\{1,123\}$. We choose $\mathsf{F}=\{1,13,123\}$ and $X=U_1$ and $Q={\rm const}$. From Theorem \ref{Th_Inner_Nested_Bound_indirect_F}, we obtain inequalities that involve split rates. Using FME to eliminate the split rates we have that rate pairs $(R_2, R_{123})$ are achievable provided 
\begin{align*}
R_{123} &\leq \min\{I(U_{13};Y_3),I(U_{123};Y_2)\} \\ 
R_{123}+R_1 &\leq \min\{  I(X;Y_1), I(X;Y_1|U_{13})+I(U_{13};Y_3), I(X;Y_1|U_{123})+I(U_{123};Y_2)\}
\end{align*}
for some $p(u_{123})p(u_{13}|u_{123})p(x|u_{13})$.
This region is identical to that given in \cite[Remark 5.2]{nair2009capacity} since our scheme for this example and choice of $\mathsf{F}=\{1,13,123\}$ and $X=U_1$ becomes that of \cite[Remark 5.2]{nair2009capacity}. Also, in \cite[Proposition 7]{nair2009capacity}, this region was shown to be capacity region when $Y_1\succeq Y_2$.  
\end{example}

\begin{example}
\label{Example_K_4_P_1}
Consider the four-receiver case with the message index set $\mathsf{E}=\{1,1234\}$.  By choosing $\mathsf{F}=\{1,12,123,1234\}$ and $Q={\rm const}$, we get from Theorem \ref{Th_Inner_Nested_Bound_indirect_F} that the set of rates satisfying
\begin{align*}
R_{1234}+R_1&\leq I(X;Y_1)\\ 
R_{1\rightarrow 1}& \leq I(X;Y_1|U_{12})\\
R_{1\rightarrow 1}+R_{1\rightarrow 12}&\leq I(X;Y_1|U_{123})\\
R_{1\rightarrow 123}+R_{1\rightarrow 12}+R_{1\rightarrow 1}&\leq I(X;Y_1|U_{1234})\\
R_{123}+R_{1\rightarrow 12}+R_{1\rightarrow 123}+R_{1\rightarrow 1234}&\leq I(U_{12};Y_2) \\
R_{123}+R_{1\rightarrow 123}+R_{1\rightarrow 1234}&\leq I(U_{123};Y_3)\\
R_{123}+R_{1\rightarrow 1234}&\leq I(U_{1234};Y_4) 
\end{align*}
for some $p(u_{1234}) p(u_{123}|u_{1234}) p(u_{12}|u_{123})p(x|u_{12})$ is achievable where $R_1=R_{1\rightarrow 1}+R_{1\rightarrow 12}+R_{1\rightarrow 123}+R_{1\rightarrow 1234}$ and $X=U_1$. By applying the FME procedure  to eliminate the sub-rates, the obtained polygon is the capacity region for the class of channels $Y_1 \succeq Y_3$ and $Y_1 \succeq Y_4$ from the result in \cite[Theorems 2 and 3]{salman2017capacity} for $K=4$. 
\end{example}

In fact, if we choose $\mathsf{F}=\{1,12,123,1234,\cdots ,1234 \cdots K\}$, we recover the result in \cite[Thoerem 2]{salman2017capacity} for any number of receivers $K$ with any number of common and private receivers. Hence, the result in \cite[Thoerem 2]{salman2017capacity} is a special case of the result in Theorem \ref{Th_Inner_Nested_Bound_indirect_F} obtained by a specific choice of $\mathsf{F}$ but unlike in Theorem \ref{Th_Inner_Nested_Bound_indirect_F} it has a polyhedral description in terms of the two message rates \cite{salman2017capacity}.

\subsection{New rate regions from Theorem \ref{Th_Inner_Nested_Bound_indirect_F} for $K=3$ and $\mathsf{E}=\{1,123\}$}

In the previous section we saw that previous rate regions of \cite{kiirner1977general}, \cite{nair2009capacity}, and \cite{salman2017capacity} can be obtained from Theorem \ref{Th_Inner_Nested_Bound_indirect_F} by choosing $\mathsf{F}$ accordingly. 
It this section, we show that Theorem \ref{Th_Inner_Nested_Bound_indirect_F} can also produce new rate regions that are potentially larger.
As stated in Remark \ref{rem-largeF}, the choice of $\mathsf{F} \supseteq \mathsf{E}$ determines the achievable region and expanding $\mathsf{F}$ cannot reduce the rate region. 
Moreover, choosing  $\mathsf{F} \supset \uparrow_{\mathsf{P}}\mathsf{E}$ yields some zero reconstruction rates in \eqref{Eq_reconstruction_rates_upset_splitting}, or equivalently, multiple codebooks for certain groups of sub-messages. 


In particular, Theorem \ref{Th_Inner_Nested_Bound_indirect_F} can produce potentially larger inner bounds
for the message set $\mathsf{E}=\{12,123\}$ than that in Example \ref{Example_K_3_P_2} (i.e., that in \cite[Corollary 1]{nair2009capacity}) 
and for the message set $\mathsf{E}=\{1,123\}$ than in Example \ref{Example:E_1,123_F_1,13,123} (i.e., that of
\cite[Remark 5.2]{nair2009capacity}). We show this in the latter case in the next example.
\begin{example}
\label{Example:E_1,123_F_1,13,23,123}
Suppose $K=3$ and $\mathsf{E}=\{1,123\}$. We choose $\mathsf{F}=\{1,13,23,123\} $ and $Q={\rm const}$. From Theorem \ref{Th_Inner_Nested_Bound_indirect_F}, we can show, after projecting away the split rates using FME,  that the rate tuples that satisfy the inequalities
\begin{align}
R_{123}&\leq \min\{I(U_{13},U_{23};Y_3), I(U_{23},U_{123};Y_2)\} \label{Eq:E_1,123_F_1,13,23,123_1}\\
R_{123}+R_1&\leq \min \{ I(U_1,U_{13},U_{123};Y_1), \label{Eq:E_1,123_F_1,13,23,123_2}
\\
& \quad \quad \quad
I(U_1;Y_1|U_{13},U_{123})+I(U_{13},U_{23};Y_3),\label{Eq:E_1,123_F_1,13,23,123_3}\\
& \quad \quad \quad I(U_1,U_{13};Y_1|U_{123})+I(U_{23},U_{123};Y_2) \} \label{Eq:E_1,123_F_1,13,23,123_4}
\end{align}
for some $p(u_{123})p(u_{23}|u_{123})p(u_{13}|u_{123})p(u_1|u_{13},u_{123})$, and with $X$ is deterministic function of $U_{\mathsf{F}}$, are achievable. Note that this region subsumes that of Example \ref{Example:E_1,123_F_1,13,123} since by setting $U_1=X$ and $U_{23}=\textit{const.}$, i.e., $U_{123}\markov U_{13}\markov U_{1}=X$, the above region reduces to that in Example \ref{Example:E_1,123_F_1,13,123}. In other words, choosing $X$ more generally as a deterministic function of $U_{\mathsf{F}}$ and taking the union over all admissible distribution of $U_{23}$ can only enlarge the inner bound compared to that in Example \ref{Example:E_1,123_F_1,13,123}. In particular, in the rate region of this example $U_{23}$ is an extra degree of freedom that receivers $Y_2,Y_3$ can exploit (cf. the bounds in the inequalities \eqref{Eq:E_1,123_F_1,13,23,123_1}, \eqref{Eq:E_1,123_F_1,13,23,123_3}-\eqref{Eq:E_1,123_F_1,13,23,123_4}) which can possibly enlarge the inner bound.   
\end{example}


\begin{example}
\label{Example:E_1,123_F_1,12,13,123}
Suppose again $K=3$ and $\mathsf{E}=\{1,123\}$. This time choose $\mathsf{F}= \uparrow_{\mathsf{P}}\mathsf{E}=\{1,12,13,123\}$, $X=U_1$ and $Q={\rm const}$. From Theorem \ref{Th_Inner_Nested_Bound_indirect_F}, we can show that the rate pairs that satisfy 
\begin{align}
&R_{123}\leq \min\{I(U_{12},U_{123};Y_{2}),I(U_{13},U_{123};Y_3)\}\label{Eq_Inner_Bound_F_diamond,1}\\
&R_{123}+R_{1} \leq\min\{I(X;Y_1), \nonumber \\ &\qquad \qquad \enspace I(X;Y_1|U_{123},U_{12})+I(U_{12},U_{123};Y_{2}), \nonumber \\ &\qquad \qquad \enspace I(X;Y_1|U_{123},U_{13})+I(U_{13},U_{123};Y_3)\} \label{Eq_Inner_Bound_F_diamond,2}\\
&2R_{123}+R_{1} \leq I(X;Y_1|U_{12},U_{13})+ I(U_{12},U_{123};Y_{2})+I(U_{13},U_{123};Y_3)\label{Eq_Inner_Bound_F_diamond,3}\\
&2R_{123}+2R_{1} \leq I(X;Y_1|U_{12},U_{13})+  I(X;Y_1|U_{123})+
\nonumber \\& \qquad \qquad \qquad I(U_{12},U_{123};Y_{2})+I(U_{13},U_{123};Y_3) \label{Eq_Inner_Bound_F_diamond,4}
\end{align} 
for some $p(u_{123})p(u_{12}|u_{123})p(u_{13}|u_{123})p(x|u_{12},u_{13})$ are achievable. It can be shown that the above region is equivalent to that in \cite[Proposition 5]{nair2009capacity} but with no binning. Note that without binning, the second and fifth inequalities in \cite[Proposition 5]{nair2009capacity} become redundant. 
\end{example}


The largest inner bound obtainable from Theorem \ref{Th_Inner_Nested_Bound_indirect_F} is obtained by setting $\mathsf{F}=\mathsf{P}$. In the next example we present the achievable rate region for this choice after FME.

\begin{example}
\label{Example:E_1,123_F_1,2,3,12,13,123}
Consider again $K=3$ and $\mathsf{E}=\{1,123\}$. 
Choosing $\mathsf{F}=\mathsf{P}$ and $Q={\rm const}$ in Theorem \ref{Th_Inner_Nested_Bound_indirect_F}, and after performing FME, it can be shown that the rate pairs that satisfy the inequalities
\begin{align}
&R_{123}\leq \min\{I(U_2,U_{12},U_{23},U_{123};Y_{2}),I(U_3,U_{13},U_{23},U_{123};Y_3)\}\label{Eq_Inner_Bound_F_diamond,1}\\
&R_{123}+R_{1} \leq\min\{I(U_1,U_{12},U_{13},U_{123};Y_1), \nonumber \\ &\qquad \qquad \enspace I(U_1,U_{13};Y_1|U_{123},U_{12})+I(U_2,U_{12},U_{23},U_{123};Y_{2}), \nonumber \\ &\qquad \qquad \enspace I(U_1,U_{12};Y_1|U_{123},U_{13})+I(U_3,U_{13},U_{23},U_{123};Y_3)\} \label{Eq_Inner_Bound_F_diamond,2}\\
&2R_{123}+R_{1} \leq I(U_1;Y_1|U_{12},U_{13})+ I(U_2,U_{12},U_{23},U_{123};Y_{2})+I(U_3,U_{13},U_{23},U_{123};Y_3)\label{Eq_Inner_Bound_F_diamond,3}\\
&2R_{123}+2R_{1} \leq I(U_1;Y_1|U_{12},U_{13})+  I(U_1,U_{12},U_{13};Y_1|U_{123})+
\nonumber \\& \qquad \qquad \qquad I(U_2,U_{12},U_{23},U_{123};Y_{2})+I(U_3,U_{13},U_{23},U_{123};Y_3) \label{Eq_Inner_Bound_F_diamond,4}
\end{align} for some $p(u_{123})p(u_{12}|u_{123})p(u_{13}|u_{123})p(u_{23}|u_{123})p(u_1|u_{12},u_{13},u_{123})p(u_2|u_{12},u_{23},u_{123})p(u_3|u_{13},u_{23},u_{123})$, with $X$ a deterministic function of $U_{\mathsf{P}}$, are achievable. By setting $U_1=X$ and setting $U_2$, $U_3$ and $U_{23}$ to be uninformative constants, the above rate region reduces to that in Example \ref{Example:E_1,123_F_1,12,13,123}. But more generally, we  allow $X$ to be a function of $U_{\mathsf{P}}$ and exploit the additional degrees of freedom $U_2$, $U_3$ and $U_{23}$ provide to the common receivers so as to likely get an enlarged rate region (this can be almost certainly made to be the case per coding distribution).
Next, we give a detailed account of the encoding and decoding functions associated with the above rate region, which we denote as ${\cal R}(K=3, \mathsf{E}=\{1,123\}, \mathsf{F}=\mathsf{P})$.

In the coding scheme of Theorem \ref{Th_Inner_Nested_Bound_indirect_F} for the choice $\mathsf{F}{=}\mathsf{P}$, illustrated in the Hasse diagram of Fig. \ref{Fig_K3_Different_F_P}, we generate codewords associated with all possible auxiliary random variables $U_1,U_2,U_3,U_{12},U_{13},U_{23},U_{123}$ even though $M_1$ is split using \eqref{Eq_RS_upset_splittin} into the four split messages as in Example \ref{Example:E_1,123_F_1,12,13,123}. Note that, as was the case
in Example \ref{Example_P_0_C1_1_C2_2}, the choice of $\mathsf{F}{=}\mathsf{P}$ yields some reconstruction rates equal to {\em zero} 
only this time those rates are $\hat{R}_2$,  $\hat{R}_3$ and $\hat{R}_{23}$. 
Unlike in Example \ref{Example:E_1,123_F_1,12,13,123},
the message pair $(m_{123}, m_{1 \rightarrow 123})$ in the present example is assigned {\em two} codebooks, namely, the $U_{123}$ and $U_{23}$ codebooks (with joint distribution $p(u_{23},u_{123})$), which can also be thought of as the $U_{23}$ codebook being a satellite of the $U_{123}$ codebook but with a single satellite codeword per cloud codeword. Similarly (also unlike in the cases of Example \ref{Example:E_1,123_F_1,12,13,123}) 
the message triple $(m_{123}, m_{1 \rightarrow 123}, m_{1 \rightarrow 13})$ is assigned {\em two} codebooks (the $U_{13}$ and $U_{3}$ codebooks) with $U_{13}$ codebook being a satellite of the cloud codebook $U_{123}$ (which is one of the two codebooks for the pair $(m_{123}, m_{1 \rightarrow 123})$) and the $U_{3}$ codebook being a satellite of the $U_{13}$  codebook but with one satellite codeword in the $U_{3}$ codebook per one codeword in the $U_{13}$ codebook. Moreover, the message triple $(m_{123}, m_{1 \rightarrow 123}, m_{1 \rightarrow 12})$ is also assigned {\em two} codebooks (the $U_{12}$ and $U_{2}$ codebooks) with each being a satellite codebook of one of the two ``cloud" codebooks $U_{123}$ and $U_{23}$ for the message pair $(m_{123}, m_{1 \rightarrow 123})$, where $U_{23}$ is itself a one-satellite-per-cloud satellite of the cloud $U_{123}$. 

As for decoding, private receiver 1 must decode both messages. To do this, it decodes the reconstructed message set $M_{\mathsf{W}_1^\mathsf{P}} = (\hat{M}_1, \hat{M}_{12}, \hat{M}_{13}, \hat{M}_{123})$ which contains all of $(M_1, M_{123})$ (and nothing more) based on unique joint typicality of the $(U_1, U_{12}, U_{13}, U_{123}) $ codeword-tuple with the received sequence $Y_{1,1}^n$. Common Receiver 2, which needs to only decode the common message $M_{123}$, will decode the reconstructed message set $M_{\mathsf{W}_2^\mathsf{P}} = (\hat{M}_2, \hat{M}_{12}, \hat{M}_{23}, \hat{M}_{123})$ by testing for joint typicality of $(U_2, U_{12}, U_{23}, U_{123}) $ codeword-tuple with the received signal $Y_{2,1}^n$, but not uniquely as explained next. Now, note that $\hat{M}_2$ and $\hat{M}_{23}$ are empty messages and $\hat{M}_{12} = M_{1 \rightarrow 12}$ and $\hat{M}_{123} = (M_{1 \rightarrow 123}, M_{123})$. Hence, Receiver 2 can decode $M_{123}$ by jointly decoding $M_{\mathsf{W}_2^\mathsf{P}}$ by testing for joint typicality of the $(U_2, U_{12}, U_{23}, U_{123} )$ codeword-tuple with the received sequence $Y_{2,1}^n$ (thereby making use of the two double codewords) with just unique decoding of $\hat{M}_{123}$ (which is sufficient to decode $M_{123}$ uniquely) and non-unique decoding of $\hat{M}_{12}$ which contains only a part of $M_1$ which Receiver 2 is not required to decode. A similar explanation can be given for non-unique decoding at Receiver 3.
\end{example}

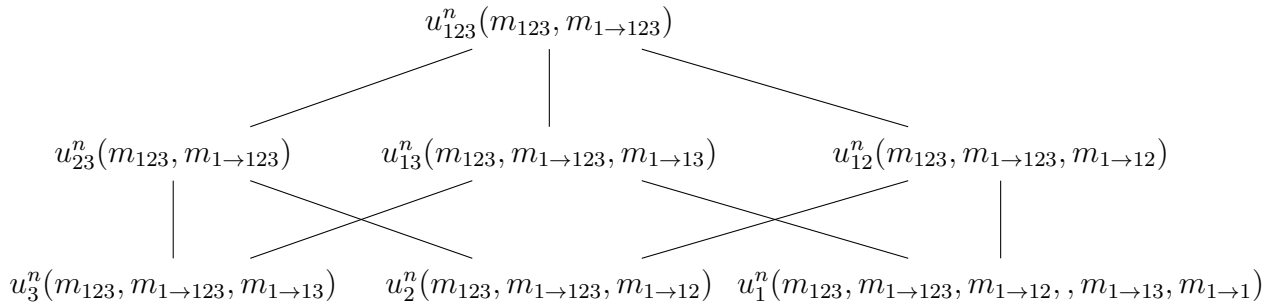
\begin{figure*}
\centering
\begin{tikzpicture}
  \node (u123) at (-1.5,1.75)  {$u_{123}^n(m_{123},m_{1\rightarrow 123})$};
  
  \node (u12) at (4.5,0)  {$  u_{12}^n(m_{123},m_{1\rightarrow 123},m_{1\rightarrow 12})$};
 
  \node (u13) at (-1.5,0)  {$ u_{13}^n(m_{123},m_{1\rightarrow 123},m_{1\rightarrow 13})$};
  
 \node (u23) at (-6.5,0)  {$ u_{23}^n(m_{123},m_{1\rightarrow 123})$};
 
 \node (u3) at (-6.5,-1.75)  {$ u_{3}^n(m_{123},m_{1\rightarrow 123},m_{1\rightarrow 13})$};
 
 \node (u2) at (-1.5,-1.75)  {$u_{2}^n(m_{123},m_{1\rightarrow 123},m_{1\rightarrow 12})$};
 
 \node (u1) at (4.5,-1.75)  {$ u_1^n(m_{123},m_{1\rightarrow 123},m_{1\rightarrow 12},,m_{1\rightarrow 13},m_{1\rightarrow 1})$};

\draw [-] (u123) -- (u12);
\draw [-] (u123) -- (u13);
\draw [-] (u123) -- (u23);
\draw [-] (u1) -- (u12);
\draw [-] (u1) -- (u13);
\draw [-] (u2) -- (u12);
\draw [-] (u2) -- (u23);
\draw [-] (u3) -- (u13);
\draw [-] (u3) -- (u23);
\end{tikzpicture}
\caption{ A Hasse diagram for the coding scheme when message index superset $\mathsf{F}=\mathsf{P}$ for $\mathsf{E}=\{1,123\}$, where the line represents superposition coding. $M_1$ is split into four parts $M_{1\rightarrow 1},M_{1\rightarrow 12},M_{1\rightarrow 13},M_{1\rightarrow 123}$. The encoding and decoding schemes are explained in detail in Example \ref{Example:E_1,123_F_1,2,3,12,13,123}.
\label{Fig_K3_Different_F_P}} 
\end{figure*}

\subsection{Explicit polyhedral representation for the inner bound with one and two common receivers}

It is important to notice that for any two nested messages when we set $\mathsf{F}=\mathsf{P}$ in Theorem \ref{Th_Inner_Nested_Bound_indirect_F} we get $2^K-1$ auxiliary random variables and the private message is split into $2^L$ parts where $L$ is the number of common receivers. Also, we can rewrite \eqref{Eq_TH_indirect_generalF_private2} as follows
\begin{align}
&\sum_{S \in \mathsf{B}}R_{S_p\rightarrow S} \leq I(U_{\downarrow_{\mathsf{W}_j^{\mathsf{P}}}  \mathsf{B}   };Y_j|U_{\mathsf{W}_j^{\mathsf{P}} \backslash \downarrow_{\mathsf{W}_j^{\mathsf{P}}} \mathsf{B}  },Q)   \qquad \forall  \mathsf{B} \in \mathcal{F}_{\downarrow}(\uparrow_{\mathsf{P}}S_p\backslash \overline{\phi}) , j\in  S_p \label{Eq_TH_indirect_private2_extra}
\end{align}

For one common receiver, we have $ S_p =\overline{K}$ and $ S_l =\{K\}$. We specialize Theorem \ref{Th_Inner_Nested_Bound_indirect_F} to get a polyhedral description for the inner bound for this case as follows.
\begin{corollary}
\label{Colollary_nested_indirect_1commonRec}
An inner bound of $K$-user DM BC for two nested messages ($M_{\overline{K}}$,$M_{\overline{\phi}}$) is the set of rate pairs ($R_{\overline{K}},R_{\overline{\phi}}$) satisfying
\begin{align}
R_{\overline{\phi}} &\leq I(U_{\mathsf{W}_K^{\mathsf{P}}};Y_K|Q) \label{Eq_Corollary_indirect_1common_1}\\
R_{\overline{\phi}}+R_{\overline{K}} &\leq  I(U_{\mathsf{W}_j^{\mathsf{P}}};Y_j|Q) \enspace \forall j\in S_p \label{Eq_Corollary_indirect_1common_2}\\
R_{\overline{\phi}}+R_{\overline{K}} &\leq I(U_{  \downarrow_{\mathsf{W}_j^{\mathsf{P}}} \{\overline{K}\} };Y_j|U_{\mathsf{W}_j^{\mathsf{P}} \backslash  \downarrow_{\mathsf{W}_j^{\mathsf{P}}} \{\overline{K}\}  },Q) + I(U_{\mathsf{W}_{K}^{\mathsf{P}}};Y_{K}|Q)  \enspace \forall j\in S_p \label{Eq_Corollary_indirect_1common_3} 
\end{align}
for some $p(q,u_{\mathsf{P}})= p(q) \prod_{S\in \mathsf{P}}p(u_S|u_{\uparrow_{\mathsf{P}} S\backslash \{ S \} },q)$ and $X$ as a deterministic function of $(Q, U_{\mathsf{P}})$.
\end{corollary}
\begin{proof}
The proof follows from Theorem \ref{Th_Inner_Nested_Bound_indirect_F} when $\mathsf{F}=\mathsf{P}$. In this case, we only split the private message $M_{\overline{K}}$ into two sub-messages $M_{\overline{K}\rightarrow \overline{K}}$, $M_{\overline{K}\rightarrow \overline{\phi}}$. By applying FME procedure, we project away the sub-rates and obtain the achievable rate region in the statement of the corollary. 
\end{proof}

On the other hand, for two common receivers case, we have $ S_p =\{1,2,\cdots,K-2\}$ and $ S_l =\{K{-}1,K\}$. In the next corollary, we specialize Theorem \ref{Th_Inner_Nested_Bound_indirect_F} to this case and obtain an explicit polyhedral description for the inner bound.
\begin{corollary}
\label{Colollary_nested_indirect_2commonRec}
An inner bound of $K$-user DM BC for two nested messages ($M_{\overline{K-1.K}}$,$M_{\overline{\phi}}$) is the set of rate pairs ($R_{\overline{K-1.K}},R_{\overline{\phi}}$) satisfying
\begin{align}
R_{\overline{\phi}} &\leq I(U_{\mathsf{W}_i^{\mathsf{P}}};Y_i|Q)  \enspace \forall i\in S_l \label{Eq_Corollary_indirect_2common_1}\\
R_{\overline{\phi}}+R_{\overline{K-1.K}} &\leq  I(U_{\mathsf{W}_j^{\mathsf{P}}};Y_j|Q) \enspace \forall j\in S_p \label{Eq_Corollary_indirect_2common_2}
\\
R_{\overline{\phi}}+R_{\overline{K-1.K}} &\leq I(U_{  \downarrow_{\mathsf{W}_j^{\mathsf{P}}} \{\overline{K-1}\} };Y_j|U_{\mathsf{W}_j^{\mathsf{P}} \backslash  \downarrow_{\mathsf{W}_j^{\mathsf{P}}} \{\overline{K-1}\}  },Q)  + I(U_{\mathsf{W}_{K-1}^{\mathsf{P}}};Y_{K-1}|Q) 
\enspace \forall j\in S_p \label{Eq_Corollary_indirect_2common_3} \\
R_{\overline{\phi}}+R_{\overline{K-1.K}} &\leq  I(U_{  \downarrow_{\mathsf{W}_j^{\mathsf{P}}} \{\overline{K}\} };Y_j|U_{\mathsf{W}_j^{\mathsf{P}} \backslash  \downarrow_{\mathsf{W}_j^{\mathsf{P}}} \{\overline{K}\}  },Q)  + I(U_{\mathsf{W}_{K}^{\mathsf{P}}};Y_{K}|Q) 
\enspace \forall j\in S_p 
\label{Eq_Corollary_indirect_2common_4} \\
2R_{\overline{\phi}}+R_{\overline{K-1.K}} &\leq I(U_{\downarrow_{\mathsf{W}_j^{\mathsf{P}}} \{\overline{K-1.K}\} };Y_j|U_{\mathsf{W}_j^{\mathsf{P}}\backslash \downarrow_{\mathsf{W}_j^{\mathsf{P}}} \{\overline{K-1.K}\} },Q)+I(U_{\mathsf{W}_{K{-}1}^{\mathsf{P}}};Y_{K{-}1}|Q)+ I(U_{\mathsf{W}_{K}^{\mathsf{P}}};Y_{K}|Q)
\enspace \forall j\in S_p 
\label{Eq_Corollary_indirect_2common_5} \\
2R_{\overline{\phi}}+2R_{\overline{K-1.K}}
&\leq  
I(U_{\downarrow_{\mathsf{W}_{j_1}^{\mathsf{P}}} \{\overline{K-1.K}\} };Y_{j_1}|U_{\mathsf{W}_{j_1}^{\mathsf{P}}\backslash \downarrow_{\mathsf{W}_{j_1}^{\mathsf{P}}} \{\overline{K-1.K}\} },Q)
+
I(U_{\downarrow_{\mathsf{W}_{j_2}^{\mathsf{P}}} \{\overline{K-1},\overline{K}\} };Y_{j_2}|U_{\mathsf{W}_{j_2}^{\mathsf{P}}\backslash \downarrow_{\mathsf{W}_{j_2}^{\mathsf{P}}} \{\overline{K-1},\overline{K}\} },Q)\nonumber \\
&\qquad
+I(U_{\mathsf{W}_{K{-}1}^{\mathsf{P}}};Y_{K{-}1}|Q)  + I(U_{\mathsf{W}_{K}^{\mathsf{P}}};Y_{K}|Q)\enspace \forall j_1,j_2\in S_p \label{Eq_Corollary_indirect_2common_6} 
\end{align}
for some $p(q, u_{\mathsf{P}})= p(q) \prod_{S\in \mathsf{P}}p(u_S|u_{\uparrow_{\mathsf{P}} S\backslash \{ S \} },q)$ and $X$ as a deterministic function of $(Q, U_{\mathsf{P}})$.
\end{corollary}
\begin{proof}
The proof follows from Theorem \ref{Th_Inner_Nested_Bound_indirect_F} when $\mathsf{F}=\mathsf{P}$. We split $M_{\overline{K-1.K}}$ into four sub-messages $M_{\overline{K-1.K}\rightarrow \overline{K-1.K}}$, $M_{\overline{K-1.K}\rightarrow \overline{K}}$, $M_{\overline{K-1.K}\rightarrow \overline{K-1}}$, $M_{\overline{K-1.K}\rightarrow  \overline{\phi}}$. Then, by using the FME procedure, we project away the sub-rates and obtain an explicit polyhedral description of the achievable rate region in the statement of the corollary.  
\end{proof}


\section{Capacity Regions For The Combination Network}
\label{Sec_Capacity_combianation}

In this section, we establish the capacity region for combination network for each of three different message sets considered in Corollaries \ref{Colollary_general_indirect_C1_1_C2_1}, \ref{Colollary_nested_indirect_1commonRec} and \ref{Colollary_nested_indirect_2commonRec}. Throughout this section, we forgo coded time-sharing by setting $Q$ to be an uninformative constant in those corollaries. In particular, we will consider $(a)$ two messages each required by $K-1$ receivers, i.e., $\mathsf{E}=\{\overline{K},\overline{K-1}\}$ by specializing Corollary \ref{Colollary_general_indirect_C1_1_C2_1} $(b)$ two nested messages with one common receiver and any number of private receivers, i.e., $\mathsf{E}=\{\overline{K},\overline{\phi}\}$ by specializing Corollary \ref{Colollary_nested_indirect_1commonRec} and $(c)$ two nested messages with two common receivers and any number of private receivers, i.e., $\mathsf{E}=\{\overline{K-1.K},\overline{\phi}\}$ by specializing Corollary \ref{Colollary_nested_indirect_2commonRec}. The capacity for each message set is presented separately in different subsections. Note that in each of those corollaries we exploit the full generality of the coding scheme proposed in Theorem \ref{TH_genral_inner_General_msgs} and \ref{Th_Inner_Nested_Bound_indirect_F} by letting $\mathsf{F}=\mathsf{P}$. We will show that the benefit of this generality is that a single distribution suffices to achieve capacity independently of which of the three message sets is considered. In particular, choosing the auxiliary random variables $U_S$ for all $S\in \mathsf{P}$ to be independent and uniformly distributed over $\mathcal{V}_S$ where $|\mathcal{V}_S|=2^{C_S}$ and $V_S=U_S$ is capacity-achieving in each case. We prove the converses using mainly the submodularity of entropy.    

We define next a function we will use throughout this section. Let the modular function (over all subsets of $\mathsf{P}$) $C_{\mathsf{W}} \triangleq \sum_{S\in \mathsf{W}}C_S $ for any $\mathsf{W}\subseteq \mathsf{P}$. 


Since the proof of the capacity region for combination networks with different message sets depends on the use of certain general identities, we state them in the following lemma for easy reference. 
\begin{lemma}
\label{Lemma_Combination_networks_proof}
The following two identities hold:
\begin{enumerate}
\item For any set $S=i_1i_2\cdots i_N \subset \{1,2,\cdots, K\}$ and $i\in [1:K]$ 
\begin{align}
C_{\mathsf{W}_i^{\mathsf{P}}}&= C_{\downarrow_{\mathsf{W}_i^{\mathsf{P}}}\{\overline{i_1},\overline{i_2},\cdots, \overline{i_N}\}} +C_{\uparrow_{\mathsf{W}_i^\mathsf{P}}\{S\}} \label{Eq_Lemma_7}
\\
C_{\mathsf{W}_i^{\mathsf{P}}}&= C_{\downarrow_{\mathsf{W}_i^{\mathsf{P}}}\{\overline{S}\}} +C_{ \uparrow_{\mathsf{W}_i^\mathsf{P}}\{i_1,i_2,\cdots ,i_N\} }
\label{Eq:Lemma_77}
\end{align}
\item In the $K$-receiver combination network where $Y_i=V_{\mathsf{W}_i^{\mathsf{P}}}$ for all $i\in \{1,2,\cdots,K\}$, if the random variables $U_S$ for all $S\in \mathsf{P}$ are independent and uniformly distributed over $\mathcal{V}_S$ where $|\mathcal{V}_S|=2^{C_S}$ and $V_S=U_S$, then for any $ \mathsf{W}\subseteq \mathsf{W}_i^{\mathsf{P}}$,
\begin{equation}
I(U_\mathsf{W};Y_i|U_{\mathsf{W}_i^{\mathsf{P}} \backslash \mathsf{W}},Q)= C_{\mathsf{W}} \label{Eq_Lemma_8}
\end{equation} where $Q$ is a time sharing random variable.

\end{enumerate}
\end{lemma}
\begin{proof}
The proof of \eqref{Eq_Lemma_7} and \eqref{Eq:Lemma_77} follows directly from \eqref{Eq_Lemma_*}-\eqref{Eq_Lemma_**} and \eqref{Eq_Lemma_***}-\eqref{Eq_Lemma_****}, respectively, in Lemma \ref{Lemma_Sets}. In equation \eqref{Eq_Lemma_8}, let $|Q|=1$ so that no coded time sharing is used. Then, we get
\begin{align}
I(U_{\mathsf{W}};Y_i|U_{\mathsf{W}_i^{\mathsf{P}} \backslash \mathsf{W}})&{=}H(U_{\mathsf{W}}|U_{\mathsf{W}_i^{\mathsf{P}} \backslash \mathsf{W}})-H(U_{\mathsf{W}}|Y_i,U_{\mathsf{W}_i^{\mathsf{P}} \backslash \mathsf{W}})\enspace  \nonumber \\
& = H(U_{\mathsf{W}}|U_{\mathsf{W}_i^{\mathsf{P}} \backslash \mathsf{W}}) \label{Eq_Ach_comb_1}\\
& = \sum_{S\in \mathsf{W}}H(U_{S}) \label{Eq_Ach_comb_2}\\
& = \sum_{S\in \mathsf{W}}\log_2|\mathcal{V}_S| \label{Eq_Ach_comb_3}\\
& = \sum_{S\in \mathsf{W}} C_S=C_{\mathsf{W}}\nonumber 
\end{align}
where \eqref{Eq_Ach_comb_1} follows from $Y_i=U_{\mathsf{W}_i^{\mathsf{P}}}$ since $V_S=U_S$ for all $S{\in}\mathsf{P} $, \eqref{Eq_Ach_comb_2} from the independence of the auxiliary random variables, and \eqref{Eq_Ach_comb_3} from $U_S$ being uniformly distributed over $\mathcal{V}_S$ where $|\mathcal{V}_S| =2^{C_S}$.
\end{proof}


\subsection{The capacity for $\mathsf{E}=\{\overline{K},\overline{K-1}\}$}
\label{Sub_section_K_1_general}

In the following theorem, we establish the capacity region of the $K$-user combination network for $P=K-2$, $L_1=1$, and $L_2=1$. 
\begin{theorem}
\label{Th_Capacity_combination_networks_C1_1_C2_1}
The capacity region of the $K$-user combination network with two messages $M_{\overline{K}},M_{\overline{K-1}}$ is the set of rate pairs ($R_{\overline{K}},R_{\overline{K-1}}$) satisfying
\begin{align}
&R_{\overline{K-1}}  \leq C_{\mathsf{W}_K^{\mathsf{P}}} \label{Eq_Capacity_comb_network_C1_1_C2_1_1}\\
&R_{\overline{K}}  \leq C_{\mathsf{W}_{K-1}^{\mathsf{P}}} \label{Eq_Capacity_comb_network_C1_1_C2_1_2}\\
&R_{\overline{K-1}}+R_{\overline{K}}  \leq C_{\mathsf{W}_{j}^{\mathsf{P}}}\enspace \forall j \in  S_p  \label{Eq_Capacity_comb_network_C1_1_C2_1_3}\\
&R_{\overline{K-1}}+R_{\overline{K}}\leq C_{ \downarrow_{\mathsf{W}_{K-1}^{\mathsf{P}}} \{\overline{K}\} }+C_{\mathsf{W}_K^{\mathsf{P}}} \label{Eq_Capacity_comb_network_C1_1_C2_1_4}\\
&2R_{\overline{K-1}}+2R_{\overline{K}}\leq C_{ \downarrow_{\mathsf{W}_{j}^{\mathsf{P}}} \{\overline{K-1},\overline{K}\} } + C_{\mathsf{W}_{K}^{\mathsf{P}}} + C_{\mathsf{W}_{K-1}^{\mathsf{P}}} \enspace j\in  S_p  \label{Eq_Capacity_comb_network_C1_1_C2_1_7}
\end{align}  
Note that \eqref{Eq_Capacity_comb_network_C1_1_C2_1_4} is also equivalent to \begin{equation}
R_{\overline{K-1}}+R_{\overline{K}}\leq C_{ \downarrow_{\mathsf{W}_{K}^{\mathsf{P}}} \{\overline{K-1}\} }+C_{\mathsf{W}_{K-1}^{\mathsf{P}}} 
\end{equation}
\end{theorem}
\begin{proof} The inequalities \eqref{Eq_Capacity_comb_network_C1_1_C2_1_1}-\eqref{Eq_Capacity_comb_network_C1_1_C2_1_7} follow from applying \eqref{Eq_Lemma_8} in Lemma \ref{Lemma_Combination_networks_proof} to the inequalities \eqref{Eq_Corollary_indirect_C1_1_C2_1_1}-\eqref{Eq_Corollary_indirect_C1_1_C2_1_2},
\eqref{Eq_Corollary_indirect_C1_1_C2_1_3} for $j=K-1$, or equivalently, \eqref{Eq_Corollary_indirect_C1_1_C2_1_4} for $j=K$, and \eqref{Eq_Corollary_indirect_C1_1_C2_1_5} in Corollary \ref{Colollary_general_indirect_C1_1_C2_1}, respectively. On the other hand, when we apply \eqref{Eq_Lemma_8} in Lemma \ref{Lemma_Combination_networks_proof} to the inequalities \eqref{Eq_Corollary_indirect_C1_1_C2_1_3} and \eqref{Eq_Corollary_indirect_C1_1_C2_1_4} for $j\in S_p$, we obtain the following two redundant inequalities
\begin{align}
&R_{\overline{K-1}}+R_{\overline{K}}\leq C_{ \downarrow_{\mathsf{W}_{j}^{\mathsf{P}}} \{\overline{K}\} }+C_{\mathsf{W}_K^{\mathsf{P}}} \enspace j\in  S_p  \label{Eq_Capacity_comb_network_C1_1_C2_1_5}\\
&R_{\overline{K-1}}+R_{\overline{K}}\leq C_{ \downarrow_{\mathsf{W}_{j}^{\mathsf{P}}} \{\overline{K-1}\} }+C_{\mathsf{W}_{K-1}^{\mathsf{P}}} \enspace j\in  S_p  \label{Eq_Capacity_comb_network_C1_1_C2_1_6}
\end{align}
In particular, \eqref{Eq_Capacity_comb_network_C1_1_C2_1_5} is redundant since for any $j\in S_p$, we have 
\begin{align}
C_{\mathsf{W}_j^{\mathsf{P}}}&=C_{ \downarrow_{\mathsf{W}_{j}^{\mathsf{P}}} \{\overline{K}\} }+C_{ \uparrow_{\mathsf{W}_{j}^{\mathsf{P}}} \{K\} }\label{Eq_proof_redundance1}\\
&<C_{ \downarrow_{\mathsf{W}_{j}^{\mathsf{P}}} \{\overline{K}\} }+C_{\mathsf{W}_K^{\mathsf{P}}}\label{Eq_proof_redundance2}
\end{align}
where \eqref{Eq_proof_redundance1} follows from \eqref{Eq_Lemma_7} in Lemma \ref{Lemma_Combination_networks_proof} and \eqref{Eq_proof_redundance2} from $ \uparrow_{\mathsf{W}_{j}^{\mathsf{P}}} \{K\}\subset \mathsf{W}_K^{\mathsf{P}}$. Similarly, \eqref{Eq_Capacity_comb_network_C1_1_C2_1_6} is redundant, thereby proving achievability. 

For the converse proof, note that the inequalities \eqref{Eq_Capacity_comb_network_C1_1_C2_1_1}-\eqref{Eq_Capacity_comb_network_C1_1_C2_1_3} are just cutset bounds. That inequality \eqref{Eq_Capacity_comb_network_C1_1_C2_1_4} is an outer bound on the capacity region follows from the following:
\begin{align}
nR_{\overline{K}}&=H(M_{\overline{K}})\nonumber \\
&= I(M_{\overline{K}};Y_{K-1}^n,M_{\overline{K-1}})\nonumber \\
&\quad +H(M_{\overline{K}}|Y_{K-1}^n,M_{\overline{K-1}})\nonumber \\
&= I(M_{\overline{K}};Y_{K-1}^n|M_{\overline{K-1}})+n\epsilon_n \label{Eq_Converse_C1_C2_1_1} \\
&=H(Y_{K-1}^n|M_{\overline{K-1}})+n\epsilon_n\nonumber \\
&\leq H(Y_{K}^n,Y_{K-1}^n|M_{\overline{K-1}})+n\epsilon_n\nonumber \\
&= H(Y_{K}^n,Y_{K-1}^n,M_{\overline{K-1}})-nR_{\overline{K-1}}+n\epsilon_n\label{Eq_Converse_C1_C2_1_2} \\
&= H(Y_{K}^n,Y_{K-1}^n)+H(M_{\overline{K-1}}|Y_{K}^n,Y_{K-1}^n)\nonumber \\
&\quad -nR_{\overline{K-1}}+n\epsilon_n\nonumber\\
&\leq H(Y_{K}^n,Y_{K-1}^n) -nR_{\overline{K-1}}+2n\epsilon_n \label{Eq_Converse_C1_C2_1_3}\\
& =H(V_{\mathsf{W}_K^{\mathsf{P}}}^n,V_{\mathsf{W}_{K-1}^{\mathsf{P}}}^n) -nR_{\overline{K-1}}+2n\epsilon_n \label{Eq_Converse_C1_C2_1_4}\\
& =H(V_{\mathsf{W}_K^{\mathsf{P}}}^n)+H(V_{\mathsf{W}_{K-1}^{\mathsf{P}}}^n)-H(V_{\uparrow_{\mathsf{W}_{K-1}^{\mathsf{P}}}\{K\} }^n) \nonumber \\
&\quad -nR_{\overline{K-1}}+2n\epsilon_n \label{Eq_Converse_C1_C2_1_5}\\
& \leq nC_{\mathsf{W}_K^{\mathsf{P}}}+nC_{\mathsf{W}_{K-1}^{\mathsf{P}}}-nC_{\uparrow_{\mathsf{W}_{K-1}^{\mathsf{P}}}\{K\} } \nonumber \\
&\quad -nR_{\overline{K-1}}+2n\epsilon_n \label{Eq_Converse_C1_C2_1_6}\\
& = nC_{\mathsf{W}_K^{\mathsf{P}}}+nC_{\downarrow_{\mathsf{W}_{K-1}^{\mathsf{P}}}\{\overline{K}\} } 
-nR_{\overline{K-1}}+2n\epsilon_n \label{Eq_Converse_C1_C2_1_7}
\end{align}
where \eqref{Eq_Converse_C1_C2_1_1} follows from Fano's inequality and independence between messages, \eqref{Eq_Converse_C1_C2_1_2} from chain rule of entropy, \eqref{Eq_Converse_C1_C2_1_3} from Fano's inequality, \eqref{Eq_Converse_C1_C2_1_4} from the fact that $Y_i^n=V_{\mathsf{W}_i^{\mathsf{P}}}^n$ for all $i\in[1:K]$, \eqref{Eq_Converse_C1_C2_1_5} from the submodularity of entropy where $\mathsf{W}_K^{\mathsf{P}} \cap \mathsf{W}_{K-1}^{\mathsf{P}}= \uparrow_{\mathsf{W}_{K-1}^{\mathsf{P}}}\{K\}=\uparrow_{\mathsf{W}_{K}^{\mathsf{P}}}\{K-1\}$, and \eqref{Eq_Converse_C1_C2_1_6} from the fact that the uniform distribution maximizes entropy. Finally, \eqref{Eq_Converse_C1_C2_1_7} follows from Lemma \ref{Lemma_Combination_networks_proof} such that $C_{\mathsf{W}_{K-1}^{\mathsf{P}}}= C_{\downarrow_{\mathsf{W}_{K-1}^{\mathsf{P}}}\{\overline{K}\} }+C_{\uparrow_{\mathsf{W}_{K-1}^{\mathsf{P}}}\{K\} }$.

Lastly, for \eqref{Eq_Capacity_comb_network_C1_1_C2_1_7}, observe that for all $j\in S_p$ 
\begin{align}
n(R_{\overline{K}}&+R_{\overline{K-1}})=H(M_{\overline{K}})+H(M_{\overline{K-1}})\nonumber \\
&\leq I(M_{\overline{K}};Y_j^n,M_{\overline{K-1}})\nonumber \\ & \quad+I(M_{\overline{K-1}};Y_K^n,M_{\overline{K}})+2n\epsilon_n \label{Eq_Converse_C1_C2_3_1}\\
&= I(M_{\overline{K}};Y_j^n|M_{\overline{K-1}})\nonumber \\ & \quad+I(M_{\overline{K-1}};Y_K^n|M_{\overline{K}})+2n\epsilon_n \label{Eq_Converse_C1_C2_3_2}\\
&= H(Y_j^n|M_{\overline{K-1}})+H(Y_K^n|M_{\overline{K}})+2n\epsilon_n \nonumber\\
&\leq H(Y_j^n|M_{\overline{K-1}})\nonumber \\ & \quad+H(Y_K^n,Y_{K-1}^n|M_{\overline{K}})+2n\epsilon_n \nonumber\\
&= H(Y_j^n,M_{\overline{K-1}})-nR_{\overline{K-1}}+2n\epsilon_n\nonumber \\ & \quad+H(Y_K^n,Y_{K-1}^n,M_{\overline{K}})-nR_{\overline{K}} \nonumber\\
&\leq H(Y_j^n)+H(Y_K^n,Y_{K-1}^n)\nonumber \\ & \quad-nR_{\overline{K-1}}-nR_{\overline{K}}+4n\epsilon_n \label{Eq_Converse_C1_C2_3_3}\\
&= H(V_{\mathsf{W}_j^{\mathsf{P}}}^n)+H(V_{\mathsf{W}_K^{\mathsf{P}}}^n,V_{\mathsf{W}_{K-1}^{\mathsf{P}}}^n)\nonumber \\ & \quad-nR_{\overline{K-1}}-nR_{\overline{K}}+4n\epsilon_n \label{Eq_Converse_C1_C2_3_4}\\
&= H(V_{\mathsf{W}_j^{\mathsf{P}}}^n)+H(V_{\mathsf{W}_K^{\mathsf{P}}}^n)+H(V_{\mathsf{W}_{K-1}^{\mathsf{P}}}^n)\nonumber \\ & \quad-H(V_{\mathsf{W}_K^{\mathsf{P}}}^n\cap V_{\mathsf{W}_{K-1}^{\mathsf{P}}}^n )-nR_{\overline{K-1}}\nonumber \\ & \quad -nR_{\overline{K}}+4n\epsilon_n \label{Eq_Converse_C1_C2_3_5}\\
&\leq H(V_{\mathsf{W}_j^{\mathsf{P}}}^n)+H(V_{\mathsf{W}_K^{\mathsf{P}}}^n)+H(V_{\mathsf{W}_{K-1}^{\mathsf{P}}}^n)\nonumber \\ & \quad-H(V_{\mathsf{W}_K^{\mathsf{P}}}^n\cap V_{\mathsf{W}_{K-1}^{\mathsf{P}}}^n \cap V_{\mathsf{W}_j^{\mathsf{P}}}^n)-nR_{\overline{K-1}}\nonumber \\ & \quad -nR_{\overline{K}}+4n\epsilon_n \nonumber\\
&= H(V_{\mathsf{W}_j^{\mathsf{P}}}^n)+H(V_{\mathsf{W}_K^{\mathsf{P}}}^n)+H(V_{\mathsf{W}_{K-1}^{\mathsf{P}}}^n)\nonumber \\ & \quad-H(  V_{  \uparrow_{   \mathsf{W}_j^{\mathsf{P}}} \{K-1.K\} }^n)-nR_{\overline{K-1}}\nonumber \\ & \quad -nR_{\overline{K}}+4n\epsilon_n \label{Eq_Converse_C1_C2_3_55}\\
&\leq nC_{\mathsf{W}_j^{\mathsf{P}}}+nC_{\mathsf{W}_K^{\mathsf{P}}}+nC_{\mathsf{W}_{K-1}^{\mathsf{P}}}\nonumber \\ & \quad-nC_{  \uparrow_{   \mathsf{W}_j^{\mathsf{P}}} \{K-1.K\} }-nR_{\overline{K-1}}\nonumber \\ & \quad -nR_{\overline{K}}+4n\epsilon_n \label{Eq_Converse_C1_C2_3_6}\\
&= nC_{  \downarrow_{   \mathsf{W}_j^{\mathsf{P}}} \{\overline{K},\overline{K-1}\} }+nC_{\mathsf{W}_K^{\mathsf{P}}}+nC_{\mathsf{W}_{K-1}^{\mathsf{P}}}\nonumber \\ & \quad-nR_{\overline{K-1}} -nR_{\overline{K}}+4n\epsilon_n \label{Eq_Converse_C1_C2_3_7}
\end{align}
where \eqref{Eq_Converse_C1_C2_3_1} follows from Fano's inequality, \eqref{Eq_Converse_C1_C2_3_2} from independence between messages, \eqref{Eq_Converse_C1_C2_3_3} from Fano's inequality and chain rule, \eqref{Eq_Converse_C1_C2_3_4} from $Y_i^n=V_{\mathsf{W}_i^{\mathsf{P}}}^n$ for all $i\in[1:K]$, \eqref{Eq_Converse_C1_C2_3_5} from the submodularity of entropy, and \eqref{Eq_Converse_C1_C2_3_55} from the following
\begin{align*}
V_{\mathsf{W}_K^{\mathsf{P}}}^n\cap V_{\mathsf{W}_{K-1}^{\mathsf{P}}}^n \cap V_{\mathsf{W}_j^{\mathsf{P}}}^n&=\uparrow_{\mathsf{P}} \{j.K-1.K\}\\
&= \uparrow_{\mathsf{P}} \{K-1.K\} \enspace  \cap \enspace  \mathsf{W}_j^{\mathsf{P}}\\
&=\uparrow_{\mathsf{W}_j^{\mathsf{P}}} \{K-1.K\}
\end{align*}
Moreover, \eqref{Eq_Converse_C1_C2_3_6} follows from the fact that the uniform distribution maximizes entropy and \eqref{Eq_Converse_C1_C2_3_7} from Lemma \ref{Lemma_Combination_networks_proof} from which we have $C_{\mathsf{W}_j^{\mathsf{P}}}= C_{  \downarrow_{   \mathsf{W}_j^{\mathsf{P}}} \{\overline{K},\overline{K-1}\} }+ C_{  \uparrow_{   \mathsf{W}_j^{\mathsf{P}}} \{K-1. K\} }$. This completes the converse proof of Theorem \ref{Th_Capacity_combination_networks_C1_1_C2_1}.

\end{proof}

\begin{remark}
If we set $K= 3$ in Theorem \ref{Th_Capacity_combination_networks_C1_1_C2_1}, we get the capacity region for $\mathsf{E}=\{12,13\}$. It is left to the reader to verify that it is a special case of \cite{grokop2008fundamental} (see also \cite{romero2016superposition}) where the capacity region for combination network with general message set $\mathsf{E}=\mathsf{P}$ is established for $K=2$ and $K=3$. 
\end{remark}

\begin{example}
The capacity region of the combination network for $K> 3$ is not known in general. We consider $K{=}4$ with $\mathsf{E} = \{ 123, 124\}$ in this example. Using Theorem \ref{Th_Capacity_combination_networks_C1_1_C2_1}, the capacity region of the four-receiver combination network for this message set is 
the set of rate pairs ($R_{123},R_{124}$) satisfying 
\begin{align}
&R_{124}  \leq C_4+C_{14}+C_{24}+C_{34}+\nonumber \\
& \qquad\qquad C_{124}+C_{134}+C_{234}+C_{1234} \\
&R_{123}  \leq C_3+C_{13}+C_{23}+C_{34}+\nonumber \\
& \qquad\qquad C_{123}+C_{134}+C_{234}+C_{1234} \\
&R_{123}+R_{124}  \leq C_1+C_{12}+C_{13}+C_{14}+\nonumber \\
& \qquad\qquad C_{123}+C_{124}+C_{134}+C_{1234} \\
&R_{123}+R_{124}  \leq C_2+C_{12}+C_{23}+C_{24}+\nonumber \\
& \qquad\qquad C_{123}+C_{124}+C_{234}+C_{1234} \\
&R_{123}+R_{124}\leq C_3+C_{13}+C_{23}+C_{123}+ \nonumber \\ 
&\qquad\qquad  C_4+C_{14}+C_{24}+C_{34}+\nonumber \\
& \qquad\qquad C_{124}+C_{134}+C_{234}+C_{1234} \\
&2R_{123}+2R_{124}\leq C_1+C_{12}+C_{13}+C_{14}+C_{123}+C_{124} + \nonumber \\
& \qquad \qquad C_3+C_{13}+C_{23}+C_{34}+\nonumber \\
& \qquad\qquad C_{123}+C_{134}+C_{234}+C_{1234}+ \nonumber \\
&\qquad\qquad C_4+C_{14}+C_{24}+C_{34}+\nonumber \\
& \qquad\qquad C_{124}+C_{134}+C_{234}+C_{1234} \\
&2R_{123}+2R_{124}\leq C_2+C_{12}+C_{23}+C_{24}+C_{123}+C_{124} + \nonumber \\
& \qquad \qquad C_3+C_{13}+C_{23}+C_{34}+\nonumber \\
& \qquad\qquad C_{123}+C_{134}+C_{234}+C_{1234}+ \nonumber \\
&\qquad\qquad C_4+C_{14}+C_{24}+C_{34}+\nonumber \\
& \qquad\qquad C_{124}+C_{134}+C_{234}+C_{1234}
\end{align}  
\end{example}



\subsection{Recovering the capacity regions for $\mathsf{E}=\{\overline{K},\overline{\phi}\}$ and $\mathsf{E}=\{\overline{K-1.K},\overline{\phi}\}$ of \cite{bidokhti2016capacity} via Theorem \ref{Th_Inner_Nested_Bound_indirect_F} }
\label{Sub_section_1Common}

We recover the capacity regions of the combination network with two nested messages for the one and two common receivers cases next. These results were first obtained in \cite[Theorem 3]{bidokhti2016capacity}. While in \cite[Proposition 1]{bidokhti2016capacity} achievability was shown using rate-splitting and linear superposition coding tailored to the combination network we show it using a top-down approach by specializing Corollaries \ref{Colollary_nested_indirect_1commonRec} and \ref{Colollary_nested_indirect_2commonRec} on the DM BC to the combination network, which is indicative of the strength of the inner bound of Theorem \ref{Th_Inner_Nested_Bound_indirect_F} on the DM BC (from which Corollaries \ref{Colollary_nested_indirect_1commonRec} and \ref{Colollary_nested_indirect_2commonRec} were obtained using FME). These regions can be shown to coincide with those obtained in \cite[Theorem 3]{bidokhti2016capacity} where they were shown to be the capacity regions. We provide the proof of converses here as well for the sake of completeness since the notation here is quite different from that in \cite[Theorem 3]{bidokhti2016capacity}.
\begin{proposition}
\label{Th_Capacity_combination_networks_C1}
The capacity region of the $K$-user combination network with $\mathsf{E}=\{\overline{K},\overline{\phi}\}$ is the set of rate pairs ($R_{\overline{K}},R_{\overline{\phi}}$) satisfying
\begin{align}
R_{\overline{\phi}} & \leq C_{\mathsf{W}_K^{\mathsf{P}}} \label{Eq_Capacity_comb_network_1C1}\\
R_{\overline{\phi}}+R_{\overline{K}} & \leq C_{\mathsf{W}_j^{\mathsf{P}}} \enspace \forall j\in [1:K-1] \label{Eq_Capacity_comb_network_2C1}
\end{align}  
where $C_{\mathsf{W}}=\sum_{S\in \mathsf{W}} C_S$ for any $ \mathsf{W} \subseteq \mathsf{P} $.
\end{proposition}
\begin{proof}
The proof of the converse follows from the receiver cutset bounds since the common receiver desires only the common message $M_{\overline{\phi}}$ and the private receivers desire both messages. The proof of achievability follows from using \eqref{Eq_Lemma_8} of Lemma \ref{Lemma_Combination_networks_proof} 
in the inequalities \eqref{Eq_Corollary_indirect_1common_1} and \eqref{Eq_Corollary_indirect_1common_2} in the achievable region of Corollary \ref{Colollary_nested_indirect_1commonRec}. We show next that the inequality \eqref{Eq_Corollary_indirect_1common_3} of that region is redundant because of \eqref{Eq_Corollary_indirect_1common_2}. From \eqref{Eq_Lemma_7} of Lemma \ref{Lemma_Combination_networks_proof}, we have $C_{\mathsf{W}_j^{\mathsf{P}}}= C_{\downarrow_{\mathsf{W}_j^{\mathsf{P}}} \{\overline{K} \} }+C_{\uparrow_{\mathsf{W}_j^{\mathsf{P}}} \{K \} }$ for any $j\in [1:K-1]$. Moreover, from the definition of $C_{\mathsf{W}}$,  we have $C_{\uparrow_{\mathsf{W}_j^{\mathsf{P}}} \{K \}}<  C_{\mathsf{W}_K^{\mathsf{P}}}$. Hence, $C_{\mathsf{W}_j^{\mathsf{P}}}< C_{\downarrow_{\mathsf{W}_j^{\mathsf{P}}} \{\overline{K} \} }+ C_{\mathsf{W}_K^{\mathsf{P}}}$ for any $j\in [1:K-1]$ and so \eqref{Eq_Corollary_indirect_1common_3} is redundant. 
\end{proof}

\begin{proposition}\cite[Theorem 3]{bidokhti2016capacity}
\label{Th_Capacity_combination_networks_C2}
The capacity region of the $K$-user combination network with $\mathsf{E}=\{\overline{K-1.K},\overline{\phi}\}$ is the set of rate pairs ($R_{\overline{K-1.K}},R_{\overline{\phi} }$) satisfying
\begin{align}
R_{\overline{\phi}} & \leq C_{\mathsf{W}_i^{\mathsf{P}}} \enspace \forall i\in  S_l \label{Eq_Capacity_comb_network_1}\\
R_{\overline{\phi}}+R_{\overline{K-1.K}} & \leq C_{\mathsf{W}_j^{\mathsf{P}}} \enspace \forall j\in S_p \label{Eq_Capacity_comb_network_2}\\
2R_{\overline{\phi}} + R_{\overline{K-1.K}} & \leq C_{\downarrow_{\mathsf{W}_j^{\mathsf{P}}}\{\overline{K-1.K}\}} + C_{\mathsf{W}_{K-1}^{\mathsf{P}}} + C_{\mathsf{W}_{K}^{\mathsf{P}}} \enspace \forall j\in S_p \label{Eq_Capacity_comb_network_3}
\end{align}  
\end{proposition}
\begin{proof}
The three inequalities of \eqref{Eq_Capacity_comb_network_1}-\eqref{Eq_Capacity_comb_network_3} follow from applying \eqref{Eq_Lemma_8} of Lemma \ref{Lemma_Combination_networks_proof} to inequalities \eqref{Eq_Corollary_indirect_2common_1}, \eqref{Eq_Corollary_indirect_2common_2},  and \eqref{Eq_Corollary_indirect_2common_5} in Corollary \ref{Colollary_nested_indirect_2commonRec}. We next show that the inequalities that result from \eqref{Eq_Corollary_indirect_2common_3}, \eqref{Eq_Corollary_indirect_2common_4} and \eqref{Eq_Corollary_indirect_2common_6} will be redundant. 

First, \eqref{Eq_Corollary_indirect_2common_6} is redundant when $j_1= j_2=j\in [1:K-2]$ because it is the sum of \eqref{Eq_Corollary_indirect_2common_3} and \eqref{Eq_Corollary_indirect_2common_4} due to the modularity of $C_{\mathsf{W}}$ for any $\mathsf{W}\subseteq \mathsf{P}$ so that $C_{\downarrow_{\mathsf{W}_j^{\mathsf{P}}} \{\overline{K-1} \}}+C_{\downarrow_{\mathsf{W}_j^{\mathsf{P}}} \{\overline{K} \}}= C_{\downarrow_{\mathsf{W}_j^{\mathsf{P}}} \{\overline{K-1}, \overline{K} \}}+C_{\downarrow_{\mathsf{W}_j^{\mathsf{P}}} \{\overline{K-1.K} \}}$. 
While, for $j_1\not =j_2$, \eqref{Eq_Corollary_indirect_2common_6} is still redundant from \eqref{Eq_Corollary_indirect_2common_2} since from \eqref{Eq_Lemma_7} and \eqref{Eq:Lemma_77}, we have
\begin{align}
    C_{\mathsf{W}_{j_1}^{\mathsf{P}}}= C_{\downarrow_{\mathsf{W}_{j_1}^{\mathsf{P}}} \{\overline{K-1},\overline{K}\}}+C_{\uparrow_{\mathsf{W}_{j_1}^{\mathsf{P}}} \{K-1.K\}}\\
    C_{\mathsf{W}_{j_2}^{\mathsf{P}}}=C_{ \downarrow_{\mathsf{W}_{j_2}^{\mathsf{P}}} \{\overline{K-1.K}\}}+C_{\uparrow_{\mathsf{W}_{j_2}^{\mathsf{P}}} \{K-1,K\}}
\end{align}
and $C_{\uparrow_{\mathsf{W}_{j_1}^{\mathsf{P}}} \{K-1.K\}}+C_{\uparrow_{\mathsf{W}_{j_2}^{\mathsf{P}}} \{K-1,K\}}\leq  C_{\mathsf{W}_{K-1}^{\mathsf{P}}}+C_{\mathsf{W}_{K}^{\mathsf{P}}}$.

Moreover, we can show that \eqref{Eq_Corollary_indirect_2common_3} and \eqref{Eq_Corollary_indirect_2common_4} are redundant because of \eqref{Eq_Corollary_indirect_2common_2} by following exactly the same argument that we used to show that \eqref{Eq_Corollary_indirect_1common_3} is redundant because of \eqref{Eq_Corollary_indirect_1common_2}. Hence, \eqref{Eq_Corollary_indirect_2common_3} is redundant since $C_{\mathsf{W}_j^{\mathsf{P}}}< C_{\downarrow_{\mathsf{W}_j^{\mathsf{P}}} \{\overline{K-1} \} }+ C_{\mathsf{W}_{K-1}^{\mathsf{P}}}$ for any $j\in [1:K-2]$ and \eqref{Eq_Corollary_indirect_2common_4} is redundant since $C_{\mathsf{W}_j^{\mathsf{P}}}< C_{\downarrow_{\mathsf{W}_j^{\mathsf{P}}} \{\overline{K} \} }+ C_{\mathsf{W}_K^{\mathsf{P}}}$ for any $j\in [1:K-2]$.

The inequalities \eqref{Eq_Capacity_comb_network_1} and \eqref{Eq_Capacity_comb_network_2} are cutset outer bounds since the common receivers $Y_i$ ($i\in S_l $) desire only the common message $M_{\overline{\phi}}$ and the private receivers $Y_j$ ($j\in  S_p $) desire both messages. More interestingly, the converse proof of \eqref{Eq_Capacity_comb_network_3} again uses the sub-modularity of entropy as in \cite{bidokhti2016noisy,prabhakaran2007broadcasting}. Assume that the transmission is done over a block of length $n$. Then, for any $j\in  S_p $, we have
\begin{align}
nR_{\overline{K-1.K}}& = H(M_{\overline{K-1.K}})\nonumber 
\\
& \hspace{-0.25in}  = I(M_{\overline{K-1.K}};Y_{j}^n,M_{ \overline{\phi} }) + H(M_{\overline{K-1.K}}|Y_{j}^n,M_{ \overline{\phi} })\nonumber 
\\
& \hspace{-0.25in}  \leq I(M_{\overline{K-1.K}};Y_{j}^n|M_{ \overline{\phi} })+n\epsilon_n \label{Eq_converse_fano_indep}
\\
& \hspace{-0.25in}  = H(Y_j^n|M_{ \overline{\phi} })+n\epsilon_n \nonumber
\\
& \hspace{-0.25in}  \leq H(Y_{K-1}^n,Y_K^n,Y_j^n|M_{ \overline{\phi} })+ n\epsilon_n\label{Eq_converse_chainRule}
\\
& \hspace{-0.25in}  = H(V_{\mathsf{W}_{K-1}^{\mathsf{P}} }^n,V_{\mathsf{W}_{K}^{\mathsf{P}} }^n,V_{\mathsf{W}_{j}^{\mathsf{P}} }^n|M_{ \overline{\phi} })+n\epsilon_n \label{Eq_converse_Y_V}
\\
& \hspace{-0.25in}  \leq H(V_{\mathsf{W}_{K-1}^{\mathsf{P}} }^n|M_{ \overline{\phi} })+ H(V_{\mathsf{W}_{K}^{\mathsf{P}} }^n|M_{ \overline{\phi} })\nonumber \\& \qquad + H(V_{\mathsf{W}_{j}^{\mathsf{P}} }^n|V_{\mathsf{W}_{K-1}^{\mathsf{P}} }^n,V_{\mathsf{W}_{K}^{\mathsf{P}} }^n,M_{ \overline{\phi} })+ n\epsilon_n\label{Eq_converse_subMod}\\
& \hspace{-0.25in}  = H(V_{\mathsf{W}_{K-1}^{\mathsf{P}} }^n,M_{ \overline{\phi} }){-}nR_{ \overline{\phi} }\nonumber \\ & \hspace{-0.25in} \qquad + H(V_{\mathsf{W}_{K}^{\mathsf{P}} }^n,M_{ \overline{\phi} }){-}nR_{ \overline{\phi} }\nonumber \\& \qquad + H(V_{\mathsf{W}_{j}^{\mathsf{P}} }^n|V_{\mathsf{W}_{K-1}^{\mathsf{P}} }^n,V_{\mathsf{W}_{K}^{\mathsf{P}} }^n,M_{ \overline{\phi} })+ n\epsilon_n\nonumber\\
& \hspace{-0.25in}  \leq H(V_{\mathsf{W}_{K-1}^{\mathsf{P}} }^n)+ H(V_{\mathsf{W}_{K}^{\mathsf{P}} }^n)-2nR_{ \overline{\phi} }\nonumber \\&\hspace{-0.25in} \qquad+ H(V_{\mathsf{W}_{j}^{\mathsf{P}} }^n|V_{\mathsf{W}_{K-1}^{\mathsf{P}} }^n,V_{\mathsf{W}_{K}^{\mathsf{P}} }^n,M_{ \overline{\phi} })+ 3n\epsilon_n \label{Eq_converse_fanos_again}\\
& \hspace{-0.25in} \leq H(V_{\mathsf{W}_{K-1}^{\mathsf{P}} }^n)+ H(V_{\mathsf{W}_{K}^{\mathsf{P}} }^n)-2nR_{ \overline{\phi} }\nonumber \\&  \qquad \qquad + H(V_{\mathsf{W}_{j}^{\mathsf{P}} }^n|V_{\mathsf{W}_{K-1}^{\mathsf{P}} }^n,V_{\mathsf{W}_{K}^{\mathsf{P}} }^n)+ 3n\epsilon_n\label{Eq_converse_cond_reduce}\\
& \hspace{-0.25in} = H(V_{\mathsf{W}_{K-1}^{\mathsf{P}} }^n)+ H(V_{\mathsf{W}_{K}^{\mathsf{P}} }^n)-2nR_{ \overline{\phi} }\nonumber \\& \qquad \qquad + H(V_{\downarrow_{\mathsf{W}_j^{\mathsf{P}}} \{\overline{K-1.K}\}}^n)+ 3n\epsilon_n\label{Eq_converse_setfact}\\
& \hspace{-0.25in} \leq nC_{\mathsf{W}_{K-1}^{\mathsf{P}} }+ nC_{\mathsf{W}_{K-1}^{\mathsf{P}} } + nC_{\downarrow_{\mathsf{W}_j^{\mathsf{P}}} \{\overline{K-1.K}\}} \nonumber \\ & \qquad \qquad
{-}2nR_{ \overline{\phi} } + 3n\epsilon_n \label{Eq_converse_unif}
\end{align}
where \eqref{Eq_converse_fano_indep} follows from Fano's inequality and the independence between the messages, \eqref{Eq_converse_chainRule} from the chain rule of entropy, \eqref{Eq_converse_Y_V} from the fact that $Y_{i}^n=V_{\mathsf{W}_{i}^{\mathsf{P}} }^n$ for any $i\in[1{:}K]$, \eqref{Eq_converse_subMod} from the submodularity of entropy, \eqref{Eq_converse_fanos_again} from Fano's inequality, \eqref{Eq_converse_cond_reduce} from conditioning reduces entropy, \eqref{Eq_converse_setfact} from the fact that $\mathsf{W}_j^{\mathsf{P}}\backslash\{\mathsf{W}_{K-1}^{\mathsf{P}} \cup \mathsf{W}_K^{\mathsf{P}}\}= \downarrow_{\mathsf{W}_j^{\mathsf{P}}} \{\overline{K-1.K}\}$ for any $j\in  S_p $, and \eqref{Eq_converse_unif} from the fact that the uniform distribution maximizes entropy. This completes the proof of Proposition \ref{Th_Capacity_combination_networks_C2}.  
\end{proof}

\begin{remark}
In Theorem \ref{Th_Capacity_combination_networks_C1_1_C2_1}, 
and in Propositions \ref{Th_Capacity_combination_networks_C1} and \ref{Th_Capacity_combination_networks_C2}, we prove achievability top-down by specializing random coding in the DM BC to the combination network. This is in contrast to the proof in \cite[Proposition 1, Theorem 3]{bidokhti2016capacity} for nested messages which is tailored to the combination network via explicit linear network coding. Also, our descriptions for the rate regions are more stream-lined, thanks to the use of order theory,
as are the proofs of the converses, even though the main ingredient, besides standard information inequalities, in the proof of the converses for Theorem \ref{Th_Capacity_combination_networks_C1_1_C2_1} and Proposition \ref{Th_Capacity_combination_networks_C2} is the sub-modularity of entropy as it is in \cite[Theorem 3]{bidokhti2016capacity}. 
\end{remark}

\begin{example}
In this example, we show the importance of choosing the message index superset $\mathsf{F}=\mathsf{P}$ in Theorem \ref{Th_Inner_Nested_Bound_indirect_F} via the example of $K=3,L=1$, i.e., $\mathsf{E}=\{1,123\}$. From Proposition \ref{Th_Capacity_combination_networks_C2}, the capacity region for this case is given by
\begin{align}
R_{123}&\leq \min\{C_2+C_{12}+C_{23}+C_{123},\; C_3+C_{13}+C_{23}+C_{123}\}\label{Eq_Capacity_K=3_C=2,1}\\
R_{123}+R_{1}&\leq C_1+C_{12}+C_{13}+C_{123}\label{Eq_Capacity_K=3_C=2,2}\\
2R_{123}+R_1 &\leq C_1+C_2+C_3+C_{12}+C_{13}+2C_{23}+2C_{123}\label{Eq_Capacity_K=3_C=2,3}
\end{align}
Consider next the coding scheme commonly used in the  literature which effectively sets $\mathsf{F}=\uparrow_{\mathsf{P}}\mathsf{E}$
in Theorem \ref{Th_Inner_Nested_Bound_indirect_F}, i.e., $\mathsf{F}=\{1,12,13,123\}$. It is not hard to show that the achievable rate region is the set of rate pairs 
satisfying
\begin{align}
R_{123} &\leq \min\{I(U_{12},U_{123};Y_2),I(U_{13},U_{123};Y_3)\}\\
R_{123}+R_{1}&\leq \min\{I(U_1;Y_1), \nonumber \\ &\qquad \qquad \enspace I(U_1;Y_1|U_{123},U_{12})+I(U_{12},U_{123};Y_2), \nonumber \\ &\qquad \qquad \enspace I(U_1;Y_1|U_{123},U_{13})+I(U_{13},U_{123};Y_3)\} \\
2R_{123}+R_{1}&\leq I(U_1;Y_1|U_{12},U_{13})+ I(U_{12},U_{123};Y_2)+I(U_{13},U_{123};Y_3)\\
2R_{123}+2R_{1}&\leq I(U_1;Y_1|U_{12},U_{13})+I(U_1;Y_1|U_{123})+\nonumber \\&\qquad \qquad \enspace I(U_{12},U_{123};Y_2)+I(U_{13},U_{123};Y_3)
\end{align} for some $p(u_{123})p(u_{12}|u_{123})p(u_{13}|u_{123})p(u_1|u_{12},u_{13})$. It can be shown that the above region does not become the capacity region by choosing independent auxiliary random variables with a single distribution for any choice of the channel input component $V_S$ (where $S\in \mathsf{P}$). 

By examining the capacity region in \eqref{Eq_Capacity_K=3_C=2,1}-\eqref{Eq_Capacity_K=3_C=2,3}, an intuitive choice for the channel input components $ V_S$ ($S\in \mathsf{P}$) and the auxiliary random variables $U_S$ ($S\in \mathsf{F}$) to achieve \eqref{Eq_Capacity_K=3_C=2,1} and \eqref{Eq_Capacity_K=3_C=2,3} is $(V_{123},V_{23})=U_{123}$, $(V_{2},V_{12})=U_{12}$, $(V_{3},V_{13})=U_{13}$, and $V_1=U_1$ where the auxiliary random variables $U_{123},U_{13},U_{12},U_1$ are uniform distributions over $\mathcal{V}_{123}\times \mathcal{V}_{23}$, $\mathcal{V}_{13}\times \mathcal{V}_{3}$, $\mathcal{V}_{12}\times \mathcal{V}_{2}$ and $\mathcal{V}_{1}$, respectively. For this choice, following the same analysis as in Proposition \ref{Th_Capacity_combination_networks_C2}, we can show that the rate pairs $(R_1, R_{123})$ that satisfy 
\begin{align}
R_{123}&\leq \min\{C_2+C_{12}+C_{23}+C_{123},\; C_3+C_{13}+C_{23}+C_{123}\}\label{Eq_Capacity_K=3_C=2,1Choice1}\\
R_{123}+R_{1}&\leq C_1\label{Eq_Capacity_K=3_C=2,2Choice1}\\
2R_{123}+R_1 &\leq C_1+C_2+C_3+C_{12}+C_{13}+2C_{23}+2C_{123}\label{Eq_Capacity_K=3_C=2,3Choice1}
\end{align}
are achievable.
Note that the last inequality is redundant. Obviously, the above region is strictly smaller than the capacity region given in \eqref{Eq_Capacity_K=3_C=2,1}-\eqref{Eq_Capacity_K=3_C=2,3}.

A different choice that achieves \eqref{Eq_Capacity_K=3_C=2,2} is $V_{23}=U_{123}$, $V_{2}=U_{12}$, $V_{3}=U_{13}$, and $(V_1,V_{12},V_{13},V_{123})=U_1$ where the auxiliary random variables $U_{123},U_{13},U_{12},U_1$ are uniform distributions according to $\mathcal{V}_{23}$, $\mathcal{V}_{3}$, $\mathcal{V}_{2}$ and $\mathcal{V}_{123}\times \mathcal{V}_{13}\times \mathcal{V}_{12}\times \mathcal{V}_{1}$, respectively.
For this choice, the rate pairs $(R_1, R_{123})$ that satisfy
\begin{align}
R_{123} &\leq \min \{C_{2}+C_{23}, C_{3}+C_{23} \}\label{Eq_Capacity_K=3_C=2,1Choice2}\\
R_{123}+R_1 &\leq C_1+C_{12}+C_{13}+C_{123}\label{Eq_Capacity_K=3_C=2,2Choice2}\\
2R_{123}+R_1 &\leq C_1+C_2+C_3+C_{12}+C_{13}+2C_{23}+C_{123}\label{Eq_Capacity_K=3_C=2,3Choice2}
\end{align} are achievable. Note that the last inequality is redundant and the above region is also strictly smaller than the capacity region.  

In Fig. \ref{Fig:Capacity_Region_and_Inner_Bounds}, we show that even the convex hull of the union of the two inner bounds, given in \eqref{Eq_Capacity_K=3_C=2,1Choice1}-\eqref{Eq_Capacity_K=3_C=2,3Choice1} and \eqref{Eq_Capacity_K=3_C=2,1Choice2}-\eqref{Eq_Capacity_K=3_C=2,3Choice2}, is strictly contained in the capacity region. 

\begin{figure}
  \includegraphics[width=\linewidth]{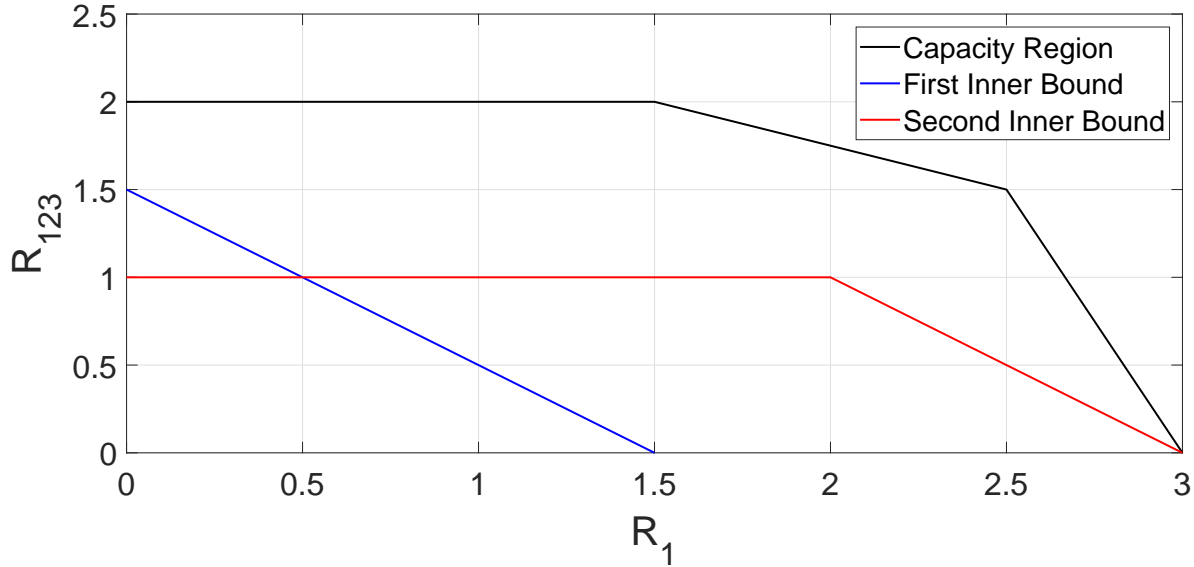}
  \caption{The capacity region is compared with the two inner bounds given in \eqref{Eq_Capacity_K=3_C=2,1Choice1}-\eqref{Eq_Capacity_K=3_C=2,3Choice1} and \eqref{Eq_Capacity_K=3_C=2,1Choice2}-\eqref{Eq_Capacity_K=3_C=2,3Choice2} for $K=3,C=2$. The finite capacity links are given as follows: $C_1=1.5$, $C_2=0.5$, $C_3=0.75$, $C_{12}=0.75$, $C_{13}=0.5$, $C_{23}=0.5$, $C_{123}=0.25$.}
  \label{Fig:Capacity_Region_and_Inner_Bounds}
\end{figure}

\label{ex:K=3L=1-CN}
\end{example}

Example \ref{ex:K=3L=1-CN} illustrates the importance of choosing the message index superset $\mathsf{F}=\mathsf{P}$. Although, this choice adds complexity to the coding scheme, it simplifies the choice of the coding distribution that achieves capacity. In particular, independent auxiliary random variables with uniform distribution are extremal. 

Nevertheless, we show in Section \ref{Sec:Smaller_F} that choosing $\mathsf{F}\subset \mathsf{P}$ can also achieve the capacity of the combination network, provided $\mathsf{F}$ is chosen appropriately depending on the message set. Moreover, independent auxiliary random variables with uniform distribution do not suffice. A certain dependency has to be introduced in them. 



\subsection{Other Message Sets}
In the three scenarios of Sections \ref{Sub_section_1Common} and \ref{Sub_section_K_1_general} for which we were able to establish the capacity region for combination networks, we had a polygonal description for the inner bound by applying FME to project away the split rates. This was possible because in these three cases, there are only up to four split rates under up-set rate splitting. For other choices of message pairs (with $K>3$) we would have many more split rates which would render FME too tedious or intractable. 
Nevertheless, we suspect that the achievable rate region of Theorem \ref{TH_genral_inner_General_msgs} is large enough to be the capacity of the combination network for other message pairs (such as say nested messages with $L\geq 3$). 

Moreover, it is likely that our approach leads to the capacity region for combination networks for more than just two messages as well. For instance, it can be shown that by using a similar analysis as in Proposition \ref{Th_Capacity_combination_networks_C2}, we can establish the capacity region for three degraded messages, i.e., $\mathsf{E}=\{M_{\overline{\phi}},M_{\overline{K}},M_{\overline{K-1.K}}\}$. Hence, discovering all  message sets (ideally, the message set $\mathsf{E}=\mathsf{P}$) for which our approach leads to the capacity region for the combination network is a topic of future research.  

\section{Further Evidence of the Strength of Theorem \ref{Th_Inner_Nested_Bound_indirect_F} 
}
\label{Sec_Relation_Bidokhti}

In this section, we provide further evidence of the strength of the inner bound of Theorem \ref{Th_Inner_Nested_Bound_indirect_F} by specializing it to the combination network to recover, with a single universal coding distribution, the entire achievable rate region of \cite[Proposition 1]{bidokhti2016capacity}, the latter being obtained via rate-splitting and linear superposition coding in \cite{bidokhti2016capacity}. Note that that region further specialized to the cases of one or two common receivers yielded the capacity region for those cases as was shown in \cite[Theorem 3]{bidokhti2016capacity} and we recovered those capacity results in Propositions \ref{Th_Capacity_combination_networks_C1} and \ref{Th_Capacity_combination_networks_C2}, with achievability in a top-down manner via Theorem \ref{Th_Inner_Nested_Bound_indirect_F}, in Section \ref{Sub_section_1Common}.

Moreover, the authors in \cite{bidokhti2016capacity} showed that the achievable region in \cite[Theorem 1]{bidokhti2016capacity}, which is achieved by adding a pre-encoder to the rate-splitting and linear superposition coding scheme of \cite[Proposition 1]{bidokhti2016capacity}, is strictly larger than that of \cite[Proposition 1]{bidokhti2016capacity} through \cite[Example 2]{bidokhti2016capacity} with $K=6$ with three common receivers, i.e., there is an achievable rate pair achieved by \cite[Theorem 1]{bidokhti2016capacity} but not by \cite[Proposition 1]{bidokhti2016capacity}. Moreover, through \cite[Example 4]{bidokhti2016capacity} in which $K=7$ and four common receivers, it was shown that the achievable region in \cite[Theorem 2]{bidokhti2016capacity}, achieved via a block Markov coding scheme, is strictly larger than that in \cite[Theorem 1]{bidokhti2016capacity}. 
In this section, we specialize the rate region of Theorem \ref{Th_Inner_Nested_Bound_indirect_F} to the combination network and
identify the auxiliary random variables that achieve the rate pairs achievable by \cite[Theorem 1]{bidokhti2016capacity} and \cite[Theorem 2]{bidokhti2016capacity} in the key examples \cite[Example 2]{bidokhti2016capacity} and \cite[Example 4]{bidokhti2016capacity}, respectively.

We begin by recovering the rate region of \cite[Proposition 1]{bidokhti2016capacity} in the following proposition in a top-down manner via the more general Theorem \ref{Th_Inner_Nested_Bound_indirect_F}.
\begin{proposition}\cite[Proposition 1]{bidokhti2016capacity}
\label{Proposition_Regenerate_Bidokhti_el_al}
The rate pair $(R_{S_p}, R_{\overline{\phi}})$ is achievable over the $K$-user combination network if  there exist real numbers $R_{S_p \rightarrow S}$ for all $S \in  \uparrow_{\mathsf{P}} \{S_p\} $ such that the following inequalities are satisfied:
\begin{align}
    R_{S_p}&=\sum_{S \in  \uparrow_{\mathsf{P}} \{S_p\} } R_{S_p \rightarrow S} 
    \label{Reproduce_prop1_Feasiablity_4}\\
    R_{S_p \rightarrow S} &\geq 0 \; \forall S \in  \uparrow_{\mathsf{P}} \{S_p\} 
    \label{Reproduce_prop1_Feasiablity_5} \\
& R_{\overline{\phi}}+R_{S_p}  \leq  C_{\mathsf{W}_j^{\mathsf{P}}} \enspace \forall j\in S_p 
\label{Reproduce_prop1_Feasiablity_1}\\
& \sum_{S\in \mathsf{B}}  R_{S_p\rightarrow S} \leq C_{\mathsf{B}}
 \qquad\forall \mathsf{B}\in \mathcal{F}_{\downarrow}(\mathsf{W}_j^{\mathsf{P} }\backslash \overline{\phi}), \; \forall j\in S_p  
 \label{Reproduce_prop1_Feasiablity_2}\\
&R_{\overline{\phi}}+\sum_{S\in \uparrow_{\mathsf{P}} \{ 1 2 \cdots Pi\}} R_{S_p\rightarrow S}\leq C_{\mathsf{W}_i^{\mathsf{P}}}  \enspace \forall i\in S_l 
\label{Reproduce_prop1_Feasiablity_3}
\end{align} 
\end{proposition}

\begin{proof}
From Theorem \ref{Th_Inner_Nested_Bound_indirect_F}, we set $\mathsf{F}=\mathsf{P}$ and choose the random variables $U_S$ for all $S\in \mathsf{P}$ to independent and uniformly distributed over $\mathcal{V}_S$ where $|\mathcal{V}_S|=2^{C_S}$ and $V_S=U_S$. For this specific choice, we can apply \eqref{Eq_Lemma_8} in Lemma \ref{Lemma_Combination_networks_proof} to Theorem \ref{Th_Inner_Nested_Bound_indirect_F} to get \eqref{Reproduce_prop1_Feasiablity_1} from \eqref{Eq_TH_indirect_generalF_private1}, \eqref{Reproduce_prop1_Feasiablity_2} from \eqref{Eq_TH_indirect_generalF_private2}, and \eqref{Reproduce_prop1_Feasiablity_3} from 
\eqref{Eq_TH_indirect_generalF_common}.
\end{proof}


\begin{remark}
The region in Proposition \ref{Proposition_Regenerate_Bidokhti_el_al} is obtained in \cite{bidokhti2016capacity} by rate-splitting of the private message and linear superposition coding scheme where the transmitted signal is obtained by the multiplication of the information symbols vector by the so-called zero-structured matrix \cite[Definition 2]{bidokhti2016capacity}. Here, Proposition \ref{Proposition_Regenerate_Bidokhti_el_al} 
is obtained
by specializing the inner bound on the achievable rate region for the DM BC of Theorem \ref{Th_Inner_Nested_Bound_indirect_F} (which is a union of polytopes over all admissible coding distributions given therein) to the combination network and by restricting the distribution of the auxiliary random variables to be a {\em single} distribution. Hence, it is possible that the entire rate region of our coding scheme of Theorem \ref{Th_Inner_Nested_Bound_indirect_F} is much larger than the one given by \cite[Proposition 1]{bidokhti2016capacity}.
We address this important point in Examples \ref{Example_Bidokhti_Example_K=6} and \ref{Example_Bidokhti_Example_K=7} to follow.
\end{remark}

\begin{remark}
The authors in \cite{bidokhti2016capacity} define Superset Saturated Subsets \cite[Definition 1]{bidokhti2016capacity}, which is equivalent to up-sets in our notation, to be able to describe the inequality in \eqref{Reproduce_prop1_Feasiablity_2} without using order theory. 
Note that the complement of any down-set is an upset, and hence, \eqref{Reproduce_prop1_Feasiablity_2} can be described using the family of down-sets or the family of upsets. In fact, \eqref{Reproduce_prop1_Feasiablity_2} can be written as 
\begin{equation}
    R_{S_p}-\sum_{S\in \mathsf{B}}R_{S_p\rightarrow S}\leq C_{\mathsf{W}_j^{\mathsf{P}} \backslash \mathsf{B} }  \; \forall \; \mathsf{B} \in \mathcal{F}_{\uparrow}(\mathsf{W}_j^{\mathsf{P} }) \backslash \mathsf{W}_j^{\mathsf{P}},\; \forall j\in S_p
\end{equation}
\end{remark}

Before proceeding to the three common receivers example in \cite[Example 2]{bidokhti2016capacity}, we specialize first Theorem \ref{Th_Inner_Nested_Bound_indirect_F} to the three common receiver case in the following corollary.

\begin{corollary}
\label{Corollary_threecommon_recs_FeasibilityProblem}
An inner bound of $K$-user DM BC for two nested messages with three common receivers ($M_{S_p}$,$M_{\overline{\phi}}$) where $S_p=\{1,2,\cdots,K-3\}$ is the set of rate pairs ($R_{S_p},R_{\overline{\phi}}$) satisfying the following for $l_1=K,l_2=K-1,l_3=K-2$ and all $j \in S_p$
\begin{align}
    &R_{\overline{\phi}}+R_{S_p\rightarrow\overline{l_2l_3}}+R_{S_p\rightarrow\overline{l_2}}+R_{S_p\rightarrow\overline{l_3}}+R_{S_p\rightarrow\overline{\phi}}
    \leq I(U_{\mathsf{W}_{l_1}^{\mathsf{P}}};Y_{l_3}|Q)
    \label{Eq:threecommon_recs_FeasibilityProblem_1}
    \\
    &R_{\overline{\phi}}+R_{S_p\rightarrow\overline{l_1l_3}}+R_{S_p\rightarrow\overline{l_1}}+R_{S_p\rightarrow\overline{l_3}}+R_{S_p\rightarrow\overline{\phi}}
    \leq I(U_{\mathsf{W}_{l_2}^{\mathsf{P}}};Y_{l_3}|Q)
    \label{Eq:threecommon_recs_FeasibilityProblem_2}
    \\
    &R_{\overline{\phi}}+R_{S_p\rightarrow\overline{l_1l_2}}+R_{S_p\rightarrow\overline{l_1}}+R_{S_p\rightarrow\overline{l_2}}+R_{S_p\rightarrow\overline{\phi}}
    \leq 
    I(U_{\mathsf{W}_{l_3}^{\mathsf{P}}};Y_{l_3}|Q)
    \label{Eq:threecommon_recs_FeasibilityProblem_3}
    \\
    &R_{\overline{l_1l_2l_3}}
    \leq 
    I(U_{  \downarrow_{\mathsf{W}_{j}^{\mathsf{P}}} \{\overline{l_1l_2l_3}\} };Y_{j}|U_{\mathsf{W}_{j}^{\mathsf{P}} \backslash  \downarrow_{\mathsf{W}_{j}^{\mathsf{P}}} \{\overline{l_1l_2l_3}\}  },Q)
    \label{Eq:threecommon_recs_FeasibilityProblem_4}
    \\
    &R_{S_p\rightarrow \overline{l_1l_2l_3}}+R_{S_p\rightarrow \overline{l_1l_2}}
    \leq I(U_{  \downarrow_{\mathsf{W}_{j}^{\mathsf{P}}} \{\overline{l_1l_2}\} };Y_{j}|U_{\mathsf{W}_{j}^{\mathsf{P}} \backslash  \downarrow_{\mathsf{W}_{j}^{\mathsf{P}}} \{\overline{l_1l_2}\}  },Q)
    \label{Eq:threecommon_recs_FeasibilityProblem_5}
    \\
    &R_{S_p\rightarrow \overline{l_1l_2l_3}}+R_{S_p\rightarrow \overline{l_1l_3}}
    \leq I(U_{  \downarrow_{\mathsf{W}_{j}^{\mathsf{P}}} \{\overline{l_1l_3}\} };Y_{j}|U_{\mathsf{W}_{j}^{\mathsf{P}} \backslash  \downarrow_{\mathsf{W}_{j}^{\mathsf{P}}} \{\overline{l_1l_3}\}  },Q)
    \label{Eq:threecommon_recs_FeasibilityProblem_6}
    \\
    &R_{S_p\rightarrow \overline{l_1l_2l_3}}+R_{S_p\rightarrow \overline{l_2l_3}}
    \leq I(U_{  \downarrow_{\mathsf{W}_{j}^{\mathsf{P}}} \{\overline{l_2l_3}\} };Y_{j}|U_{\mathsf{W}_{j}^{\mathsf{P}} \backslash  \downarrow_{\mathsf{W}_{j}^{\mathsf{P}}} \{\overline{l_2l_3}\}  },Q)
    \label{Eq:threecommon_recs_FeasibilityProblem_7}
    \\
    &R_{S_p\rightarrow \overline{l_1l_2l_3}}
    +R_{S_p\rightarrow \overline{l_1l_2}}
    +R_{S_p\rightarrow \overline{l_1l_3}}
    \leq I(U_{  \downarrow_{\mathsf{W}_{j}^{\mathsf{P}}} \{\overline{l_1l_2},\overline{l_1l_3}\} };Y_{j}|U_{\mathsf{W}_{j}^{\mathsf{P}} \backslash  \downarrow_{\mathsf{W}_{j}^{\mathsf{P}}} \{\overline{l_1l_2},\overline{l_1l_3}\}  },Q)
    \label{Eq:threecommon_recs_FeasibilityProblem_8}
    \\
    &R_{S_p\rightarrow \overline{l_1l_2l_3}}
    +R_{S_p\rightarrow \overline{l_1l_2}}
    +R_{S_p\rightarrow \overline{l_2l_3}}
    \leq I(U_{  \downarrow_{\mathsf{W}_{j}^{\mathsf{P}}} \{\overline{l_1l_2},\overline{l_2l_3}\} };Y_{j}|U_{\mathsf{W}_{j}^{\mathsf{P}} \backslash  \downarrow_{\mathsf{W}_{j}^{\mathsf{P}}} \{\overline{l_1l_2},\overline{l_2l_3}\}  },Q)
    \label{Eq:threecommon_recs_FeasibilityProblem_9}
    \\
    &R_{S_p\rightarrow \overline{l_1l_2l_3}}
    +R_{S_p\rightarrow \overline{l_1l_3}}
    +R_{S_p\rightarrow \overline{l_2l_3}}
    \leq I(U_{  \downarrow_{\mathsf{W}_{j}^{\mathsf{P}}} \{\overline{l_1l_3},\overline{l_2l_3}\} };Y_{j}|U_{\mathsf{W}_{j}^{\mathsf{P}} \backslash  \downarrow_{\mathsf{W}_{j}^{\mathsf{P}}} \{\overline{l_1l_3},\overline{l_2l_3}\}  },Q)
    \label{Eq:threecommon_recs_FeasibilityProblem_10}
    \\
    &R_{S_p\rightarrow \overline{l_1l_2l_3}}
    +R_{S_p\rightarrow \overline{l_1l_2}}
    +R_{S_p\rightarrow \overline{l_1l_3}}
    +R_{S_p\rightarrow \overline{l_2l_3}}
    \leq I(U_{  \downarrow_{\mathsf{W}_{j}^{\mathsf{P}}} \{\overline{l_1l_2},\overline{l_1l_3},\overline{l_2l_3}\} };Y_{j}|U_{\mathsf{W}_{j}^{\mathsf{P}} \backslash  \downarrow_{\mathsf{W}_{j}^{\mathsf{P}}} \{\overline{l_1l_2},\overline{l_1l_3},\overline{l_2l_3}\}  },Q)
    \label{Eq:threecommon_recs_FeasibilityProblem_11}
    \\
    & R_{S_p\rightarrow \overline{l_1l_2l_3}}
    +R_{S_p\rightarrow \overline{l_1l_2}}
    +R_{S_p\rightarrow \overline{l_1l_3}}
    +R_{S_p\rightarrow \overline{l_1}}
    \leq I(U_{  \downarrow_{\mathsf{W}_{j}^{\mathsf{P}}} \{\overline{l_1}\} };Y_{j}|U_{\mathsf{W}_{j}^{\mathsf{P}} \backslash  \downarrow_{\mathsf{W}_{j}^{\mathsf{P}}} \{\overline{l_1}\}  },Q)
    \label{Eq:threecommon_recs_FeasibilityProblem_12}
    \\
    &R_{S_p\rightarrow \overline{l_1l_2l_3}}
    +R_{S_p\rightarrow \overline{l_1l_2}}
    +R_{S_p\rightarrow \overline{l_1l_3}}
    +R_{S_p\rightarrow \overline{l_2l_3}}
    +R_{S_p\rightarrow \overline{l_1}}
    \leq I(U_{  \downarrow_{\mathsf{W}_{j}^{\mathsf{P}}} \{\overline{l_1},\overline{l_2l_3}\} };Y_{j}|U_{\mathsf{W}_{j}^{\mathsf{P}} \backslash  \downarrow_{\mathsf{W}_{j}^{\mathsf{P}}} \{\overline{l_1},\overline{l_2l_3}\}  },Q)
    \label{Eq:threecommon_recs_FeasibilityProblem_13}
    \\
    & R_{S_p\rightarrow \overline{l_1l_2l_3}}
    +R_{S_p\rightarrow \overline{l_1l_2}}
    +R_{S_p\rightarrow \overline{l_2l_3}}
    +R_{S_p\rightarrow \overline{l_2}}
    \leq I(U_{  \downarrow_{\mathsf{W}_{j}^{\mathsf{P}}} \{\overline{l_2}\} };Y_{j}|U_{\mathsf{W}_{j}^{\mathsf{P}} \backslash  \downarrow_{\mathsf{W}_{j}^{\mathsf{P}}} \{\overline{l_2}\}  },Q)
    \label{Eq:threecommon_recs_FeasibilityProblem_14}
    \\
    &R_{S_p\rightarrow \overline{l_1l_2l_3}}
    +R_{S_p\rightarrow \overline{l_1l_2}}
    +R_{S_p\rightarrow \overline{l_1l_3}}
    +R_{S_p\rightarrow \overline{l_2l_3}}
    +R_{S_p\rightarrow \overline{l_2}}
    \leq I(U_{  \downarrow_{\mathsf{W}_{j}^{\mathsf{P}}} \{\overline{l_2},\overline{l_1l_3}\} };Y_{j}|U_{\mathsf{W}_{j}^{\mathsf{P}} \backslash  \downarrow_{\mathsf{W}_{j}^{\mathsf{P}}} \{\overline{l_2},\overline{l_1l_3}\}  },Q)
    \label{Eq:threecommon_recs_FeasibilityProblem_15}
    \\
    &R_{S_p\rightarrow \overline{l_1l_2l_3}}
    +R_{S_p\rightarrow \overline{l_1l_3}}
    +R_{S_p\rightarrow \overline{l_2l_3}}
    +R_{S_p\rightarrow \overline{l_3}}
    \leq I(U_{  \downarrow_{\mathsf{W}_{j}^{\mathsf{P}}} \{\overline{l_3}\} };Y_{j}|U_{\mathsf{W}_{j}^{\mathsf{P}} \backslash  \downarrow_{\mathsf{W}_{j}^{\mathsf{P}}} \{\overline{l_3}\}  },Q)
    \label{Eq:threecommon_recs_FeasibilityProblem_16}
    \\
    &R_{S_p\rightarrow \overline{l_1l_2l_3}}
    +R_{S_p\rightarrow \overline{l_1l_2}}
    +R_{S_p\rightarrow \overline{l_1l_3}}
    +R_{S_p\rightarrow \overline{l_2l_3}}
    +R_{S_p\rightarrow \overline{l_3}}
    \leq I(U_{  \downarrow_{\mathsf{W}_{j}^{\mathsf{P}}} \{\overline{l_3},\overline{l_1l_2}\} };Y_{j}|U_{\mathsf{W}_{j}^{\mathsf{P}} \backslash  \downarrow_{\mathsf{W}_{j}^{\mathsf{P}}} \{\overline{l_3},\overline{l_1l_2}\}  },Q)
    \label{Eq:threecommon_recs_FeasibilityProblem_17}
    \\
    &R_{S_p\rightarrow \overline{l_1l_2l_3}}
    +R_{S_p\rightarrow \overline{l_1l_2}}
    +R_{S_p\rightarrow \overline{l_1l_3}}
    +R_{S_p\rightarrow \overline{l_2l_3}}
    +R_{S_p\rightarrow \overline{l_1}}
    +R_{S_p\rightarrow \overline{l_2}}
    \leq I(U_{  \downarrow_{\mathsf{W}_{j}^{\mathsf{P}}} \{\overline{l_1},\overline{l_2}\} };Y_{j}|U_{\mathsf{W}_{j}^{\mathsf{P}} \backslash  \downarrow_{\mathsf{W}_{j}^{\mathsf{P}}} \{\overline{l_1},\overline{l_2}\}  },Q)
    \label{Eq:threecommon_recs_FeasibilityProblem_18}
    \\
    &R_{S_p\rightarrow \overline{l_1l_2l_3}}
    +R_{S_p\rightarrow \overline{l_1l_2}}
    +R_{S_p\rightarrow \overline{l_1l_3}}
    +R_{S_p\rightarrow \overline{l_2l_3}}
    +R_{S_p\rightarrow \overline{l_1}}
    +R_{S_p\rightarrow \overline{l_3}}
    \leq I(U_{  \downarrow_{\mathsf{W}_{j}^{\mathsf{P}}} \{\overline{l_1},\overline{l_3}\} };Y_{j}|U_{\mathsf{W}_{j}^{\mathsf{P}} \backslash  \downarrow_{\mathsf{W}_{j}^{\mathsf{P}}} \{\overline{l_1},\overline{l_3}\}  },Q)
    \label{Eq:threecommon_recs_FeasibilityProblem_19}
    \\
    &R_{S_p\rightarrow \overline{l_1l_2l_3}}
    +R_{S_p\rightarrow \overline{l_1l_2}}
    +R_{S_p\rightarrow \overline{l_1l_3}}
    +R_{S_p\rightarrow \overline{l_2l_3}}
    +R_{S_p\rightarrow \overline{l_2}}
    +R_{S_p\rightarrow \overline{l_3}}
    \leq I(U_{  \downarrow_{\mathsf{W}_{j}^{\mathsf{P}}} \{\overline{l_2},\overline{l_3}\} };Y_{j}|U_{\mathsf{W}_{j}^{\mathsf{P}} \backslash  \downarrow_{\mathsf{W}_{j}^{\mathsf{P}}} \{\overline{l_2},\overline{l_3}\}  },Q)
    \label{Eq:threecommon_recs_FeasibilityProblem_20}
    \\
    &R_{S_p\rightarrow \overline{l_1l_2l_3}}
    +R_{S_p\rightarrow \overline{l_1l_2}}
    +R_{S_p\rightarrow \overline{l_1l_3}}
    +R_{S_p\rightarrow \overline{l_2l_3}}
    +R_{S_p\rightarrow \overline{l_1}}
    +R_{S_p\rightarrow \overline{l_2}}
    +R_{S_p\rightarrow \overline{l_3}}
    \leq \nonumber \\
    & \hspace{10cm}
    I(U_{  \downarrow_{\mathsf{W}_{j}^{\mathsf{P}}} \{\overline{l_1},\overline{l_2},\overline{l_3}\} };Y_{j}|U_{\mathsf{W}_{j}^{\mathsf{P}} \backslash  \downarrow_{\mathsf{W}_{j}^{\mathsf{P}}} \{\overline{l_1}, \overline{l_2},\overline{l_3}\}  },Q)
    \label{Eq:threecommon_recs_FeasibilityProblem_21}
    \\
    &R_{\overline{\phi}}+R_{S_p}
    \leq 
    I(U_{\mathsf{W}_{j}^{\mathsf{P}}};Y_{j}|Q)
    \label{Eq:threecommon_recs_FeasibilityProblem_22}
\end{align}
for some $p(q,u_{\mathsf{P}})= p(q) \prod_{S\in \mathsf{P}}p(u_S|u_{\uparrow_{\mathsf{P}} S\backslash \{ S \} },q)$ and $X$ as a deterministic function of $(Q, U_{\mathsf{P}})$ where
\begin{align}
R_{S_p\rightarrow S}&\geq 0 \; \forall S\in \uparrow_{\mathsf{P}}\{S_p\}    \label{Eq:threecommon_recs_FeasibilityProblem_23}
\\
R_{S_p}&= R_{S_p\rightarrow \overline{\phi}}+R_{S_p\rightarrow \overline{l_1}}+R_{S_p\rightarrow \overline{l_2}}+R_{S_p\rightarrow \overline{l_3}}+R_{S_p\rightarrow \overline{l_1l_2}}+R_{S_p\rightarrow \overline{l_1 l_3}}+R_{S_p\rightarrow \overline{l_2l_3}}+R_{S_p\rightarrow \overline{l_1l_2l_3}}
    \label{Eq:threecommon_recs_FeasibilityProblem_24}
\end{align}
    
\end{corollary}

\begin{proof}
We obtain the achievable rate region in the statement of the corollary by setting $\mathsf{F}=\mathsf{P}$ in Theorem \ref{Th_Inner_Nested_Bound_indirect_F}. Note that we split the private message $M_{S_p}$ into eight sub-messages $M_{S_p\rightarrow \overline{\phi}}$, $M_{S_p\rightarrow \overline{l_1}}$, $M_{S_p\rightarrow \overline{l_2}}$, $M_{S_p\rightarrow \overline{l_3}}$, $M_{S_p\rightarrow \overline{l_1l_2}}$, $M_{S_p\rightarrow \overline{l_1l_3}}$, $M_{S_p\rightarrow \overline{l_2l_3}}$, $M_{S_p\rightarrow \overline{l_1l_2l_3}}$ since there are 3 common receivers. Hence, we obtain \eqref{Eq:threecommon_recs_FeasibilityProblem_1}-\eqref{Eq:threecommon_recs_FeasibilityProblem_3} from \eqref{Eq_TH_indirect_generalF_common}, \eqref{Eq:threecommon_recs_FeasibilityProblem_4}-\eqref{Eq:threecommon_recs_FeasibilityProblem_21} from \eqref{Eq_TH_indirect_generalF_private2}, and  \eqref{Eq:threecommon_recs_FeasibilityProblem_22} from \eqref{Eq_TH_indirect_generalF_private1}. 
\end{proof}

In the following example, we consider the example in \cite[Example 2]{bidokhti2016capacity} wherein it was shown that the rate-splitting and linear coding scheme of \cite[Propositon 1]{bidokhti2016capacity} does not achieve capacity. However, when a linear pre-encoder is used together with that scheme a strictly larger rate region is obtained in \cite[Theorem 1]{bidokhti2016capacity} which achieves capacity in \cite[Example 2]{bidokhti2016capacity}. We show that the capacity region of \cite[Example 2]{bidokhti2016capacity} can actually be recovered in a top-down fashion from Corollary \ref{Corollary_threecommon_recs_FeasibilityProblem}, but two distributions of the auxilary random variables are needed.

\begin{example}[Example 2 of \cite{bidokhti2016capacity}]
\label{Example_Bidokhti_Example_K=6}
We consider $K=6$ with three common receivers, i.e., $\mathsf{E}=\{123,123456\}$ where $S_p=\{1,2,3\}$ and $S_l=\{4,5,6\}$. Moreover, we set $C_{124}=C_{135}=C_{236}=1$ and $C_S=0$ for all $S\in\mathsf{P}\backslash \{124,135,236\}$ so that we have only three intermediate nodes. Clearly, for such a connection, the capacity region is a line between the two corner points ($R_{S_p},R_{\overline{\phi}})=(2,0)$ and ($R_{S_p},R_{\overline{\phi}})=(0,1)$. By choosing the random variables $U_S$ for all $S\in\{124,135,236\}$ to be independent and uniformly distributed over $\mathcal{V}_S$ where $|\mathcal{V}_S|=2^{C_S}$ and $V_S=U_S$ with $|Q|=1$ and $U_S=\textit{const.}$ for $S\in\mathsf{P}\backslash \{124,135,236\}$ in the region given in Corollary \ref{Corollary_threecommon_recs_FeasibilityProblem}, we get that the rate pair ($R_{S_p},R_{\overline{\phi}})$ is achievable if there exists non-negative numbers $R_{S_p\rightarrow S}\geq 0 \; \forall S\in \uparrow_{\mathsf{P}}\{S_p\}$ with $R_{S_p}= R_{S_p\rightarrow \overline{\phi}}+R_{S_p\rightarrow \overline{6}}+R_{S_p\rightarrow \overline{5}}+R_{S_p\rightarrow \overline{4}}+R_{S_p\rightarrow \overline{56}}+R_{S_p\rightarrow \overline{46}}+R_{S_p\rightarrow \overline{45}}+R_{S_p\rightarrow \overline{456}}$ such that the following inequalities are satisfied:
\begin{align}
    R_{\overline{\phi}}+R_{S_p\rightarrow\overline{45}}+R_{S_p\rightarrow\overline{5}}+R_{S_p\rightarrow\overline{4}}+R_{S_p\rightarrow\overline{\phi}}
    &\leq 1
    \label{Eq:threecommon_recs_Sol1_FeasibilityProblem_1}
    \\
    R_{\overline{\phi}}+R_{S_p\rightarrow\overline{46}}+R_{S_p\rightarrow\overline{6}}+R_{S_p\rightarrow\overline{4}}+R_{S_p\rightarrow\overline{\phi}}
    &\leq 1 
    \label{Eq:threecommon_recs_Sol1_FeasibilityProblem_2}
    \\
    R_{\overline{\phi}}+R_{S_p\rightarrow\overline{56}}+R_{S_p\rightarrow\overline{6}}+R_{S_p\rightarrow\overline{5}}+R_{S_p\rightarrow\overline{\phi}}
    &\leq 1 
    \label{Eq:threecommon_recs_Sol1_FeasibilityProblem_3}
    \\
    R_{\overline{456}}
    &\leq 0 
    \label{Eq:threecommon_recs_Sol1_FeasibilityProblem_4}
    \\
    R_{S_p\rightarrow \overline{456}}+R_{S_p\rightarrow \overline{56}}
    &\leq 0 
    \label{Eq:threecommon_recs_Sol1_FeasibilityProblem_5}
    \\
    R_{S_p\rightarrow \overline{456}}+R_{S_p\rightarrow \overline{46}}
    &\leq 0 
    \label{Eq:threecommon_recs_Sol1_FeasibilityProblem_6}
    \\
    R_{S_p\rightarrow \overline{456}}+R_{S_p\rightarrow \overline{45}}
    &\leq 0 
    \label{Eq:threecommon_recs_Sol1_FeasibilityProblem_7}
    \\
    R_{S_p\rightarrow \overline{456}}
    +R_{S_p\rightarrow \overline{56}}
    +R_{S_p\rightarrow \overline{46}}
    &\leq 1 
    \label{Eq:threecommon_recs_Sol1_FeasibilityProblem_8}
    \\
    R_{S_p\rightarrow \overline{456}}
    +R_{S_p\rightarrow \overline{56}}
    +R_{S_p\rightarrow \overline{45}}
    &\leq 1 
    \label{Eq:threecommon_recs_Sol1_FeasibilityProblem_9}
    \\
    R_{S_p\rightarrow \overline{456}}
    +R_{S_p\rightarrow \overline{46}}
    +R_{S_p\rightarrow \overline{45}}
    &\leq 1 
    \label{Eq:threecommon_recs_Sol1_FeasibilityProblem_10}
    \\
    R_{S_p\rightarrow \overline{456}}
    +R_{S_p\rightarrow \overline{56}}
    +R_{S_p\rightarrow \overline{46}}
    +R_{S_p\rightarrow \overline{45}}
    &\leq 2 
    \label{Eq:threecommon_recs_Sol1_FeasibilityProblem_11}
    \\
    R_{S_p\rightarrow \overline{456}}
    +R_{S_p\rightarrow \overline{56}}
    +R_{S_p\rightarrow \overline{46}}
    +R_{S_p\rightarrow \overline{6}}
    &\leq 1 
    \label{Eq:threecommon_recs_Sol1_FeasibilityProblem_12}
    \\
    R_{S_p\rightarrow \overline{456}}
    +R_{S_p\rightarrow \overline{56}}
    +R_{S_p\rightarrow \overline{46}}
    +R_{S_p\rightarrow \overline{45}}
    +R_{S_p\rightarrow \overline{6}}
    &\leq 2 
    \label{Eq:threecommon_recs_Sol1_FeasibilityProblem_13}
    \\
     R_{S_p\rightarrow \overline{456}}
    +R_{S_p\rightarrow \overline{56}}
    +R_{S_p\rightarrow \overline{45}}
    +R_{S_p\rightarrow \overline{5}}
    &\leq 1 
    \label{Eq:threecommon_recs_Sol1_FeasibilityProblem_14}
    \\
    R_{S_p\rightarrow \overline{456}}
    +R_{S_p\rightarrow \overline{56}}
    +R_{S_p\rightarrow \overline{46}}
    +R_{S_p\rightarrow \overline{45}}
    +R_{S_p\rightarrow \overline{5}}
    &\leq 2 
    \label{Eq:threecommon_recs_Sol1_FeasibilityProblem_15}
    \\
    R_{S_p\rightarrow \overline{456}}
    +R_{S_p\rightarrow \overline{46}}
    +R_{S_p\rightarrow \overline{45}}
    +R_{S_p\rightarrow \overline{4}}
    &\leq 1 
    \label{Eq:threecommon_recs_Sol1_FeasibilityProblem_16}
    \\
    R_{S_p\rightarrow \overline{456}}
    +R_{S_p\rightarrow \overline{56}}
    +R_{S_p\rightarrow \overline{46}}
    +R_{S_p\rightarrow \overline{45}}
    +R_{S_p\rightarrow \overline{4}}
    &\leq 2 
    \label{Eq:threecommon_recs_Sol1_FeasibilityProblem_17}
    \\
    R_{S_p\rightarrow \overline{456}}
    +R_{S_p\rightarrow \overline{56}}
    +R_{S_p\rightarrow \overline{46}}
    +R_{S_p\rightarrow \overline{45}}
    +R_{S_p\rightarrow \overline{6}}
    +R_{S_p\rightarrow \overline{5}}
    &\leq 2 
    \label{Eq:threecommon_recs_Sol1_FeasibilityProblem_18}
    \\
    R_{S_p\rightarrow \overline{456}}
    +R_{S_p\rightarrow \overline{56}}
    +R_{S_p\rightarrow \overline{46}}
    +R_{S_p\rightarrow \overline{45}}
    +R_{S_p\rightarrow \overline{6}}
    +R_{S_p\rightarrow \overline{4}}
    &\leq 2 
    \label{Eq:threecommon_recs_Sol1_FeasibilityProblem_19}
    \\
    R_{S_p\rightarrow \overline{456}}
    +R_{S_p\rightarrow \overline{56}}
    +R_{S_p\rightarrow \overline{46}}
    +R_{S_p\rightarrow \overline{45}}
    +R_{S_p\rightarrow \overline{5}}
    +R_{S_p\rightarrow \overline{4}}
    &\leq 2 
    \label{Eq:threecommon_recs_Sol1_FeasibilityProblem_20}
    \\
    R_{S_p\rightarrow \overline{456}}
    +R_{S_p\rightarrow \overline{56}}
    +R_{S_p\rightarrow \overline{46}}
    +R_{S_p\rightarrow \overline{45}}
    +R_{S_p\rightarrow \overline{6}}
    +R_{S_p\rightarrow \overline{5}}
    +R_{S_p\rightarrow \overline{4}}
    &\leq 2 
    \label{Eq:threecommon_recs_Sol1_FeasibilityProblem_21}
    \\
    R_{\overline{\phi}}+R_{S_p}
    &\leq 2 
    \label{Eq:threecommon_recs_Sol1_FeasibilityProblem_22}
\end{align}
Clearly, for this choice of the auxiliary random variable, the rate pair ($R_{S_p},R_{\overline{\phi}})=(0,1)$ is achievable since the above linear program is feasible. However, the rate pair ($R_{S_p},R_{\overline{\phi}})=(2,0)$ is not achievable. To achieve the rate pair ($R_{S_p},R_{\overline{\phi}})=(2,0)$, we have to set, for instance, $R_{S_p\rightarrow \overline{46}}=R_{S_p\rightarrow \overline{56}}=1$. By doing so, we see that the above linear program is infeasible due to the inequalities \eqref{Eq:threecommon_recs_Sol1_FeasibilityProblem_5}, \eqref{Eq:threecommon_recs_Sol1_FeasibilityProblem_6}, \eqref{Eq:threecommon_recs_Sol1_FeasibilityProblem_8}, and \eqref{Eq:threecommon_recs_Sol1_FeasibilityProblem_12}. However, the inequality \eqref{Eq:threecommon_recs_Sol1_FeasibilityProblem_1} has a right hand side equal to $1$ with left hand side equal $0$. This inequality is the reliability condition for the common receiver $Y_6$. In fact, this common receiver does not benefit from setting $U_{236}=V_{236}$ since it does not have to decode any information. Hence, we have to change the choice of the auxiliary random variable to make the private receivers exploit $V_{236}$ to relax the inequalities \eqref{Eq:threecommon_recs_Sol1_FeasibilityProblem_5}, \eqref{Eq:threecommon_recs_Sol1_FeasibilityProblem_6}, \eqref{Eq:threecommon_recs_Sol1_FeasibilityProblem_8}, and \eqref{Eq:threecommon_recs_Sol1_FeasibilityProblem_12} and make the linear program feasible. To do so, we choose the auxiliary random variables $U_{135},U_{124},U_{123}$ to be independent and uniformly distributed over $\mathcal{V}_{135},\mathcal{V}_{124},\mathcal{V}_{236}$, respectively, and set $U_{135}=V_{135},U_{124}=V_{124},U_{123}=V_{236}$ with $|Q|=1$ and $U_S=\textit{const.}$ for $S\in\mathsf{P}\backslash \{123,124,135\}$. With this second choice of the random variable, we get from Corollary \ref{Corollary_threecommon_recs_FeasibilityProblem} that the rate pair ($R_{S_p},R_{\overline{\phi}})$ is achievable if there exist non-negative numbers $ R_{S_p\rightarrow S}\geq 0 \; \forall S\in \uparrow_{\mathsf{P}}\{S_p\} $ with $ R_{S_p}= R_{S_p\rightarrow \overline{\phi}}+R_{S_p\rightarrow \overline{6}}+R_{S_p\rightarrow \overline{5}}+R_{S_p\rightarrow \overline{4}}+R_{S_p\rightarrow \overline{56}}+R_{S_p\rightarrow \overline{46}}+R_{S_p\rightarrow \overline{45}}+R_{S_p\rightarrow \overline{456}} $ that satisfy the following inequalities

\begin{align}
    R_{\overline{\phi}}+R_{S_p\rightarrow\overline{45}}+R_{S_p\rightarrow\overline{5}}+R_{S_p\rightarrow\overline{4}}+R_{S_p\rightarrow\overline{\phi}}
    &\leq 0
    \label{Eq:threecommon_recs_Sol2_FeasibilityProblem_1}
    \\
    R_{\overline{\phi}}+  R_{S_p\rightarrow\overline{46}}+R_{S_p\rightarrow\overline{6}}+R_{S_p\rightarrow\overline{4}}+R_{S_p\rightarrow\overline{\phi}}
    &\leq 1 
    \label{Eq:threecommon_recs_Sol2_FeasibilityProblem_2}
    \\
    R_{\overline{\phi}}+R_{S_p\rightarrow\overline{56}}+R_{S_p\rightarrow\overline{6}}+R_{S_p\rightarrow\overline{5}}+R_{S_p\rightarrow\overline{\phi}}
    &\leq 1 
    \label{Eq:threecommon_recs_Sol2_FeasibilityProblem_3}
    \\
    R_{\overline{456}}
    &\leq 0 
    \label{Eq:threecommon_recs_Sol2_FeasibilityProblem_4}
    \\
    R_{S_p\rightarrow \overline{456}}+R_{S_p\rightarrow \overline{56}}
    &\leq 1 
    \label{Eq:threecommon_recs_Sol2_FeasibilityProblem_5}
    \\
    R_{S_p\rightarrow \overline{456}}+R_{S_p\rightarrow \overline{46}}
    &\leq 1 
    \label{Eq:threecommon_recs_Sol2_FeasibilityProblem_6}
    \\
    R_{S_p\rightarrow \overline{456}}+R_{S_p\rightarrow \overline{45}}
    &\leq 0 
    \label{Eq:threecommon_recs_Sol2_FeasibilityProblem_7}
    \\
    R_{S_p\rightarrow \overline{456}}
    +R_{S_p\rightarrow \overline{56}}
    +R_{S_p\rightarrow \overline{46}}
    &\leq 2 
    \label{Eq:threecommon_recs_Sol2_FeasibilityProblem_8}
    \\
    R_{S_p\rightarrow \overline{456}}
    +R_{S_p\rightarrow \overline{56}}
    +R_{S_p\rightarrow \overline{45}}
    &\leq 1 
    \label{Eq:threecommon_recs_Sol2_FeasibilityProblem_9}
    \\
    R_{S_p\rightarrow \overline{456}}
    +R_{S_p\rightarrow \overline{46}}
    +R_{S_p\rightarrow \overline{45}}
    &\leq 1 
    \label{Eq:threecommon_recs_Sol2_FeasibilityProblem_10}
    \\
    R_{S_p\rightarrow \overline{456}}
    +R_{S_p\rightarrow \overline{56}}
    +R_{S_p\rightarrow \overline{46}}
    +R_{S_p\rightarrow \overline{45}}
    &\leq 2 
    \label{Eq:threecommon_recs_Sol2_FeasibilityProblem_11}
    \\
    R_{S_p\rightarrow \overline{456}}
    +R_{S_p\rightarrow \overline{56}}
    +R_{S_p\rightarrow \overline{46}}
    +R_{S_p\rightarrow \overline{6}}
    &\leq 2 
    \label{Eq:threecommon_recs_Sol2_FeasibilityProblem_12}
    \\
    R_{S_p\rightarrow \overline{456}}
    +R_{S_p\rightarrow \overline{56}}
    +R_{S_p\rightarrow \overline{46}}
    +R_{S_p\rightarrow \overline{45}}
    +R_{S_p\rightarrow \overline{6}}
    &\leq 2 
    \label{Eq:threecommon_recs_Sol2_FeasibilityProblem_13}
    \\
     R_{S_p\rightarrow \overline{456}}
    +R_{S_p\rightarrow \overline{56}}
    +R_{S_p\rightarrow \overline{45}}
    +R_{S_p\rightarrow \overline{5}}
    &\leq 1 
    \label{Eq:threecommon_recs_Sol2_FeasibilityProblem_14}
    \\
    R_{S_p\rightarrow \overline{456}}
    +R_{S_p\rightarrow \overline{56}}
    +R_{S_p\rightarrow \overline{46}}
    +R_{S_p\rightarrow \overline{45}}
    +R_{S_p\rightarrow \overline{5}}
    &\leq 2 
    \label{Eq:threecommon_recs_Sol2_FeasibilityProblem_15}
    \\
    R_{S_p\rightarrow \overline{456}}
    +R_{S_p\rightarrow \overline{46}}
    +R_{S_p\rightarrow \overline{45}}
    +R_{S_p\rightarrow \overline{4}}
    &\leq 1 
    \label{Eq:threecommon_recs_Sol2_FeasibilityProblem_16}
    \\
    R_{S_p\rightarrow \overline{456}}
    +R_{S_p\rightarrow \overline{56}}
    +R_{S_p\rightarrow \overline{46}}
    +R_{S_p\rightarrow \overline{45}}
    +R_{S_p\rightarrow \overline{4}}
    &\leq 2 
    \label{Eq:threecommon_recs_Sol2_FeasibilityProblem_17}
    \\
    R_{S_p\rightarrow \overline{456}}
    +R_{S_p\rightarrow \overline{56}}
    +R_{S_p\rightarrow \overline{46}}
    +R_{S_p\rightarrow \overline{45}}
    +R_{S_p\rightarrow \overline{6}}
    +R_{S_p\rightarrow \overline{5}}
    &\leq 2 
    \label{Eq:threecommon_recs_Sol2_FeasibilityProblem_18}
    \\
    R_{S_p\rightarrow \overline{456}}
    +R_{S_p\rightarrow \overline{56}}
    +R_{S_p\rightarrow \overline{46}}
    +R_{S_p\rightarrow \overline{45}}
    +R_{S_p\rightarrow \overline{6}}
    +R_{S_p\rightarrow \overline{4}}
    &\leq 2 
    \label{Eq:threecommon_recs_Sol2_FeasibilityProblem_19}
    \\
    R_{S_p\rightarrow \overline{456}}
    +R_{S_p\rightarrow \overline{56}}
    +R_{S_p\rightarrow \overline{46}}
    +R_{S_p\rightarrow \overline{45}}
    +R_{S_p\rightarrow \overline{5}}
    +R_{S_p\rightarrow \overline{4}}
    &\leq 2 
    \label{Eq:threecommon_recs_Sol2_FeasibilityProblem_20}
    \\
    R_{S_p\rightarrow \overline{456}}
    +R_{S_p\rightarrow \overline{56}}
    +R_{S_p\rightarrow \overline{46}}
    +R_{S_p\rightarrow \overline{45}}
    +R_{S_p\rightarrow \overline{6}}
    +R_{S_p\rightarrow \overline{5}}
    +R_{S_p\rightarrow \overline{4}}
    &\leq 2 
    \label{Eq:threecommon_recs_Sol2_FeasibilityProblem_21}
    \\
    R_{\overline{\phi}}+R_{S_p}
    &\leq 2 
    \label{Eq:threecommon_recs_Sol2_FeasibilityProblem_22}
\end{align}
It can be verified that the rate pair ($R_{S_p},R_{\overline{\phi}})=(2,0)$ is now indeed achievable since  $R_{S_p\rightarrow \overline{46}}=R_{S_p\rightarrow \overline{56}}=1$ satisfy the above inequalities.
Note that this is not the only possible choice of auxiliary random variables. We can set $R_{S_p\rightarrow \overline{46}}=R_{S_p\rightarrow \overline{45}}=1$ or $R_{S_p\rightarrow \overline{45}}=R_{S_p\rightarrow \overline{56}}=1$ and change the choice of the random variables accordingly to keep the linear program feasible. 
\end{example}

\begin{remark}
Since by choosing the random variables $U_S$ for all $S\in \mathsf{P}$ to be independent and uniformly distributed over $\mathcal{V}_S$ where $|\mathcal{V}_S|=2^{C_S}$ and $V_S=U_S$ with $|Q|=1$ in Theorem \ref{Th_Inner_Nested_Bound_indirect_F}, we obtain the rate region of \cite[Proposition 1]{bidokhti2016capacity} which is not optimal for more than two common receivers as shown by the above example. 
Unlike in our approach, the right hand side of the linear program in \cite{bidokhti2016capacity} (equivalent here to \eqref{Eq:threecommon_recs_Sol1_FeasibilityProblem_1}-\eqref{Eq:threecommon_recs_Sol1_FeasibilityProblem_22}) is function only of the network connections and hence it cannot be changed to make the linear program feasible as we did in \eqref{Eq:threecommon_recs_Sol2_FeasibilityProblem_1}-\eqref{Eq:threecommon_recs_Sol2_FeasibilityProblem_22}. This is why the authors in \cite{bidokhti2016capacity} add a pre-encoder to the linear code to relax the left hand side of the linear program to get \cite[Theorem 1]{bidokhti2016capacity} that allows the value of $R_{S_p\rightarrow S_p}$ to be negative, i.e., relaxing the inequality in \eqref{Eq:threecommon_recs_FeasibilityProblem_23} to be $R_{S_p\rightarrow S}\geq 0 \; \forall S\in \uparrow_{\mathsf{P}}\{S_p\} \backslash\{S_p\}$. However, in our approach, we do not have to change our coding scheme, i.e., superposition coding with up-set rate splitting suffices. We have to just choose suitable auxiliary random variables to achieve the corner point ($R_{S_p},R_{\overline{\phi}})=(2,0)$.
\end{remark}

In the following example, we consider the same example as in \cite[Example 4]{bidokhti2016capacity} where the authors showed that their linear coding scheme even with the pre-encoder cannot achieve the rate pair ($R_{S_p},R_{\overline{\phi}})=(3,1)$ for a network with $K=7$ and four common receivers. Hence, a block Markov coding scheme was proposed therein to relax the left hand side of the linear program, and consequently, achieve the rate pair  ($R_{S_p},R_{\overline{\phi}})=(3,1)$. However, as suggested by Example \ref{Example_Bidokhti_Example_K=6}, instead of changing our coding scheme, we might be able to achieve the rate pair under consideration by finding a suitable choice of the auxiliary random variables in Theorem \ref{Th_Inner_Nested_Bound_indirect_F} to achieve this rate pair. This is what we do in the next example.

\begin{example}[Example 4 of \cite{bidokhti2016capacity}]
\label{Example_Bidokhti_Example_K=7}
We consider $K=7$ with four common receivers, i.e., $\mathsf{E}=\{123,1234567\}$ where $S_p=\{1,2,3\}$ and $S_l=\{4,5,6,7\}$. Moreover, we have only six intermediate nodes, i.e., $C_{1245}=C_{1257}=C_{1346}=C_{1347}=C_{2356}=C_{2367}=1$ and $C_S=0$ for all $S\in\mathsf{P}\backslash \{1245,1346,1347,2356,1257,2367\}$. Let us start by choosing $U_S$ for all $S\in \mathsf{P}$ to be independent and uniformly distributed over $\mathcal{V}_S$ where $|\mathcal{V}_S|=2^{C_S}$ with $|Q|=1$. Hence, from \eqref{Eq_TH_indirect_generalF_common} in Theorem \ref{Th_Inner_Nested_Bound_indirect_F}, assuming that $\mathsf{F}=\mathsf{P}$, we know that we have the following four inequalities in the feasibility problem (with $R_{S_p}= \sum_{S\in \uparrow_{\mathsf{P}} \{S_p\}} R_{S_p\rightarrow S}$)
\begin{align}
    R_{\overline{\phi}}+ R_{S_p\rightarrow 1234}+ R_{S_p\rightarrow 12347}+    R_{S_p\rightarrow 12346}+ R_{S_p\rightarrow 12345}+ R_{S_p\rightarrow \overline{5}}+ R_{S_p\rightarrow \overline{6}}+ R_{S_p\rightarrow \overline{7}}+ R_{S_p\rightarrow \overline{\phi}}&\leq 3 
    \label{Eq_Feasibility_problem_K=7_choice1_1}
    \\
    R_{\overline{\phi}}+ R_{S_p\rightarrow 1235}+ R_{S_p\rightarrow 12357}+ R_{S_p\rightarrow 12356}+ R_{S_p\rightarrow 12345}+ R_{S_p\rightarrow \overline{4}}+ R_{S_p\rightarrow \overline{6}}+ R_{S_p\rightarrow \overline{7}}+ R_{S_p\rightarrow \overline{\phi}}&\leq 3 
    \label{Eq_Feasibility_problem_K=7_choice1_2}
    \\
    R_{\overline{\phi}}+ R_{S_p\rightarrow 1236}+ R_{S_p\rightarrow 12367}+ R_{S_p\rightarrow 12356}+ R_{S_p\rightarrow 12346}+ R_{S_p\rightarrow \overline{4}}+ R_{S_p\rightarrow \overline{5}}+ R_{S_p\rightarrow \overline{7}}+ R_{S_p\rightarrow \overline{\phi}}&\leq 3 
    \label{Eq_Feasibility_problem_K=7_choice1_3}
    \\
    R_{\overline{\phi}}+ R_{S_p\rightarrow 1237}+ R_{S_p\rightarrow 12367}+ R_{S_p\rightarrow 12357}+ R_{S_p\rightarrow 12347}+ R_{S_p\rightarrow \overline{4}}+ R_{S_p\rightarrow \overline{5}}+ R_{S_p\rightarrow \overline{6}}+ R_{S_p\rightarrow \overline{\phi}}&\leq 3
    \label{Eq_Feasibility_problem_K=7_choice1_4}
\end{align}
To achieve the rate pair  ($R_{S_p},R_{\overline{\phi}})=(3,1)$, we do not need to have $3$ on the right hand side of all four inequalities \eqref{Eq_Feasibility_problem_K=7_choice1_1}-\eqref{Eq_Feasibility_problem_K=7_choice1_4}. Here, the private message is split into $16$ sub-messages since we have four common receivers. We have to set the sum rate of a subset of these sub-messages equal $3$ in a way that allows us to restrict at least one of the above inequalities. Clearly, we cannot set the sum rate of $M_{S_p\rightarrow \overline{4}},M_{S_p\rightarrow \overline{5}},M_{S_p\rightarrow \overline{6}},M_{S_p\rightarrow \overline{7}}$ to be $3$ while the sum rate of each three of them is upper bounded by $2$ from \eqref{Eq_Feasibility_problem_K=7_choice1_1}-\eqref{Eq_Feasibility_problem_K=7_choice1_4}. Hence, let set $R_{S_p\rightarrow 12345}=R_{S_p\rightarrow 12347}=R_{S_p\rightarrow 12357}=1$. This is one possible choice from many other choices. Obviously, for this choice we can restrict the inequality \eqref{Eq_Feasibility_problem_K=7_choice1_3}, i.e., make its right hand side equal $1$ instead of $3$. To do so, we set 
\begin{align}
    U_{1346}= V_{1346} \label{Choice2_Aux_1}\\
    U_{2356}= \text{const.}\\
    U_{2367}= \text{const.}
\end{align} instead of the original choice which was $U_{S}= V_{S}$ for $S\in\{1246,2356,2367\}$. 
To make sure that all private receivers $Y_1,Y_2,Y_3$ benefit from this restriction, we set 
\begin{align}
    U_{1235}=V_{2356}
    \label{Choice2_Aux_2}\\
    U_{1237}=V_{2367}
    \label{Choice2_Aux_3}
\end{align}

We have only three sub-messages with non-zero rates, namely, $R_{S_p\rightarrow 12345}=R_{S_p\rightarrow 12347}=R_{S_p\rightarrow 12357}=1$ and hence, we have only seven $\mathsf{B}$'s from \eqref{Eq_TH_indirect_generalF_private2} in Theorem \ref{Th_Inner_Nested_Bound_indirect_F} that give non-redundant inequalities. Thus, the achievable rate region, from Theorem \ref{Th_Inner_Nested_Bound_indirect_F}, is the set of rate pairs $(R_{123}, R_{\overline{\phi}})$ satisfying for all $j\in S_p$
\begin{align}
& R_{\overline{\phi}}+R_{123}  \leq  I(U_{\mathsf{W}_j^{\mathsf{P}}};Y_j|Q)  \label{Eq_AchievableRateRegion_K=7_1}
\\
& R_{123}\leq I( U_{ \downarrow_{\mathsf{W}_j^{\mathsf{P}}} \{12345,12347,12357\} } ;Y_{j}|U_{ \mathsf{W}_j^{\mathsf{P}} \backslash \downarrow_{\mathsf{W}_j^{\mathsf{P}}} \{12345,12347,12357\} } )
\label{Eq_AchievableRateRegion_K=7_2}
\\
& R_{123\rightarrow 12345}+R_{123\rightarrow 12347}\leq I( U_{ \downarrow_{\mathsf{W}_j^{\mathsf{P}}} \{12345,12347\} } ;Y_{j}|U_{ \mathsf{W}_j^{\mathsf{P}} \backslash \downarrow_{\mathsf{W}_j^{\mathsf{P}}} \{12345,12347\} } )
\label{Eq_AchievableRateRegion_K=7_3}
\\
& R_{123\rightarrow 12345}+R_{123\rightarrow 12357}\leq I( U_{ \downarrow_{\mathsf{W}_j^{\mathsf{P}}} \{12345,12357\} } ;Y_{j}|U_{ \mathsf{W}_j^{\mathsf{P}} \backslash \downarrow_{\mathsf{W}_j^{\mathsf{P}}} \{12345,12357\} } )
\label{Eq_AchievableRateRegion_K=7_4}
\\
&R_{123\rightarrow 12347}+R_{123\rightarrow 12357}\leq I( U_{ \downarrow_{\mathsf{W}_j^{\mathsf{P}}} \{12347,12357\} } ;Y_{j}|U_{ \mathsf{W}_j^{\mathsf{P}} \backslash \downarrow_{\mathsf{W}_j^{\mathsf{P}}} \{12347,12357\} } )
\label{Eq_AchievableRateRegion_K=7_5}
\\
& R_{123\rightarrow 12345}\leq I( U_{ \downarrow_{\mathsf{W}_j^{\mathsf{P}}} \{12345\} } ;Y_{j}|U_{ \mathsf{W}_j^{\mathsf{P}} \backslash \downarrow_{\mathsf{W}_j^{\mathsf{P}}} \{12345\} } )
\label{Eq_AchievableRateRegion_K=7_6}
\\
& R_{123\rightarrow 12347}\leq I( U_{ \downarrow_{\mathsf{W}_j^{\mathsf{P}}} \{12347\} } ;Y_{j}|U_{ \mathsf{W}_j^{\mathsf{P}} \backslash \downarrow_{\mathsf{W}_j^{\mathsf{P}}} \{12347\} } )
\label{Eq_AchievableRateRegion_K=7_7}
\\
& R_{123\rightarrow 12357}\leq I( U_{ \downarrow_{\mathsf{W}_j^{\mathsf{P}}} \{12357\} } ;Y_{j}|U_{ \mathsf{W}_j^{\mathsf{P}} \backslash \downarrow_{\mathsf{W}_j^{\mathsf{P}}} \{12357\} } )
\label{Eq_AchievableRateRegion_K=7_8}
\\
&R_{\overline{\phi}}+R_{123\rightarrow 12345}+R_{123\rightarrow 12347} \leq  I(U_{\mathsf{W}_4^{\mathsf{P}}};Y_4|Q) 
\label{Eq_AchievableRateRegion_K=7_9}
\\
&R_{\overline{\phi}}+R_{123\rightarrow 12345}+R_{123\rightarrow 12357} \leq  I(U_{\mathsf{W}_5^{\mathsf{P}}};Y_5|Q) 
\label{Eq_AchievableRateRegion_K=7_10}
\\
&R_{\overline{\phi}}\leq  I(U_{\mathsf{W}_4^{\mathsf{P}}};Y_6|Q) 
\label{Eq_AchievableRateRegion_K=7_11}
\\
&R_{\overline{\phi}}+R_{123\rightarrow 12347}+R_{123\rightarrow 12357} \leq  I(U_{\mathsf{W}_7^{\mathsf{P}}};Y_7|Q) 
\label{Eq_AchievableRateRegion_K=7_12}\\
&R_{123}=R_{123\rightarrow 12345}+R_{123\rightarrow 12347}+R_{123\rightarrow 12357}
\label{Eq_AchievableRateRegion_K=7_13}
\end{align}
for some $p(q,u_{\mathsf{P}})=p(q) \prod_{S\in \mathsf{F}}p(u_S|u_{\uparrow_{\mathsf{P}} S\backslash    \{S\}     },q)$ and some deterministic function $X$ of the time sharing and auxiliary random variables $(Q,U_{\mathsf{P}})$. Besides \eqref{Choice2_Aux_1}-\eqref{Choice2_Aux_3}, we have to assign $V_{1245},V_{1347},V_{1257}$ to the auxiliary random variable such that all the above inequalities \eqref{Eq_AchievableRateRegion_K=7_1}-\eqref{Eq_AchievableRateRegion_K=7_13} are satisfied simultaneously for ($R_{S_p},R_{\overline{\phi}})=(3,1)$. We choose the random variables $U_{12345},U_{12347},U_{12357},U_{1235},U_{1237},U_{1346}$ to be independent and uniformly distributed over $\mathcal{V}_{1245},\mathcal{V}_{1347},\mathcal{V}_{1257},\mathcal{V}_{2356},\mathcal{V}_{2367},\mathcal{V}_{1346}$, respectively, and set 
\begin{align}
&U_{12345}=V_{1245}\\
&U_{12347}=V_{1347}\\
&U_{12357}=V_{1257}\\
&U_{1235}=V_{2356}\\
&U_{1237}=V_{2367}\\
&U_{1346}=V_{1346}  
\end{align} 
while setting the rest of the auxiliary random variable to const. with $|Q|=1$. For this particular choice of auxiliary random variables, we can show that the linear program is \eqref{Eq_AchievableRateRegion_K=7_1}-\eqref{Eq_AchievableRateRegion_K=7_13} is feasible for the rate pair ($R_{S_p},R_{\overline{\phi}})=(3,1)$ with $R_{S_p\rightarrow 12345}=R_{S_p\rightarrow 12347}=R_{S_p\rightarrow 12357}=1$, and hence, achievable. It is worth noting that there are other choices for the auxiliary random variables that lead to a feasible linear program \eqref{Eq_AchievableRateRegion_K=7_1}-\eqref{Eq_AchievableRateRegion_K=7_13} for the rate pair ($R_{S_p},R_{\overline{\phi}})=(3,1)$ with $R_{S_p\rightarrow 12345}=R_{S_p\rightarrow 12347}=R_{S_p\rightarrow 12357}=1$. 

\end{example}

\begin{remark}
We showed in Example \ref{Example_Bidokhti_Example_K=6} and \ref{Example_Bidokhti_Example_K=7} that Theorem \ref{Th_Inner_Nested_Bound_indirect_F} achieves the rate pairs that needed a pre-encoder along with rate-splitting and linear superposition coding \cite[Theorem 1]{bidokhti2016capacity} or the proposed block Markov coding scheme \cite[Theorem 2]{bidokhti2016capacity} in \cite{bidokhti2016capacity} by just changing the choice of the auxiliary random variables. We did not have to include any dependency among the auxiliary random variables (which is admissible as seen in the allowable distributions in Theorem \ref{Th_Inner_Nested_Bound_indirect_F}). Note that in Example \ref{Example_Bidokhti_Example_K=7}, Theorem \ref{Th_Inner_Nested_Bound_indirect_F} achieves a rate pair that is not achieved by \cite[Theorem 1]{bidokhti2016capacity} which is known to be capacity achieving for three common receivers \cite[Theorem 4]{bidokhti2016capacity}. This leads us to conjecture that the achievable region in Theorem \ref{Th_Inner_Nested_Bound_indirect_F} is also capacity achieving for three common receivers. In fact, it might subsume the region obtained from linear coding, i.e., the transmitted signal is a linear combination of the information symbols since the region in Theorem \ref{Th_Inner_Nested_Bound_indirect_F} is the union of feasibility regions over all admissible coding distributions as specified therein.
\end{remark}

\begin{remark}
 Example \ref{Example_Bidokhti_Example_K=7} is the only example given in \cite{bidokhti2016capacity} that shows that block Markov coding {\em might} be needed for the more general DM BCs. This was the reason for the authors in \cite{bidokhti2016capacity} to have extended their block Markov coding scheme to the general DM BC in \cite[Theorem 6]{bidokhti2016capacity}. However, we showed in Example \ref{Example_Bidokhti_Example_K=7} that in that example, our coding scheme of Theorem \ref{Th_Inner_Nested_Bound_indirect_F} is enough to achieve the rate pair that needed the block Markov coding scheme of \cite[Theorem 2]{bidokhti2016capacity}. This Example \ref{Example_Bidokhti_Example_K=7} leaves us with no example that motivates the block Markov coding scheme of \cite[Theorem 2]{bidokhti2016capacity}, leading us to suggest that perhaps the block Markov coding scheme of \cite{bidokhti2016capacity} is never needed. To strengthen this suggestion, further work is needed.
\end{remark}

\section{Is $\mathsf{F}=\mathsf{P}$ necessary?}
\label{Sec:Smaller_F}

We have shown that by using a single coding scheme corresponding to $\mathsf{F}=\mathsf{P}$, we can achieve the capacity region of general (asymmetric) combination networks for three different message sets with uniform independent auxiliary random variables and a specification of the channel input via the simple relation $V_S=U_S$. In this part, we  consider the question of whether the capacity region be can achieved with a simpler coding scheme, i.e., with $\mathsf{F}\subset \mathsf{P}$? The answer to this question, as we show here, is yes. However, both $\mathsf{F}$, the expanded message set, and the coding distribution (i.e., the distribution of $U_\mathsf{F}$) and the function $X(U_\mathsf{F})$ must be tailored to the message set. We summarize the results of this section next.

For two order-$(K-1)$ messages, the coding scheme of Theorem \ref{TH_genral_inner_General_msgs} associated with $\mathsf{F}=\uparrow_{\mathsf{P}} \mathsf{E} \enspace \cup \enspace \{S_p\}$ is sufficient to achieve capacity. For two nested messages with one and two common receivers, the coding scheme of Theorem \ref{Th_Inner_Nested_Bound_indirect_F} corresponding to $\mathsf{F}=\mathsf{E}$ and $\mathsf{F}=\uparrow_{\mathsf{P}} \mathsf{E}$, respectively, are sufficient to achieve capacity. In particular, we show that these respective smaller $\mathsf{F}$'s can be used to achieve the capacity region provided the distribution of the auxiliary random variables $U_\mathsf{F}$ and the function $X(U_\mathsf{F})$ are chosen accordingly. In particular, a dependent set of auxiliary random variables must be considered. Hence, both the coding scheme and the coding distribution must be tailored to the message set under consideration. 

\subsection{$\mathsf{E}=\{\overline{K},\overline{\phi}\}$}

For two nested messages with one common receiver, we choose $\mathsf{F}=\mathsf{E}=\{\overline{K},\overline{\phi}\}$. With this specific choice, we get from Theorem \ref{Th_Inner_Nested_Bound_indirect_F} (with $|Q|=1$) that the rate pairs
\begin{align}
R_{\overline{\phi}}\leq &I(U_{\overline{\phi}};Y_K) \label{Region_C=1_F=E,1}\\
R_{\overline{\phi}}+R_{\overline{K}}\leq &I(U_{\overline{K}};Y_j) \enspace j\in S_p \label{Region_C=1_F=E,2}\\
R_{\overline{\phi}}+R_{\overline{K}}\leq &I(U_{\overline{K}};Y_j|U_{\overline{\phi}})+ I(U_{\overline{\phi}};Y_K) \enspace j\in S_p \label{Region_C=1_F=E,3}
\end{align}
for some $p(u_{\overline{\phi}})p(u_{\overline{K}}|u_{\overline{\phi}})$ are achievable where $X$ is a deterministic function of the auxiliary random variables. Note that the above region is the direct extension of that in \cite{kiirner1977general} from $K=2$ to arbitrary $K$. To achieve the capacity region of the combination network given by Proposition \ref{Th_Capacity_combination_networks_C1}, 
we choose $U_{\overline{\phi}}$ to be  uniformly distributed over $\prod_{S\in \mathsf{W}_K^{\mathsf{P}}} \mathcal{V}_S$ and $U_{\overline{K}}$ to be uniformly distributed over $\prod_{S\in \mathsf{P}} \mathcal{V}_S$ where $| \mathcal{V}_S|=2^{C_S}$ for any $S\in \mathsf{P}$. We set the channel input components $V_S$ for all $S\in\mathsf{P}$ to be independent and uniformly distributed over $\mathcal{V}_S$ where $V_{\mathsf{W}_K^{\mathsf{P}}}=U_{\overline{\phi}}$ and $V_{\mathsf{P}}=U_{\overline{K}}$. It is clear that for this specific choice of the channel input components $V_S$, we have \begin{equation}
H(V_{\mathsf{W}})=C_{\mathsf{W}} \label{eq:choice_of_Vs}
\end{equation} for any $\mathsf{W}\subseteq \mathsf{P}$ since the channel input components $V_S$ are independent and uniform distributed over $\mathcal{V}_S$.

Next, we compute the bounds in \eqref{Region_C=1_F=E,1}-\eqref{Region_C=1_F=E,3}.
Following a similar analysis as in Proposition \ref{Th_Capacity_combination_networks_C1}, we have \begin{equation}
I(U_{\overline{\phi}};Y_K) =C_{\mathsf{W}_K^{\mathsf{P}}} \label{eq:capacity_C=1,1}
\end{equation} and for each $j \in S_p$ we have
\begin{align}
I(U_{\overline{K}};Y_j) &=H(Y_j)\nonumber \\
&=H(V_{\mathsf{W}_j}^{\mathsf{P}})\nonumber \\
&= C_{W_j^{\mathsf{P}}} \label{eq:follows_choice_of_Vs}  
\end{align}
where \eqref{eq:follows_choice_of_Vs} follows directly from \eqref{eq:choice_of_Vs}. Hence, substituting \eqref{eq:capacity_C=1,1} and \eqref{eq:follows_choice_of_Vs} into \eqref{Region_C=1_F=E,1} and \eqref{Region_C=1_F=E,2}, we obtain the capacity region given in Proposition \ref{Th_Capacity_combination_networks_C1}. Moreover, we show that \eqref{Region_C=1_F=E,3} is redundant because for all $j \in S_p$
\begin{align}
I(U_{\overline{K}};Y_j|U_{\overline{\phi}})&=H(Y_j|U_{\overline{\phi}}) \nonumber\\
&=H(V_{\mathsf{W}_j^{\mathsf{P}}}|V_{\mathsf{W}_K^{\mathsf{P}}})\label{eq:Y_i_and_U}\\
&=H(V_{\mathsf{W}_j^{\mathsf{P}}\backslash \mathsf{W}_K^{\mathsf{P}}}|V_{\mathsf{W}_K^{\mathsf{P}}})\nonumber \\
&=H(V_{\downarrow_{\mathsf{W}_j^{\mathsf{P}}} \{\overline{K}\}}) \label{eq:Vs_indep2_and_sets} \\
&= C_{\downarrow_{\mathsf{W}_j^{\mathsf{P}}} \{\overline{K}\}} \label{eq:follows_from_choice_of_Vs}
\end{align}
where \eqref{eq:Y_i_and_U} follows from that fact that $Y_j=V_{\mathsf{W}_j}^{\mathsf{P}}$ and $U_{\overline{\phi}}=V_{\mathsf{W}_K}^{\mathsf{P}}$, \eqref{eq:Vs_indep2_and_sets} from 
 $\mathsf{W}_j^{\mathsf{P}}\backslash\mathsf{W}_K^{\mathsf{P}}= \downarrow_{\mathsf{W}_j^{\mathsf{P}}} \{\overline{K}\}$ and the independence among the channel input components, and \eqref{eq:follows_from_choice_of_Vs} from \eqref{eq:choice_of_Vs}. Note that for any $j \in S_p$
\begin{align}
C_{\mathsf{W}_j^{\mathsf{P}}}&= C_{\downarrow_{\mathsf{W}_j^{\mathsf{P}}} \{\overline{K}\} }+ C_{\uparrow_{\mathsf{W}_j^{\mathsf{P}}} \{K\} } \label{eq:lemma} \\
&\leq C_{\downarrow_{\mathsf{W}_j^{\mathsf{P}}} \{\overline{K}\} }+C_{\mathsf{W}_K^{\mathsf{P}}} \label{eq:sets}
\end{align}
where \eqref{eq:lemma} follows from Lemma \ref{Lemma_Combination_networks_proof} when $S=K$ and \eqref{eq:sets} from $\uparrow_{\mathsf{W}_j^{\mathsf{P}}} \{K\} \subseteq \mathsf{W}_K^{\mathsf{P}}$. The last inequality proves the redundancy of \eqref{Region_C=1_F=E,3}.

\subsection{$\mathsf{E}=\{\overline{\phi},\overline{K-1.K}\}$}

For two common receivers, i.e., $\mathsf{E}=\{\overline{\phi},\overline{K-1.K}\}$, we choose $\mathsf{F}=\uparrow_{\mathsf{P}}\mathsf{E}=\{\overline{\phi},\overline{K},\overline{K-1},\overline{K-1.K}\}$. Hence, from Theorem \ref{Th_Inner_Nested_Bound_indirect_F} (setting $|Q|=1$), we can show that for $j\in S_p$ the set of rate pairs ($R_{\overline{\phi}},R_{\overline{K-1.K}}$) satisfying 
\begin{align}
R_{\overline{\phi}}&\leq \min\{I(U_{\overline{K}},U_{\overline{\phi}};Y_{K-1}),I(U_{\overline{K-1}},U_{\overline{\phi}};Y_K)\}\label{Eq_Inner_Bound_C=2_F,1}\\
R_{\overline{\phi}}+R_{\overline{K-1.K}} &\leq\min\{I(U_{\overline{K-1.K}};Y_j), \nonumber \\ &\qquad \qquad \enspace I(U_{\overline{K-1.K}};Y_j|U_{\overline{\phi}},U_{\overline{K}})+I(U_{\overline{K}},U_{\overline{\phi}};Y_{K-1}), \nonumber \\ &\qquad \qquad \enspace I(U_{\overline{K-1.K}};Y_j|U_{\overline{\phi}},U_{\overline{K-1}})+I(U_{\overline{K-1}},U_{\overline{\phi}};Y_K)\} \label{Eq_Inner_Bound_C=2_F,2}\\
2R_{\overline{\phi}}+R_{\overline{K-1.K}} &\leq I(U_{\overline{K-1.K}};Y_j|U_{\overline{K}},U_{\overline{K-1}})+\nonumber \\& \qquad \qquad \enspace I(U_{\overline{K}},U_{\overline{\phi}};Y_{K-1})+I(U_{\overline{K-1}},U_{\overline{\phi}};Y_K)\label{Eq_Inner_Bound_C=2_F,3}\\
2R_{\overline{\phi}}+2R_{\overline{K-1.K}} &\leq I(U_{\overline{K-1.K}};Y_j|U_{\overline{K}},U_{\overline{K-1}})+  I(U_{\overline{K-1.K}};Y_j|U_{\overline{\phi}})+\nonumber \\& \qquad \qquad\enspace I(U_{\overline{K}},U_{\overline{\phi}};Y_{K-1})+I(U_{\overline{K-1}},U_{\overline{\phi}};Y_K) \label{Eq_Inner_Bound_C=2_F,4}
\end{align} for some $p(u_{\overline{\phi}})p(u_{\overline{K}}|u_{\overline{\phi}})p(u_{\overline{K-1}}|u_{\overline{\phi}})p(u_{\overline{K-1.K}}|u_{\overline{K}},u_{\overline{K-1}})$ and with $X$ a deterministic function $U_{\mathsf{F}}$, is achievable. The above region is the extension of the region in \cite[Proposition 5]{nair2009capacity} without binning from $K=3,L=2$ to arbitrary $K$ with two common receivers. In fact, as we mentioned before, the above region, with $K=3,L=2$, was shown in \cite[Proposition 7]{nair2009capacity} to be the capacity region when $Y_1$ is less noisy than $Y_2$ or $Y_3$. Here, we show that the above region is the capacity region for general combination networks with two common receivers and arbitrary $K$. 
We choose the dependent auxiliary random variables $U_{\overline{\phi}}$, $U_{\overline{K}}$, $U_{\overline{K-1}}$, and  $U_{\overline{K-1.K}}$ to be uniformly distributed over $\prod_{S\in \uparrow_{\mathsf{W}_{K-1}^{\mathsf{P}}}\{K\} } \mathcal{V}_S$, $\prod_{S\in \mathsf{W}_{K-1}^{\mathsf{P}}} \mathcal{V}_S$, and $\prod_{S\in \mathsf{W}_{K}^{\mathsf{P}}} \mathcal{V}_S$, and $\prod_{S\in\mathsf{P}}\mathcal{V}_S$,  respectively. On the other hand, we set the channel input components $V_S$ for all $S\in\mathsf{P}$ to be independent and uniformly distributed over $\mathcal{V}_S$ where 
\begin{align}
&V_{\uparrow_{\mathsf{W}_{K}^{\mathsf{P}}}\{K-1\}}=V_{\uparrow_{\mathsf{W}_{K-1}^{\mathsf{P}}}\{K\}}=U_{\overline{\phi}}\label{Eq:choice_C=2,1}\\
&V_{\uparrow_{\mathsf{W}_{K-1}^{\mathsf{P}}}\{K\}} , V_{\downarrow_{\mathsf{W}_{K-1}^{\mathsf{P}}}\{\overline{K}\}}=V_{\mathsf{W}_{K-1}^{\mathsf{P}}}=U_{\overline{K}}\label{Eq:choice_C=2,2}\\
&V_{\uparrow_{\mathsf{W}_{K}^{\mathsf{P}}}\{K-1\}} , V_{\downarrow_{\mathsf{W}_K^{\mathsf{P}}}\{\overline{K-1}\}}= V_{\mathsf{W}_{K}^{\mathsf{P}}}=U_{\overline{K-1}}\label{Eq:choice_C=2,3}\\
&V_{\mathsf{P}}=U_{\overline{K-1.K}}\label{Eq:choice_C=2,4}
\end{align} Note that the first equality in \eqref{Eq:choice_C=2,1} follows from the fact that $\uparrow_{\mathsf{W}_{K-1}^{\mathsf{P}}}\{K\}=\uparrow_{\mathsf{W}_{K}^{\mathsf{P}}}\{K-1\}$ and that in \eqref{Eq:choice_C=2,2} and \eqref{Eq:choice_C=2,3} directly from \eqref{Eq_Lemma_*} and \eqref{Eq_Lemma_**}.

For this choice of auxiliary random variables and channel input components and following similar analysis as in the one common receiver case, we have
\begin{align}
I(U_{\overline{K}},U_{\overline{\phi}};Y_{K-1})&=C_{\mathsf{W}_{K-1}^{\mathsf{P}}} \label{Values_ofInformation_Terms,1}\\
I(U_{\overline{K-1}},U_{\overline{\phi}};Y_K)&= C_{\mathsf{W}_{K}^{\mathsf{P}}}\\
I(U_{\overline{K-1.K}};Y_j)&=C_{\mathsf{W}_{j}^{\mathsf{P}}}\\
I(U_{\overline{K-1.K}};Y_j|U_{\overline{\phi}},U_{\overline{K}})&=C_{\downarrow_{\mathsf{W}_j^{\mathsf{P}}}\{\overline{K-1}\}}\\ I(U_{\overline{K-1.K}};Y_j|U_{\overline{\phi}},U_{\overline{K-1}})=& C_{\downarrow_{\mathsf{W}_j^{\mathsf{P}}}\{\overline{K}\}} \label{Values_ofInformation_Terms,5}\\
I(U_{\overline{K-1.K}};Y_j|U_{\overline{\phi}})&=C_{\downarrow_{\mathsf{W}_j^{\mathsf{P}}} \{\overline{K-1}, \overline{K} \}}\label{Values_ofInformation_Terms,6}\\
I(U_{\overline{K-1.K}};Y_j|U_{\overline{K}},U_{\overline{K-1}})  &= C_{\downarrow_{\mathsf{W}_j^{\mathsf{P}}} \{\overline{K-1.K} \}}\label{Values_ofInformation_Terms,7}
\end{align}

Notice that \eqref{Values_ofInformation_Terms,6} follows from 
\begin{align}
I(U_{\overline{K-1.K}};Y_j|U_{\overline{\phi}})&=H(Y_j|U_{\overline{\phi}})\nonumber \\
& = H(V_{\mathsf{W}_j^{\mathsf{P}}}|V_{\uparrow_{\mathsf{W}_{K-1}^{\mathsf{P}}}\{K\}})\label{Eq:proof_last_inequality1}\\
&= H(V_{\mathsf{W}_j^{\mathsf{P}}}|V_{\uparrow_{\mathsf{P}}\{K-1.K\}})\nonumber\\
&= H(V_{\mathsf{W}_j^{\mathsf{P}}}|V_{\uparrow_{\mathsf{P}}\{K-1.K\} \cap \mathsf{W}_j^{\mathsf{P}}  })\label{Eq:proof_last_inequality2}\\
&= H(V_{\mathsf{W}_j^{\mathsf{P}}}|V_{\uparrow_{\mathsf{W}_j^{\mathsf{P}}}\{K-1.K\}})\nonumber\\
&= H(V_{\mathsf{W}_j^{\mathsf{P}} \backslash \uparrow_{\mathsf{W}_{j}^{\mathsf{P}}}\{K-1.K\}})\label{Eq:proof_last_inequality3}\\
&=H(V_{\downarrow_{\mathsf{W}_j^{\mathsf{P}}} \{\overline{K-1},\overline{K}\}})\label{Eq:proof_last_inequality4}\\
& = C_{\downarrow_{\mathsf{W}_j^{\mathsf{P}}} \{\overline{K-1},\overline{K}\}}\label{Eq:proof_last_inequality5}
\end{align}
where \eqref{Eq:proof_last_inequality1} follows from \eqref{Eq:choice_C=2,1}, \eqref{Eq:proof_last_inequality2} and \eqref{Eq:proof_last_inequality3} from the independence of the channel input components $V_S$ for all $S\in\mathsf{P}$, \eqref{Eq:proof_last_inequality4} from \eqref{Eq_Lemma_*} and \eqref{Eq_Lemma_**}. Finally, \eqref{Eq:proof_last_inequality5} follows directly from \eqref{eq:choice_of_Vs}. 

On the other hand, \eqref{Values_ofInformation_Terms,7} can be shown as follows. For any $j\in S_p$, we have
\begin{align}
I(U_{\overline{K-1.K}};Y_j|U_{\overline{K}},U_{\overline{K-1}})&=H(Y_j|U_{\overline{K}},U_{\overline{K-1}})\nonumber \\
&\hspace{-1cm} = H(V_{\mathsf{W}_j^{\mathsf{P}}}|V_{\mathsf{W}_{K-1}^{\mathsf{P}}},V_{\mathsf{W}_K^{\mathsf{P}}}) \nonumber \\
&\hspace{-1cm} = H(V_{\mathsf{W}_j^{\mathsf{P}} \backslash \mathsf{W}_{K-1}^{\mathsf{P}}\cup \mathsf{W}_K^{\mathsf{P}}}) \label{Eq:proofC=2_inner_bound_3}\\
&\hspace{-1cm} = H(V_{\downarrow_{\mathsf{W}_j^{\mathsf{P}}} \{\overline{K-1.K}\}})\label{Eq:proofC=2_inner_bound_4}\\
&\hspace{-1cm} = C_{\downarrow_{\mathsf{W}_j^{\mathsf{P}}} \{\overline{K-1.K}\}}\label{Eq:proofC=2_inner_bound_5}
\end{align} where \eqref{Eq:proofC=2_inner_bound_3} follows from the independence of the channel input components $V_S$ for all $S\in\mathsf{P}$, \eqref{Eq:proofC=2_inner_bound_4} from removing from $\mathsf{W}_j^{\mathsf{P}}$ all the sets that contain $K$ or $K-1$, and \eqref{Eq:proofC=2_inner_bound_5} from \eqref{eq:choice_of_Vs}.

When we substitute \eqref{Values_ofInformation_Terms,1}-\eqref{Values_ofInformation_Terms,7} in \eqref{Eq_Inner_Bound_C=2_F,1}-\eqref{Eq_Inner_Bound_C=2_F,4}, we get an equivalent region to \eqref{Eq_Corollary_indirect_2common_1}-\eqref{Eq_Corollary_indirect_2common_6} for the choice of independent auxiliary random variables $U_S$ for all $S\in \mathsf{P}$ uniformly distributed over $\mathcal{V}_S$ where $|\mathcal{V}_S|=2^{C_S}$ and $V_S=U_S$. More precisely, the inequality in \eqref{Eq_Inner_Bound_C=2_F,1} becomes equivalent to \eqref{Eq_Corollary_indirect_2common_1}, \eqref{Eq_Inner_Bound_C=2_F,2} equivalent to \eqref{Eq_Corollary_indirect_2common_2}-\eqref{Eq_Corollary_indirect_2common_4}, \eqref{Eq_Inner_Bound_C=2_F,3} equivalent to \eqref{Eq_Corollary_indirect_2common_5} and \eqref{Eq_Inner_Bound_C=2_F,4} equivalent to \eqref{Eq_Corollary_indirect_2common_6}. Hence, the region in \eqref{Eq_Inner_Bound_C=2_F,1}-\eqref{Eq_Inner_Bound_C=2_F,4} achieves the capacity region established in Proposition \ref{Th_Capacity_combination_networks_C2}.

\subsection{$\mathsf{E}=\{\overline{K},\overline{K-1}\}$}

For the case two messages each required by $K-1$ receivers $\mathsf{E}=\{\overline{K},\overline{K-1}\}$, we set $\mathsf{F}=\uparrow_{\mathsf{P} }\mathsf{E}\cup \{S_p\}=\{\overline{\phi},\overline{K-1},\overline{K},\overline{K-1.K}\}$, where $S_p=\overline{K-1.K}$. Since we are choosing $\mathsf{F}$ strictly bigger that $\uparrow_{\mathsf{P} }\mathsf{E}$, we have one zero reconstruction rate as we mentioned before in Remark \ref{rem-largeF}. From Theorem \ref{TH_genral_inner_General_msgs}, the inner bound of $K$-user DM BC for the message index set $\mathsf{E}=\{\overline{K},\overline{K-1}\}$ is the set of rate pairs ($R_{\overline{K}},R_{\overline{K-1}}$) satisfying
\begin{align}
&R_{\overline{K-1}} \leq I(U_{\overline{K-1}},U_{\overline{\phi}};Y_K) \label{Region_InnerBound_C1_C2_SmallF_1}\\
&R_{\overline{K}} \leq  I(U_{\overline{K}},U_{\overline{\phi}};Y_{K-1})\label{Region_InnerBound_C1_C2_SmallF_2}\\
&R_{\overline{K-1}}+R_{\overline{K}} \leq I(U_{\overline{K-1.K}};Y_j)
\label{Region_InnerBound_C1_C2_SmallF_3} \\ &R_{\overline{K-1}}+R_{\overline{K}}\leq I(U_{\overline{K-1.K}};Y_j|U_{\overline{\phi}},U_{\overline{K}})+I(U_{\overline{K}},U_{\overline{\phi}};Y_{K-1}) 
\label{Region_InnerBound_C1_C2_SmallF_4} \\ &R_{\overline{K-1}}+R_{\overline{K}}
\leq I(U_{\overline{K-1.K}};Y_j|U_{\overline{\phi}},U_{\overline{K-1}})+I(U_{\overline{K-1}},U_{\overline{\phi}};Y_K)
\label{Region_InnerBound_C1_C2_SmallF_5} \\&R_{\overline{K-1}}+R_{\overline{K}} \leq I(U_{\overline{K}};Y_{K-1}|U_{\overline{\phi}})+I(U_{\overline{K-1}},U_{\overline{\phi}};Y_K)
\label{Region_InnerBound_C1_C2_SmallF_6} \\ &R_{\overline{K-1}}+R_{\overline{K}}
\leq I(U_{\overline{K-1}};Y_{K}|U_{\overline{\phi}})+I(U_{\overline{K}},U_{\overline{\phi}};Y_{K-1})  \label{Region_InnerBound_C1_C2_SmallF_7}\\
&2R_{\overline{K-1}}+2R_{\overline{K}}\leq   I(U_{\overline{K-1.K}};Y_j|U_{\overline{\phi}})+ I(U_{\overline{K}},U_{\overline{\phi}};Y_{K-1})+I(U_{\overline{K-1}},U_{\overline{\phi}};Y_K) \label{Region_InnerBound_C1_C2_SmallF_8}
\end{align} for all $j\in S_p$ and some $p(u_{\overline{\phi}})$ $p(u_{\overline{K}}|u_{\overline{\phi}})$ $p(u_{\overline{K-1}}|u_{\overline{\phi}})$ $p(u_{\overline{K-1.K}}|u_{\overline{K}},u_{\overline{K-1}})$ and $X$ a deterministic function of $U_{\mathsf{F}}$ is achievable. To achieve the capacity given in Theorem \ref{Th_Capacity_combination_networks_C1_1_C2_1}, we choose the auxiliary random variables and the channel input components exactly as in the two common receivers case. Hence, the equalities in \eqref{Values_ofInformation_Terms,1}-\eqref{Values_ofInformation_Terms,7} hold. Moreover, we have 
\begin{align}
I(U_{\overline{K}};Y_{K-1}|U_{\overline{\phi}})&=H(Y_{K-1}|U_{\overline{\phi}})\nonumber \\
&=H(V_{\mathsf{W}_{K-1}^{\mathsf{P}}}|V_{\uparrow_{\mathsf{W}_{K-1}^{\mathsf{P}}} \{K\} })\label{proof_information_1}\\
&= H(V_{\mathsf{W}_{K-1}^{\mathsf{P}}   \backslash \uparrow_{\mathsf{W}_{K-1}^{\mathsf{P}}} \{K\} })\label{proof_information_2}\\
&= H(V_{\downarrow_{\mathsf{W}_{K-1}^{\mathsf{P}}} \{\overline{K}\}})\label{proof_information_3}\\
&= C_{\downarrow_{\mathsf{W}_{K-1}^{\mathsf{P}}} \{\overline{K}\}}\label{proof_information_4}
\end{align} where \eqref{proof_information_1} follows from \eqref{Eq:choice_C=2,1}, \eqref{proof_information_2} from all channel input components are independent, \eqref{proof_information_3} from \eqref{Eq_Lemma_*} and \eqref{Eq_Lemma_**}, and \eqref{proof_information_4} from \eqref{eq:choice_of_Vs}.

Similarly, we can show that \begin{equation}
I(U_{\overline{K-1}};Y_{K}|U_{\overline{\phi}})=C_{\downarrow_{\mathsf{W}_{K}^{\mathsf{P}}} \{\overline{K-1}\}} \label{proof_information_5}
\end{equation}

By substituting \eqref{Values_ofInformation_Terms,1}-\eqref{Values_ofInformation_Terms,6}, \eqref{Eq:proof_last_inequality5}, \eqref{proof_information_4}, and \eqref{proof_information_5}, in the region \eqref{Region_InnerBound_C1_C2_SmallF_1}-\eqref{Region_InnerBound_C1_C2_SmallF_8}, we get exactly equivalent region to the one obtained from \eqref{Eq_Corollary_indirect_C1_1_C2_1_1}-\eqref{Eq_Corollary_indirect_C1_1_C2_1_5}, when the auxiliary random variables $U_S$ for all $S\in \mathsf{P}$ are chosen to be independent and uniformly distributed over $\mathcal{V}_S$ where $|\mathcal{V}_S|=2^{C_S}$ and $V_S=U_S$, which is the capacity region of the combination network given in Theorem \ref{Th_Capacity_combination_networks_C1_1_C2_1}. More precisely, the inequalities \eqref{Region_InnerBound_C1_C2_SmallF_1}-\eqref{Region_InnerBound_C1_C2_SmallF_3} and \eqref{Region_InnerBound_C1_C2_SmallF_8} become equivalent to \eqref{Eq_Corollary_indirect_C1_1_C2_1_1}-\eqref{Eq_Corollary_indirect_C1_1_C2_1_2} and \eqref{Eq_Corollary_indirect_C1_1_C2_1_5}. Also, \eqref{Region_InnerBound_C1_C2_SmallF_4} and \eqref{Region_InnerBound_C1_C2_SmallF_7} become equivalent to \eqref{Eq_Corollary_indirect_C1_1_C2_1_4}. Finally, \eqref{Region_InnerBound_C1_C2_SmallF_5} and \eqref{Region_InnerBound_C1_C2_SmallF_6} become equivalent to \eqref{Eq_Corollary_indirect_C1_1_C2_1_3}.

\begin{remark}
The most commonly used coding scheme for DM BCs for any message index set $\mathsf{E}$ is the one corresponding to $\mathsf{F}= \uparrow_{\mathsf{P}} \mathsf{E}$. If we restrict ourselves here to this kind of coding scheme, it is unclear how we can find the optimal distribution for the auxiliary random variables to achieve the capacity region for general combination networks with $\mathsf{E}=\{\overline{K},\overline{K-1}\}$. However, with a slight increase in the complexity of the coding scheme, i.e., setting $\mathsf{F}=\uparrow_{\mathsf{P} }\mathsf{E}\cup \{S_p\}$, we were able to identify the optimal distribution for the auxiliary random variables.
\end{remark}

In this section, we showed that there is a trade-off between the complexity of the coding scheme by changing $\mathsf{F}$ and that of the distribution of the auxiliary random variables and the encoding function that must be chosen to achieve the capacity region in three scenarios. We needed to consider {\em dependent} auxiliary random variables when we choose $\mathsf{F}\subset \mathsf{P}$ while independent auxiliary random variables were enough when $\mathsf{F}=\mathsf{P}$ as in Theorem \ref{Th_Capacity_combination_networks_C1_1_C2_1} and Propositions \ref{Th_Capacity_combination_networks_C1}, and \ref{Th_Capacity_combination_networks_C2}. It is interesting to ask if it is always true that using the most complex coding scheme by setting $\mathsf{F}=\mathsf{P}$ compensates for the need for dependent auxiliary random variables? This is possible to be true depending on Theorem \ref{Th_Capacity_combination_networks_C1_1_C2_1} and Propositions \ref{Th_Capacity_combination_networks_C1}, and \ref{Th_Capacity_combination_networks_C2} and Examples \ref{Example_Bidokhti_Example_K=6} and \ref{Example_Bidokhti_Example_K=7}. If this is true, then finding the optimal distribution for the auxiliary random variables to achieve a corner point is the capacity region could be much simpler than finding it over all admissible dependent distributions for the auxiliary random variables. This question needs further investigation.

\section{Conclusion}
\label{Sec_Conc}

In this paper, we propose a novel and general achievable scheme for the $K$-receiver DM BC with two groupcast messages that involves new twists and generalizations of the techniques of message splitting, superposition coding and indirect decoding. The language of order theory is used to describe it succinctly and to characterize its achievable rate region. To demonstrate the efficacy of this scheme we obtain its specialization to the combination network and show that in the three special cases of (a) nested messages with one common receiver, (b) nested messages with two common receivers and (c) with two messages each required by $K{-}1$ receivers the proposed achievable rate region coincides with the capacity region. In particular, the descriptions of the capacity regions are given as explicit polygons that reveals their combinatorial structure.

It remains to be seen if Theorem \ref{Th_Inner_Nested_Bound_indirect_F}, when specialized to the $K$-user combination network, yields its capacity region for {\em any} two nested groupcast messages. More generally, we are curious to know if Theorem \ref{TH_genral_inner_General_msgs} might yield the capacity region of the combination network for any two arbitrary groupcast messages. In future work, it is also of interest to generalize the results of this work in the direction of expanding message sets to contain more than two messages.


\begin{appendices}
\section{Proof of (\ref{Eq_Lemma_*}) and (\ref{Eq_Lemma_***}) in Lemma \ref{Lemma_Sets} }
\label{App_Prove_Lemma_Sets}
For any set $S=i_1i_2\cdots i_N \subset \{1,2,\cdots, K\}$ and $i\in [1:K]$ , we have
\begin{align}
\mathsf{W}_i^{\mathsf{P}}&= \mathsf{W}_i^{\mathsf{P}} \backslash \uparrow_{\mathsf{W}_i^\mathsf{P}}\{S\} \enspace  \cup \enspace  \mathsf{W}_i^{\mathsf{P}} \cap \uparrow_{\mathsf{W}_i^\mathsf{P}}\{S\}\nonumber \\
&= \mathsf{W}_i^{\mathsf{P}} \backslash \uparrow_{\mathsf{W}_i^\mathsf{P}}\{S\}  \enspace  \cup \enspace \uparrow_{\mathsf{W}_i^\mathsf{P}}\{S\}\label{Eq_Lemma_proof_1_1}\\
&= \downarrow_{\mathsf{W}_i^\mathsf{P}}\{ \overline{1},\overline{2},\cdots ,\overline{K}\}  \backslash \uparrow_{\mathsf{W}_i^\mathsf{P}}\{S\}  \enspace  \cup \enspace  \uparrow_{\mathsf{W}_i^\mathsf{P}}\{S\} \label{Eq_Lemma_proof_1_2}\\
&= \downarrow_{\mathsf{W}_i^\mathsf{P}}\{ \overline{1},\cdots ,\overline{i-1},\overline{i+1},\cdots, \overline{K} \}  \backslash \uparrow_{\mathsf{W}_i^\mathsf{P}}\{S\}  \enspace  \cup \enspace  \uparrow_{\mathsf{W}_i^\mathsf{P}}\{S\} \label{Eq_Lemma_proof_1_3}\\
&= \downarrow_{\mathsf{W}_i^{\mathsf{P}}}\{\overline{i_1},\overline{i_2},\cdots, \overline{i_N}\} \backslash \uparrow_{\mathsf{W}_i^\mathsf{P}}\{S\}  \enspace  \cup \enspace  \uparrow_{\mathsf{W}_i^\mathsf{P}}\{S\} \label{Eq_Lemma_proof_1_4}\\
&= \downarrow_{\mathsf{W}_i^{\mathsf{P}}}\{\overline{i_1},\overline{i_2},\cdots, \overline{i_N}\}  \enspace  \cup \enspace  \uparrow_{\mathsf{W}_i^\mathsf{P}}\{S\} \label{Eq_Lemma_proof_1_5}
\end{align}
where \eqref{Eq_Lemma_proof_1_1} follows from $\uparrow_{\mathsf{W}_i^\mathsf{P}}\{S\} \subseteq \mathsf{W}_i^{\mathsf{P}}$. For \eqref{Eq_Lemma_proof_1_2}, we know that the set $12\cdots K\in \uparrow_{\mathsf{W}_i^\mathsf{P}}\{S\}$ for any $S=i_1i_2\cdots i_N \subseteq \{1,2,\cdots, K\}$. Also, $\mathsf{W}_i^{\mathsf{P}}=\{12\cdots K\} \cup \downarrow_{\mathsf{W}_i^\mathsf{P}}\{ \overline{1},\overline{2},\cdots , \overline{K} \}$, in words, this means that $\mathsf{W}_i^{\mathsf{P}}$ is the union of the down-set of all $K$ sets with cardinality $K-1$ and the set with cardinality $K$ (i.e, $\{1,2,\cdots, K\}$). Hence, we can replace $\mathsf{W}_i^{\mathsf{P}}$ by $\downarrow_{\mathsf{W}_i^\mathsf{P}}\{ \overline{1},\overline{2},\cdots , \overline{K} \}$ since $ 12\cdots K\in \uparrow_{\mathsf{W}_i^\mathsf{P}}\{S\}$. Then, \eqref{Eq_Lemma_proof_1_3} follows from the fact that $\cup_{S\in \mathsf{W}}\downarrow_{\mathsf{W}_i^{\mathsf{P}}}\{S\}= \downarrow_{\mathsf{W}_i^{\mathsf{P}} }\{\cup_{S\in \mathsf{W}}\}$ for any $\mathsf{W}\subseteq \mathsf{P} $ and $\downarrow_{\mathsf{W}_i^{\mathsf{P}}}\{\overline{i}\}=\phi $. In \eqref{Eq_Lemma_proof_1_4}, we remove all the sets in $\{ \overline{1}, \overline{2},\cdots ,\overline{K}\}$ that are available in $\uparrow_{\mathsf{W}_i^\mathsf{P}}\{S\}$ for $N<K$. Note that we have $K$ sets with cardinality $K{-}1$ in $\mathsf{P}$ and $K{-}N$ sets with cardinality $K{-}1$ contains $S=i_1 i_2 \cdots i_N$. Finally, \eqref{Eq_Lemma_proof_1_5} follows from $\downarrow_{\mathsf{W}_i^{\mathsf{P}}}\{\overline{i_1},\overline{i_2},\cdots, \overline{i_N}\} \cap \uparrow_{\mathsf{W}_i^\mathsf{P}}\{i_1i_2\cdots i_N\}=\phi$.

In a similar manner, we can prove \eqref{Eq_Lemma_***} as follows

\begin{align}
\mathsf{W}_i^{\mathsf{P}}&= \mathsf{W}_i^{\mathsf{P}} \backslash \uparrow_{\mathsf{W}_i^\mathsf{P}}\{i_1,i_2,\cdots,i_N\} \enspace  \cup \enspace  \mathsf{W}_i^{\mathsf{P}} \cap \uparrow_{\mathsf{W}_i^\mathsf{P}}\{i_1,i_2,\cdots,i_N\}\nonumber \\
&= \downarrow_{\mathsf{W}_i^\mathsf{P}}\{ 123\cdots K\}  \backslash \uparrow_{\mathsf{W}_i^\mathsf{P}}\{i_1,i_2,\cdots,i_N\}  \enspace  \cup \enspace  \uparrow_{\mathsf{W}_i^\mathsf{P}}\{i_1,i_2,\cdots,i_N\} 
\label{Eq_Lemma_proof_2_1}\\
&= \downarrow_{\mathsf{W}_i^\mathsf{P}}\{ \overline{S}\}  \backslash \uparrow_{\mathsf{W}_i^\mathsf{P}}\{i_1,i_2,\cdots,i_N\}  \enspace  \cup \enspace  \uparrow_{\mathsf{W}_i^\mathsf{P}}\{i_1,i_2,\cdots,i_N\} 
\nonumber\\
&= \downarrow_{\mathsf{W}_i^{\mathsf{P}}}\{\overline{S}\}  \enspace  \cup \enspace  \uparrow_{\mathsf{W}_i^\mathsf{P}}\{i_1,i_2,\cdots,i_N\} 
\label{Eq_Lemma_proof_2_2}
\end{align}
where \eqref{Eq_Lemma_proof_2_1} from the fact that $\mathsf{W}_i^{\mathsf{P}}= \downarrow_{\mathsf{W}_i^\mathsf{P}}\{ 123\cdots K\}$, and \eqref{Eq_Lemma_proof_2_2} from \eqref{Eq_Lemma_****} where  $\downarrow_{\mathsf{W}_i^{\mathsf{P}}}\{\overline{S}\} \cap \uparrow_{\mathsf{W}_i^\mathsf{P}}\{i_1,i_2,\cdots, i_N\}=\phi$. 

\section{Proof of Theorem \ref{TH_genral_inner_General_msgs} }
\label{App_Prove_General_msgs}

Fix $\mathsf{F}$ ($\mathsf{E}\subseteq \mathsf{F}\subseteq \mathsf{P}$) to be ordered by set inclusion such that $S^{'}\leq S$ only if $S^{'}\subseteq S$. For this choice of $\mathsf{F}$, we do the following: ($a$) split the messages $M_S$ $S\in \mathsf{E}$ using the up-set splitting technique proposed in \cite{romero2017rate} such that 
\begin{equation}
M_S= ( M_{S \rightarrow S^{'}}, \enspace  S^{'}\in \uparrow_{\mathsf{F}} S )
\label{Eq_message_splitting}
\end{equation}
($b$) create the reconstruction messages $\hat{M}_S$ $S\in \mathsf{F}$ with rates given in \eqref{Eq_reconstruction_rates_upset_splitting} such that 
\begin{equation}
\hat{M}_S= ( M_{S^{'} \rightarrow S}, \enspace S^{'}\in \downarrow_{\mathsf{E}}S )
\label{Eq_message_reconstruction}
\end{equation}
and $(c)$ fix the coded time-sharing, auxiliary and input random variables ($Q, U_{\mathsf{F}}, X$) such that $X$ is a deterministic function of $Q, U_{\mathsf{F}}$ whose joint distribution is given by 
\begin{equation}
 p(q, u_{\mathsf{F}})= p(q) \prod_{S\in \mathsf{F}}p(u_S|u_{\uparrow_{\mathsf{F},q} S\backslash \{S \}},q)
\end{equation}


Then, we enumerate the sets in $\mathsf{F}$ in non-increasing order such that $\mathsf{F}=\{S_{i_1},S_{i_2},...,S_{i_N}\}$, where $N=|\mathsf{F}|$. Generate a random time-sharing sequence $q^n$ according to $\prod_{o=1}^n p_Q(q_i)$. For each $j\in [1:N]$ and collection of reconstruction messages $\hat{m}_{\uparrow_{\mathsf{F}}S_{i_j}}$, generate $2^{n\hat{R}_{S_{i_j}}}$ codewords $u_{S_{i_j}}^n(\hat{m}_{\uparrow_{\mathsf{F}}S_{i_j}})$ according to $\prod_{t=1}^n p(u_{S_{i_j}}|u_{\uparrow_{\mathsf{F}} S_{i_j}},q_i)$. This process can be done for all $j$ from $j=1$ to $j=N$.

Receivers $Y_j$ ($j\in  S_p $) jointly decode the reconstruction messages $\hat{M}_{\mathsf{W}_j^{\mathsf{F}} }$ which contain both desired messages ($M_{S_1},M_{S_2}$). Using the result in \cite{romero2017rate}, we can show that the probability of error vanishes if 
\begin{equation}
\sum_{S^{'}\in \mathsf{B}}  \hat{R}_{S^{'}} \leq I(U_{\mathsf{B}}; Y_j|U_{\mathsf{W}_j^{\mathsf{F}}\backslash \mathsf{B}   },Q) \quad \forall \mathsf{B}\in \mathcal{F}_{\downarrow}(\mathsf{W}_j^{\mathsf{F}}),j\in S_p  
\end{equation}

On the other hand, the receiver $Y_j$, with $j\in S_{l_1} \cup S_{l_2}$, only needs one message $M_{\mathsf{W}_j^{\mathsf{E}}}$, 
and hence, non-unique decoding can be employed. To analyze the error probabilities and obtain the conditions such that these probabilities vanish in the limit of large block length when non-unique decoding is used, we use the following two steps; $(a)$ obtain the conditions such that the probabilities of error vanish when unique decoding is used, and then $(b)$ remove all the inequalities that contain only the rates of the undesired messages. For part $(a)$, we know that the probability of error vanishes if 
\begin{equation}
\sum_{S^{'}\in \mathsf{B}}  \hat{R}_{S^{'}} \leq I(U_{\mathsf{B}}; Y_j|U_{\mathsf{W}_j^{\mathsf{F}}\backslash \mathsf{B}   },Q) \quad \forall \mathsf{B}\in \mathcal{F}_{\downarrow}(\mathsf{W}_j^{\mathsf{F}}),j\in S_{l_1} \cup S_{l_2}
\label{Eq_relaiable_condition_commonRec}
\end{equation}
Since each non-private receiver $Y_j$ desires only the message $M_{\mathsf{W}_j^{\mathsf{E}}}$ 
which is part of the reconstruction messages $\hat{M}_{S}$ with $S\in \uparrow_{\mathsf{F}} \mathsf{W}_j^{\mathsf{E}}$ as shown in \eqref{Eq_message_splitting} and \eqref{Eq_message_reconstruction}. Hence, from \eqref{Eq_relaiable_condition_commonRec}, we need to remove from $\mathsf{B}$ all the sets that do not contain $\uparrow_{\mathsf{F}} \mathsf{W}_j^{\mathsf{E}}$. Since all the sets in  $\mathsf{B}$ are down-sets, then if none of these sets contain $\mathsf{W}_j^{\mathsf{E}}$, they must also do not contain any of $\uparrow_{\mathsf{F}} \mathsf{W}_j^{\mathsf{E}}$. Thus, replacing $\mathsf{B} \in \mathcal{F}_{\downarrow}(\mathsf{W}_j^{\mathsf{F}})$ in \eqref{Eq_relaiable_condition_commonRec} by $\mathsf{B}\in \mathcal{F}_{\downarrow_{  \{S_1\} }}(\mathsf{W}_j^{\mathsf{F}}) $ for $j\in S_{l_1}$ and $\mathsf{B}\in \mathcal{F}_{\downarrow_{  \{S_2\} }}(\mathsf{W}_j^{\mathsf{F}}) $ for $j\in S_{l_2}$ removes all the inequalities that contain only the rates of the undesired messages. Hence, all receivers can reliably decode their desired messages if \eqref{Eq_condition_reconstruction_rates1} and \eqref{Eq_condition_reconstruction_rates2} are satisfied. This completes the proof of Theorem \ref{TH_genral_inner_General_msgs}.

\section{Proof of the converse for Theorem \ref{Th-Capacity_Region_Two_K-1_Order_msgs}}
   \label{Appendix_Proof_Th_Two_K-1_Order_msgs}
 The converse proofs in this section use the information inequality of \cite[Lemma 1]{nair2011capacity}, and hence, we state it here for easy reference.
\begin{lemma}
\label{Lemma_Information_Inequality}
Let $X {\markov} (Y,Z)$ be a DM BC without feedback and $Y  \succeq   Z$. Consider $M$ to be any random variable such that
$M {\markov}X^n {\markov} (Y^n,Z^n)$ forms a Markov chain. Then, we have \begin{align}
I(Y^{i-1};Z_{i}|M)\geq & I(Z^{i-1};Z_{i}|M) \enspace 1\leq i\leq n \nonumber \\
I(Y^{i-1};Y_{i}|M)\geq & I(Z^{i-1};Y_{i}|M) \enspace 1\leq i\leq n \nonumber
\end{align}
\end{lemma}

In this part, we establish the converse proof for only \eqref{Eq:Capacity_Region_Two_K-1_Order_msgs1_1}-\eqref{Eq:Capacity_Region_Two_K-1_Order_msgs1_4}. In particular, we show that for every sequence of $(2^{nR_{\overline{K-1}}}, 2^{nR_{\overline{K}}}, n)$ codes with $\lim_{n \rightarrow \infty} P_e^{(n)} = 0 $ the inequalities in \eqref{Eq:Capacity_Region_Two_K-1_Order_msgs1_1}-\eqref{Eq:Capacity_Region_Two_K-1_Order_msgs1_4} hold for the given classes of channels for some $p(u,x)$ for which $U \markov X \markov (Y_1,Y_2, \cdots , Y_K)$.

For the first class of channels where $Y_j  \sqsupseteq  Y_{K-1} \sqsupseteq Y_K $ for all $j\in S_p $, we show that the optimal choice of $U_i= M_{\overline{K-1}},Y_{K-1,1}^{i-1},Y_{K,i+1}^n$. While for the second class where $Y_i \succeq Y_K $ for all $i\in S_p \cup \{K-1\}$, the optimal choice of $U_i=M_{\overline{K-1}} Y_{K,1}^{i-1}$. 

The converse proof of \eqref{Eq:Capacity_Region_Two_K-1_Order_msgs1_1} is straightforward for both classes of channels and hence is omitted.

For the first class of channels where $Y_j  \sqsupseteq  Y_{K-1} \sqsupseteq Y_K $ for all $j\in S_p $, the converse proof for \eqref{Eq:Capacity_Region_Two_K-1_Order_msgs1_4} depends on Csiszar sum lemma as that in \cite{el1979capacity}, i.e., 
\begin{align}
n (R_{\overline{K-1}} + R_{\overline{K}}) & = H(M_{\overline{K-1}})+H(M_{\overline{K}}) \nonumber 
\\
& \leq I(M_{\overline{K}};Y_{K-1}^n) + I(M_{\overline{K}} ; M_{\overline{K-1}},Y_{K}^n) +  2n \epsilon_n 
 \label{Converse_proof_Eq4_k-1Msgs_1}
 \\
& = \sum_{i=1}^n  I(M_{\overline{K}};Y_{K-1,i}|Y_{K-1,i+1}^n) + I(M_{\overline{K-1}};Y_{K,i}|M_{\overline{K}},Y_{K,1}^{i-1})+ 2n \epsilon_n
 \label{Converse_proof_Eq4_k-1Msgs_2} 
\\
& \leq  \sum_{i=1}^n I(M_{\overline{K}},Y_{K-1,i+1}^n;Y_{K-1,i}) + I(M_{\overline{K-1}};Y_{K,i}|M_{\overline{K}},Y_{K,1}^{i-1}) +2 n \epsilon_n \nonumber \\
& \leq  \sum_{i=1}^n I(M_{\overline{K}},Y_{K,1}^{i-1},Y_{K-1,i+1}^n;Y_{K-1,i})- I(Y_{K,1}^{i-1};Y_{K-1,i}|M_{\overline{K}},Y_{K-1,i+1}^n)  \nonumber 
\\ 
& \quad + \sum_{i=1}^n I(M_{\overline{K-1}};Y_{K,i}|M_{\overline{K}},Y_{K,1}^{i-1},Y_{K-1,i+1}^n) + I(Y_{K-1,i+1}^n;Y_{K,i}|M_{\overline{K}},Y_{K,1}^{i-1}) 
+2 n \epsilon_n  \label{Converse_proof_Eq4_k-1Msgs_3} 
\\
&=  \sum_{i=1}^n I(M_{\overline{K}},Y_{K,1}^{i-1},Y_{K-1,i+1}^n;Y_{K-1,i})+ I(M_{\overline{K-1}};Y_{K,i}|M_{\overline{K}},Y_{K,1}^{i-1},Y_{K-1,i+1}^n)  +2n \epsilon_n \label{Converse_proof_Eq4_k-1Msgs_4}  \\
& \leq \sum_{i=1}^n I(W_{i};Y_{K-1,i})+ I(X_i;Y_{K,i}|W_{i}) + 2n \epsilon_n \label{Converse_proof_Eq4_k-1Msgs_5}\\
& \leq \sum_{i=1}^n I(W_{i};Y_{K-1,i})+ I(X_i;Y_{K-1,i}|W_{i}) + 2n \epsilon_n \label{Converse_proof_Eq4_k-1Msgs_6}\\
& = \sum_{i=1}^n I(X_i;Y_{K-1,i}) + 2n \epsilon_n \nonumber
\end{align}
where $W_{i}=M_{\overline{K}},Y_{K,1}^{i-1},Y_{K-1,i+1}^n$. The inequality \eqref{Converse_proof_Eq4_k-1Msgs_1} follows from Fano's inequality, \eqref{Converse_proof_Eq4_k-1Msgs_2} from the chain rule for mutual information and the independence between $M_{\overline{K}}$ and $M_{\overline{K}}$. Inequality \eqref{Converse_proof_Eq4_k-1Msgs_3} is true since $I(W;V) \leq I(W;V|U)+I(U;V)$, \eqref{Converse_proof_Eq4_k-1Msgs_4} is due to the Csiszar sum lemma, \eqref{Converse_proof_Eq4_k-1Msgs_5} follows from the fact that $M_{\overline{K-1}},M_{\overline{K}},Y_{K-1,1}^{i-1},Y_{K,i+1}^n \markov X_i\markov Y_{K-1,i}$ forms a Markov chain, and \eqref{Converse_proof_Eq4_k-1Msgs_6} from the condition $Y_{K-1}\sqsupseteq Y_K$.

While, for the converse of \eqref{Eq:Capacity_Region_Two_K-1_Order_msgs1_3}, we have 
\begin{align}
n (R_{\overline{K-1}} + R_{\overline{K}}) & = H(M_{\overline{K-1}})+H(M_{\overline{K}}) \nonumber 
\\
& \leq I(M_{\overline{K-1}};Y_{K}^n) + I(M_{\overline{K}} ; M_{\overline{K-1}},Y_{K-1}^n) +  2n \epsilon_n 
 \label{Converse_proof_Eq3_k-1Msgs_1}
 \\
& = \sum_{i=1}^n  I(M_{\overline{K-1}};Y_{K,i}|Y_{K,i+1}^n) + I(M_{\overline{K}};Y_{K-1,i}|M_{\overline{K-1}},Y_{K-1,1}^{i-1})+ 2n \epsilon_n
 \label{Converse_proof_Eq3_k-1Msgs_2} 
\\
&\leq \sum_{i=1}^n I(M_{\overline{K-1}},Y_{K-1,1}^{i-1},Y_{K,i+1}^n;Y_{K,i})+ I(M_{\overline{K}};Y_{K-1,i}|M_{\overline{K-1}},Y_{K-1,1}^{i-1},Y_{K,i+1}^n)  +2n \epsilon_n \label{Converse_proof_Eq3_k-1Msgs_4}  \\
& \leq  \sum_{i=1}^n I(U_{i};Y_{K,i})+ I(X_i;Y_{K-1,i}|U_{i}) + 2n \epsilon_n \label{Converse_proof_Eq3_k-1Msgs_5}
\end{align}
where $U_{i}=M_{\overline{K-1}},Y_{K-1,1}^{i-1},Y_{K,i+1}^n$. The inequality \eqref{Converse_proof_Eq3_k-1Msgs_1} follows from Fano's inequality, \eqref{Converse_proof_Eq3_k-1Msgs_2} from the chain rule for mutual information and the independence between $M_{\overline{K}}$ and $M_{\overline{K}}$. Inequality \eqref{Converse_proof_Eq3_k-1Msgs_4} follows from using the Csiszar sum lemma as in \eqref{Converse_proof_Eq4_k-1Msgs_3}-\eqref{Converse_proof_Eq4_k-1Msgs_4}.

On the other hand, for the second class where $Y_i \succeq Y_K $ for all $i\in S_p \cup \{K-1\}$, for the converse of \eqref{Eq:Capacity_Region_Two_K-1_Order_msgs1_3}, we have for any $j\in S_p \cup \{K-1\}$
\begin{align}
n (R_{\overline{K-1}} + R_{\overline{K}}) & \leq  \sum_{i=1}^n  I(M_{\overline{K-1}};Y_{K,i}|Y_{K,1}^{i-1}) + I(M_{\overline{K}};Y_{j,i}|M_{\overline{K-1}},Y_{j,1}^{i-1})+ 2n \epsilon_n
\nonumber \\
 &\leq \sum_{i=1}^n  I(M_{\overline{K-1}}, Y_{K,1}^{i-1} ;Y_{K,i}) + I(M_{\overline{K}}, Y_{j,1}^{i-1};Y_{j,i}|M_{\overline{K-1}})- I( Y_{j,1}^{i-1};Y_{j,i}|M_{\overline{K-1}}) + 2n \epsilon_n
 \label{Converse_proof_Eq33_k-1Msgs_1}\\
  &\leq \sum_{i=1}^n  I(M_{\overline{K-1}}, Y_{K,1}^{i-1} ;Y_{K,i}) + I(M_{\overline{K}}, Y_{j,1}^{i-1};Y_{j,i}|M_{\overline{K-1}})- I( Y_{K,1}^{i-1};Y_{j,i}|M_{\overline{K-1}}) + 2n \epsilon_n
 \label{Converse_proof_Eq33_k-1Msgs_2}\\
   &\leq \sum_{i=1}^n  I(M_{\overline{K-1}}, Y_{K,1}^{i-1} ;Y_{K,i}) + I(X_i;Y_{j,i}|M_{\overline{K-1}})- I( Y_{K,1}^{i-1};Y_{j,i}|M_{\overline{K-1}}) + 2n \epsilon_n
 \nonumber \\
   &\leq \sum_{i=1}^n  I(M_{\overline{K-1}}, Y_{K,1}^{i-1} ;Y_{K,i}) + I(X_i;Y_{j,i}|M_{\overline{K-1}},Y_{K,1}^{i-1})+ 2n \epsilon_n
 \label{Converse_proof_Eq33_k-1Msgs_3} \\
 &= \sum_{i=1}^n  I(U_i ;Y_{K,i}) + I(X_i;Y_{j,i}|U_i)+ 2n \epsilon_n \nonumber
\end{align} 
where \eqref{Converse_proof_Eq33_k-1Msgs_1} follows from the chain rule for mutual information, \eqref{Converse_proof_Eq33_k-1Msgs_2} from Lemma \ref{Lemma_Information_Inequality} where $Y_j\succeq Y_K$ for all $j\in S_p\cup \{K-1\}$, and \eqref{Converse_proof_Eq33_k-1Msgs_3} from the chain rule and the non-negativity of conditional mutual information,

The rest of the proof proceeds along standard lines. Define a time-sharing uniform random variable $Q$ over $[1{:}n]$ that is independent of all other involved random variables. Identify $U{=}(U_Q,Q)$ and $Y_{i}{=}Y_{iQ}$ for all $i\in [1:K]$ and take the limit as $n \rightarrow \infty$, so that $\epsilon_n \rightarrow 0$.

\end{appendices}

\bibliographystyle
{IEEEtran}
\bibliography{IEEEabrv,Cite}

\end{document}